\documentclass[smallcondensed]{svjour3}
\usepackage[utf8x]{inputenc}
\usepackage[T1]{fontenc}
\usepackage{graphicx}
\usepackage{float}
\usepackage{amsmath}
\usepackage{amsfonts}
\usepackage{bbm}
\usepackage{multirow}
\usepackage{xfrac}
\usepackage{algorithm}
\usepackage[noend]{algpseudocode}
\usepackage{url}

\def\mc {\mathcal}
\def\mk {\mathfrak}

\def\rk  {{\rm rk}}

\def\coker {{\rm coker}}

\def\im {{\rm im}}
\def\hol {{\rm hol}}
\def\Hol {{\rm Hol}}
\def\Hom {{\rm Hom}}

\def\mc {\mathcal}
\def\mk {\mathfrak}

\def\ZZ {{\mathbb Z}}
\def\CC {{\mathbb C}}
\def\RR {{\mathbb R}}

\def\QQ {{\mathbb Q}}
\def\NN {{\mathbb N}}
\def\KK {{\mathbb K}}

\newtheorem{prop}[theorem]{Proposition}

\newtheorem{rem}{Remark}[section]
\newtheorem{fact}{Fact}[section]

\begin{document}

\title{
Non-abelian Quantum Statistics on Graphs
}
\author{Tomasz Maci\k{a}\.zek$^{1,2}$ \and Adam Sawicki$^1$}
\institute {$^1$ Center for Theoretical Physics, Polish Academy of Sciences, Al. Lotnik\'ow 32/46, 02-668 Warszawa, Poland \and $^2$ School of Mathematics, University of Bristol, Bristol BS8 1TW, UK}
\date{}
\maketitle
\tableofcontents

\begin{abstract}
We show that non-abelian quantum statistics can be studied using certain topological invariants which are the homology groups of configuration spaces. In particular, we formulate a general framework for describing quantum statistics of particles constrained to move in a topological space $X$. The framework involves a study of isomorphism classes of flat complex vector bundles over the configuration space of $X$ which can be achieved by determining its homology groups. We apply this methodology for configuration spaces of graphs.  As a conclusion, we provide families of graphs which are good candidates for studying simple effective models of anyon dynamics as well as models of non-abelian anyons on networks that are used in quantum computing. These conclusions are based on our solution of the so-called universal presentation problem for homology groups of graph configuration spaces for certain families of graphs. 
\end{abstract}

\section{Introduction}\label{chap:introduction}

The main conceptual difference in the description of classical and quantum particles is the indistinguishability of the latter. Mathematically, indistinguishability of particles can be imposed already on the level of many particle configuration space. For $n$ particles that live in a topological space $X$ this is done by considering some particular tuples of length $n$ that consist of points from $X$, i.e. elements of $X^{\times n}$. Namely, these are the unordered tuples of distinct points from $X$. In other words, we consider space $C_n(X)$ defined as follows.
\[
C_n(X):=(X^{\times n}-\Delta_n)/S_n,
\]
where $\Delta_n:=\{(x_1,\dots,x_n)\in X^{\times n}:\ \exists_{i\neq j}\ x_i=x_j\}$ and $S_n$ is the permutation group that acts on $X^{\times n}$ by permuting coordinates \cite{L-M}. It is easy to see that exchanges of particles on $X$ correspond to closed loops in $C_n(X)$ \cite{Souriau,L-M,Wilczek}. Under this identification all possible quantum statistics (QS) are classified by unitary representations of the fundamental group $\pi_1(C_n(X))$. When $X=\mathbb{R}^2$ this group is known to be the braid group and when $X=\mathbb{R}^k$, where $k\geq 3$, it is the permutation group $S_n$. QS corresponding to a one-dimensional unitary representation of $\pi_1(C_n(X))$ is called abelian whereas QS  corresponding to a higher dimensional non-abelian unitary representation is called non-abelian. Quantum statistics can be also viewed as a flat connection on the configuration space $C_n(X)$ that modifies definition of the momentum operator according to minimal coupling principle. The flatness of the connection ensures that there are no classical forces associated with it and the resulting physical phenomena are purely quantum \cite{BR97,ChruJam} (cf. Aharonov-Bohm effect \cite{AB})

The first part of this paper (sections \ref{chap:introduction}-\ref{chap:bundles}) contains a meta analysis of literature concerning connections between topology of configuration spaces and the existence of different types of quantum statistics. Because the relevant literature is rather scarce, it was a nontrivial task to make such a meta analysis and we consider it an essential step in describing our results. This is because we see the need of introducing in a systematic and concise way the framework for studying quantum statistics which is designed specifically for graphs. The most challenging part in formulating such a framework is to avoid the language of differential geometry, as graph configuration spaces are not manifolds, whereas the great majority of results in the field concerns quantum statistics on manifolds. As a result, we obtain a universal framework whose many features can be utilised for a very wide class of topological spaces. The framework relies on the following mains steps: i) defining flat bundles as quotients of the trivial bundle over the universal cover of the configuration space (theorem \ref{thm:flat-bundles}), ii) defining Chern characteristic classes solely by pullbacks of the universal bundle (subsection \ref{sec:universal-bundles}) , iii) pointing out the role of the moduli space of flat $U(n)$-bundles as an algebraic variety in $U(n)^{\times r}$, $r$ being the rank of the fundamental group of the respective configuration space (subsection \ref{sec:flat}).

We particularly emphasise the important role of nontrivial flat vector bundles that can lead to spontaneously occurring non-abelian quantum statistics. This is motivated by the fact that in $\mathbb{R}^3$ fermions and bosons correspond to two non-isomorphic vector bundles that admit flat connections.  Our approach to classification of quantum statistics is connected to classification of possible quantum kinematics, i.e. defining the space of wave functions and deriving momentum operators that satisfy the canonical commutation rules. Then our classification scheme for quantum kinematics of rank $k$ on a topological space $X$ is divided into two steps
\begin{enumerate}
\item \textbf{Topological classification of wave functions}. Classify isomorphism classes of flat hermitian vector bundles of rank $k$ over $C_n(X)$. Here we also point out that in fact physically meaningful is the classification of vector bundles with respect to the so-called stable equivalence, as nonisomorphic but stable equivalent vector bundles have identical Chern numbers. An important role is played by the reduced $K$-theory and (co)homology groups of $C_n(X)$. Calculation of those groups for various graph configuration spaces is the main problem we solve in section \ref{chap:computations}.
\item \textbf{Classification of statistical properties}. If $X$ is a manifold, for each flat hermitian vector bundle, classify the flat connections. The parallel transport around loops in $C_n(X)$ determines the statistical properties. For general paracompact $X$, this point can be phrased as classification of the $U(k)$ - representations of the corresponding braid group, i.e. the fundamental group of $C_n(X)$.
\end{enumerate}
The above two-step distinction is relevant, as on a bundle which is isomorphic to the trivial bundle, one can define such a connection that the resulting representation of the braid group is trivial. However, one cannot obtain a trivial braiding for wavefunctions which are sections of a non-trivial bundle. Therefore, the very fact that the considered wavefunction lives on a non-trivial bundle excludes the possibility of having trivial braiding. This may be relevant in situations where changing the braiding properties is possible by tuning some parameters of the considered quantum system.

General methods that we describe in the first three sections of this paper, are applied to a special class of configuration spaces of particles on graphs (treated as $1$-dimensional CW complexes). Graph configuration spaces serve as simple models for studying quantum statistical phenomena in the context of abelian anyons \cite{HKR11,HKRS} or multi-particle dynamics of fermions and bosons on networks \cite{Bolte17,Bolte13a,Bolte13b}. Quantum graphs already proved to be useful in other branches of physics such as quantum chaos and scattering theory \cite{Uzy,RSS,sirko}. Of particular relevance to this paper are explicit physical models of non-abelian anyons on networks. One of the most notable directions of studies in this area is constructing models for Majorana fermions which can be braided thanks to coupling together a number of Kitaev chains \cite{kitaev,alicea}. Such models lead to new robust proposals of architectures for topological quantum computers that are based on networks. Another general way of constructing models for anyons is via an effective Chern-Simons interaction \cite{wilczek-anyons,lee}. Such models can also be adapted to the setting of graphs as self-adjoint extensions of a certain Chern-Simons hamiltonian which is defined locally on cells of the graph configuration space \cite{BE92}.  All such physical models realise some unitary representations of a graph braid group. 

In section \ref{chap:computations} we compute homology groups of graph configuration spaces to determine a coarse grained picture of isomorphism classes of flat $U(n)$ bundles over the graph configuration space. The core result of our paper concerns solving the so-called {\it universal presentation} problem of homology groups. This problem relies on constructing 
\begin{itemize}
\item a set of universal generators which generate all homology groups of graph configuration spaces
\item a set of universal relations which generate all relations between universal generators.
\end{itemize}
From the physical point of view, this is the most relevant direction of studying the homology groups of graph configuration spaces. This is because our goal is to produce universal and general statements concerning quantum statistics on graphs without the need of performing complicated calculations for every graph which would be of interest. The only way to accomplish such a general understanding is to tackle the problem of universal presentation of homology groups. We solved the above problem for i) wheel graphs (subsection \ref{sec:wheel}), ii) graph $K_{3,3}$ (subsection \ref{sec:k33}), iii) graphs $K_{2,p}$ (subsection \ref{sec:K_2p}).  The universal generators were so-called product cycles (subsection \ref{sec:prod-cycles}) and triple tori (subsection \ref{sec:K_2p}). We also solved the universal presentation problem for the second homology group of graph configuration spaces of a large class of graphs that have at most one essential vertex of degree greater than three. Solving the universal presentation problem for the above families of graphs allows us to predict the coarse-grained structure of quantum statistics independently of the number of particles. In particular, the vanishing of torsion in the homology of wheel graphs tells us that in the asymptotic limit of bundles with a sufficiently high rank, there is just one isomorphism class of flat $U(n)$ bundles.

While solving the universal presentation problem we used not only the state-of-the-art methods that have been used previously in a different context by us and other authors, but also developed new computational tools. The already existing methods were in particular i) discrete models of graph configuration spaces by Abrams and \'{S}wi\k{a}tkowski \cite{AbramsPhD,swiatkowski}, ii) the product-cycle {\it ansatz} introduced in our previous paper concerning tree graphs \cite{MS17}, iii) the vertex blow-up method introduced by Knudsen et. al. \cite{Knudsen}, iv) discrete Morse theory for graph configuration spaces introduced by Farley and Sabalka \cite{FSbraid}. However, these methods have not been used before to tackle the universal presentation problem. New computational tools we used mainly relied on i) introducing explicit techniques for calculating homology groups appearing in the homological exact sequence stemming from the vertex blow-up, ii) demonstrating a new strategy of decomposing a given graph by a sequence of vertex blow-ups and using inductive arguments to compute the homology groups, iii) further formalising and developing the product-cycle {\it ansatz} so that it can be adapted for more general graphs than just tree graphs iv) new {\it ansatz} for non-product universal generators which are homeomorphic to triple tori, v) implementing discrete Morse theory for graph configuration spaces in a computer code. A non-trivial combination of the above methods that we have applied has proved to be very effective in tackling the universal presentation problem. Nevertheless, while formulating our general framework for studying quantum statistics we already arrive at a number of new very general corollaries. This in particular concerns the structure of abelian statistics on spaces with a finitely-generated fundamental group and pointing out the role of $K$-theory in studying non-abelian statistics of a high rank.

\subsection{Quantum kinematics on smooth manifolds}
A quantisation procedure for configuration spaces, where $X$ is a smooth manifold, known under the name of {\it Borel quantisation}, has been formulated by H.D. Doebner et. al. and formalised in a series of papers \cite{Doebner97,Doebner99,tolar,DN96,DST96}. Borel quantisation on smooth manifolds can be also viewed as a version of the geometric quantisation \cite{qiang}. The main point of Borel quantisation is the fact that the possible quantum kinematics on $C_n(X)$ are in a one-to-one correspondence with conjugacy classes of unitary representations of the fundamental group of the configuration space. We denote this fact by
\[QKin_k(C_n(X))\cong \sfrac{{\rm Hom}(\pi_1(C_n(X)),U(k))}{U(k)},\]
where $QKin_k$ are the quantum kinematics of rank $k$. i.e. kinematics, where wave functions have values in $\CC^k$ and $\pi_1$ is the fundamental group. Let us next briefly describe the main ideas standing behind the Borel quantisation which will be the starting point for building an analogous theory for indistinguishable particles on graphs.

In Borel quantisation on smooth manifolds, wave functions are viewed as square-integrable sections of hermitian vector bundles. For a fixed hermitian vector bundle, the momentum operators are constructed by assigning a self-adjoint operator $\hat p_A$ acting on sections of $E$ to a vector field $A$ that is tangent to $C_n(X)$ in the way that respects the Lie algebra structure of tangent vector fields. Namely, we require the standard commutation rule for momenta, i.e.
\begin{equation}\label{commutation-pp}
[\hat p_A,\hat p_B]=\iota \hat p_{[A,B]},\ A,B\in TC_n(X).
\end{equation}
Moreover, for the position operator that acts on sections as multiplication by smooth functions
\[\hat q_f(\sigma):=f\sigma,\ f\in C^\infty(C_n(X)),\ \sigma\in Sec(E),\]
we require the remaining standard commutation rules, i.e.
\begin{equation}\label{commutation-qp}
[\hat p_A,\hat q_f]=\hat q_{Af}.
\end{equation}
It turns out that such a requirement implies the form of the momentum operator which is well-known form the minimal coupling principle, namely 
\begin{equation}\label{momentum}
\hat p_A=\iota\nabla_A+\frac{\iota}{2}{\rm div}(A),
\end{equation}
where $\nabla_A$ is a covariant derivative in the direction of $A$ that is compatible with the hermitian structure. Moreover, commutation rule (\ref{commutation-pp}) implies that $\nabla_A$ is necessarily the covariant derivative stemming from a flat connection. The component proportional to ${\rm div}(A)$ in formula (\ref{momentum}) comes from the fact that map $A\to \hat p_A$ must be valid for an arbitrary complete vector field. Usually, one considers momentum operators coming from some specific vector fields that form an orthonormal basis of local sections of $TC_n(X)$. The divergence of such a basis sections usually vanishes and formula (\ref{momentum}) describes the standard minimal coupling principle, see example \ref{example:aharonov-bohm} below. Flat hermitian connections of rank $k$ are classified by conjugacy classes of $U(k)$ representations of $\pi_1(C_n(X))$ (see \cite{Kobayashi}). Representatives of these classes can be picked by specifying the holonomy on a fixed set of loops generating the fundamental group. In order to illustrate these concepts, consider the following example of one particle restricted to move on the plane and its scalar wave functions. 

\begin{example}\label{example:aharonov-bohm}
Quantum kinematics of rank $1$ for a single particle on the plane. The momentum has two components that are given by (\ref{momentum}) for $A=\partial_x=:\partial_1$ and $A=\partial_y=:\partial_2$.
\[\hat p_1:=\hat p_{\partial_x}=\frac{1}{\iota}\partial_x-\alpha_1,\ \hat p_2:=\hat p_{\partial_y}=\frac{1}{\iota}\partial_y-\alpha_2.\]
By a straightforward calculation, one can check that commutation rule (\ref{commutation-qp}) is satisfied.
\[\forall_\Psi\ \ [\hat p_i,\hat q_f]\Psi=-\iota\Psi\partial_if=\hat q_{-\iota\partial_if}\Psi.\]
However, commutation rule (\ref{commutation-pp}) requires $[\hat p_1,\hat p_2]=0$. The commutator reads
\[\forall_\Psi\ \ [\hat p_1,\hat p_2]\Psi=\iota\Psi(\partial_1\alpha_2-\partial_2\alpha_1).\]
Therefore, in order to satisfy the momentum commutation rule, we need $\partial_1\alpha_2-\partial_2\alpha_1=0$. This is precisely the condition for the connection form $\Gamma:=\alpha_1dx+\alpha_2dy$ to have zero curvature, i.e. $d\Gamma=0$. The plane is a contractible space, hence the problem of classifying flat connections is trivial and there are no topological effects in the quantum kinematics. However, we can make the problem nontrivial by considering the situation, where a particle is moving on a plane without a point, i.e. $X=\RR^2-\{*\}$. Then, $\pi_1(X)=\ZZ$ generated by a circle around $\{*\}$ travelled clockwise. Let us denote such a loop by $\gamma$. The parallel transport of $\Psi$ around $\gamma$ gives
\[\hat T_\gamma\Psi=e^{\iota\int_\gamma \Gamma}\Psi.\]
The phase factor $e^{\iota\int_\gamma \Gamma}$ does not depend on the choice of the circle. In order to see this, choose a different circle $\gamma'$ that contains $\gamma$. Denote by $D$ the area between the circles. We have $\partial D=\gamma'-\gamma$. Hence, by the Stokes theorem
\[0=\int_D dxdy(d\Gamma)=\int_{\partial D}\Gamma=\int_{\gamma'}\Gamma-\int_\gamma \Gamma.\]
Hence, all $U(1)$ representations of $\pi_1(X)$  are the representations that assign a phase factor $e^{i\phi}$ to a chosen non-contractible loop. Physically, these representations can be realised as the Aharonov-Bohm effect and phase $\phi$ is the magnetic flux through point $*$ that is perpendicular to the plane.
\end{example}
Let us next review two scenarios that originally appeared in the paper by Leinaas and Myrheim \cite{L-M} and that led to a topological explanation of the existence of bosons, fermions and anyons \cite{wilczek-anyons}. These are the scenarios of two particles in $\RR^2$ and $\RR^3$. In both cases, the configuration space can be parametrised by the centre of mass coordinate $R$ and the relative position $r$. In terms of the positions of particles, we have
\[R=\frac{1}{2}(x_1+x_2),\ r=x_2-x_1,\ x_i\in\RR^m.\]
Then, $C_2(\RR^m)=\{(R,r):R\in \RR^m,\ r\in\RR^m-{0}\}/S_2$. Permutation of particles results with changing $r$ to $-r$, while $R$ remains unchanged, hence 
\[C_2(\RR^m)=\RR^m\times \left(\left(\RR^m-{0}\right)/\sim\right)\cong\RR^m\times{\bf RP}^{m-1}.\]
In the above formula, ${\bf RP}^{m-1}:=S^{m-1}/\sim$ is the real projective space that is constructed by identifying pairs of opposite points of the sphere. Space $\left(\RR^m-{0}\right)/\sim$ can be deformation retracted to ${\bf RP}^{m-1}$ by contracting all vectors so that they have length $1$. In the case when $m=2$, ${\bf RP}^{1}$ is topologically a circle. Equivalently, $\left(\RR^2-{0}\right)/\sim$ is a cone. Hence, we have
\[\pi_1(C_2(\RR^2))=\ZZ,\]
so similarly to Example \ref{example:aharonov-bohm}, there is a continuum of $U(1)$-representations of the fundamental group that assign an arbitrary phase factor to the wave function when transported around a non-contractible loop. Note that a loop in the configuration space corresponds to an exchange of particles (see Fig. \ref{fig:cone}).
 \begin{figure}[ht]
\centering
\includegraphics[width=0.7\textwidth]{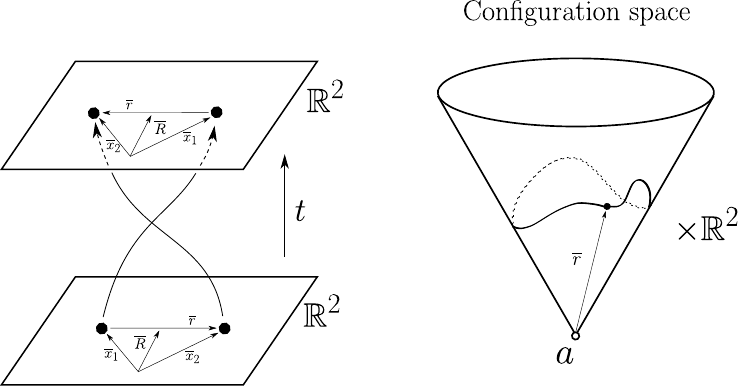}
\caption{Exchange of two particles on the plane and the resulting loop in $C_2(\RR^2)$.}
\label{fig:cone}
\end{figure}
The case of two particles moving in $\RR^3$ has an important difference when compared to the other cases analysed in this paper so far. Namely, there are two non-isomorphic hermitian vector bundles of rank $1$ that admit a flat connection. In all previous cases there was only one such vector bundle which was isomorphic to the trivial vector bundle $E_0\cong C_n(X)\times \CC$. For $m=3$, there is one more flat hermitian vector bundle which we denote by $E'$. Neglecting the $\RR^3$ - component of $C_2(\RR^3)$ which is contractible, bundles $E_0$ and $E'$ can be constructed from a trivial vector bundle on $S^2$ in the following way.
\begin{gather*}
E_0=\left(S^2\times\CC\right)/\sim,\ (r,z)\sim(-r,z)\cong {\rm RP}^2\times \CC,\\
E'=\left(S^2\times\CC\right)/\sim',\ (r,z)\sim'(-r,-z).
\end{gather*}
Intuitively, nontrivial bundle $E'$ is constructed from the trivial vector bundle on $S^2$ by twisting fibres over antipodal points. In order to determine the statistical properties corresponding to each bundle, we consider $U(1)$ representations of the fundamental group for each vector bundle. The choice of statistical properties for each vector bundle is a consequence of a general construction of flat vector bundles which we describe in more detail in section \ref{sec:flat}. The fundamental group reads
\[\pi_1(C_2(\RR^3))\cong\pi_1({\rm RP}^2)\cong\ZZ_2.\]
There are two types of loops, the contractible ones and the non-contractible ones which become contractible when composed twice (see Fig. \ref{fig:r3loops}).
 \begin{figure}[ht]
\centering
\includegraphics[width=0.5\textwidth]{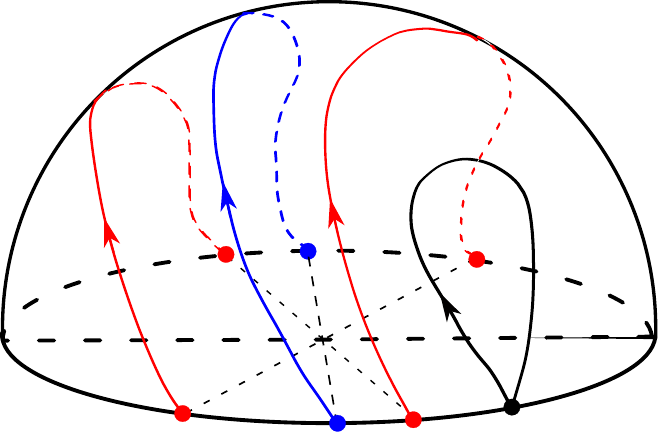}
\caption{Two types of loops in ${\rm RP}^2$ pictured as a half-sphere with the opposite points on the circumference of the base identified. Black loop and red loop are contractible, while blue loop is non-cntractible. Blue loop becomes homotopy equivalent to the red loop when crossed twice.}
\label{fig:r3loops}
\end{figure}
\noindent Bundle $E_0$ corresponds to the trivial representation of $\pi_1$, while $E'$ corresponds to the alternating representation that acts with multiplication by a phase factor $e^{i\pi}$. Consequently, the holonomy group changes the sign of the wave function from $E'$ when transported along a non-contractible loop, while the transport of a wave function from the trivial bundle results with the identity transformation. Therefore, bundle $E_0$  is called bosonic bundle, whereas bundle $E'$ is called the fermionic bundle. 

As we have seen in the above examples, there is a fundamental difference between anyons in $\RR^2$ and bosons and fermions in $\RR^3$. Anyons emerge as different flat connections on the trivial line bundle over $C_2(\RR^2)$, while fermions and bosons emerge as flat connections on non isomorphic line bundles over $C_2(\RR^3)$. As we explain in section \ref{chap:bundles}, these results generalise to arbitrary numbers of particles.

In this paper, we approach the problem of classifying complex vector bundles by computing the cohomology groups of configuration spaces over integers. Such strategy has also been used used in \cite{Doebner97} to partially classify vector bundles over configuration spaces of distinguishable particles in $\RR^m$. To this end, we combine the following methods concerning the structure of ${\rm Vect^\CC}(B)$, the set of complex vector bundles over a paracompact base space $B$.
\begin{enumerate}
\item Classification of complex vector bundles by maps $f:\ B\to Gr_k(\CC^\infty)$ and Chern classes (subsections \ref{sec:universal-bundles} and \ref{sec:flat}). 
\item Classification of vector bundles of rank $1$ by the second cohomology group (subsection \ref{sec:universal-bundles}).
\item Classification of stable equivalence classes of vector bundles using $K$-theory (subsections \ref{sec:k-theory} and \ref{sec:flat}).
\end{enumerate}
A possible source of new signatures of topology in quantum mechanics would be the existence of non-trivial vector bundles that admit a flat connection. These bundles can be detected by the Chern classes which for flat bundles belong to torsion components of $H^{2i}(B,\ZZ)$. We explain this fact and its relation with quantum statistics in section \ref{sec:flat}.

\subsection{Quantum kinematics on graphs}\label{sec:qkin}
Configuration spaces of indistinguishable particles on graphs are defined as
\[C_n(\Gamma):=(\Gamma^{\times n}-\Delta_n)/S_n,\]
where $\Delta_n=\{(x_1,\dots,x_n)\in \Gamma^{\times n}:\ \exists_{i\neq j}\ x_i=x_j\}$ and graph $\Gamma$ is regarded as a $1$-dimensional cell complex.

\begin{example}\label{example:y-conf} 
{\bf Configuration space of two particles on graph $Y$.} In $Y\times Y$ there are $9$ two-cells. Six of them are products of distinct (but not disjoint) edges of $Y$. Their intersect with $\Delta_2$ is a single point which we denote by $(2,2)$. The three remaining two-cells are of the form $e\times e$. They have the form of squares which intersect $\Delta_2$ along the diagonal. Graph $Y$ and space $C_2(Y)$ are shown on Fig. \ref{fig:conf-y}.
 \begin{figure}[H]
\centering
\includegraphics[width=0.8\textwidth]{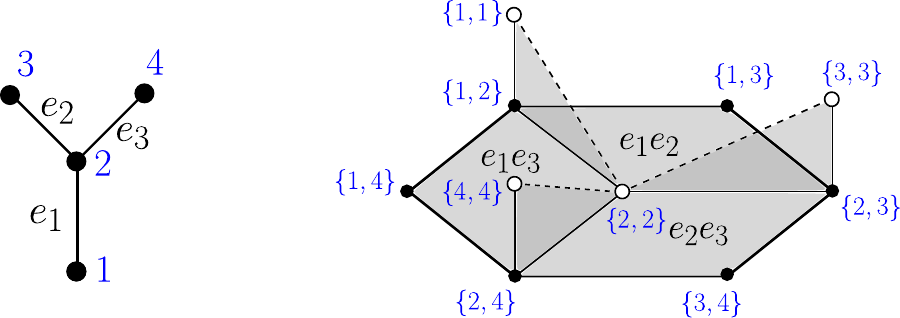}
\caption{Graph $Y$ and its two-particle configuration space. White dots and dashed lined denote the diagonal $\Delta_2$.}
\label{fig:conf-y}
\end{figure}
\end{example}
The fact that $C_n(\Gamma)$ is composed of pieces that are locally isomorphic to $\RR^n$ is the key property that allows one to define quantum kinematics as gluing the local quantum kinematics on $\RR^n$. Namely, the momentum operator on $(e_1\times e_2\times\dots\times e_n-\Delta_n)/S_n$ has $n$ components that are defined as
\[\hat p_i=-\iota \partial_i-\alpha_i,\ i=1,\dots,n.\]
We may define orthonormal coordinates and connection coefficients on each $n$-cell separately. For each $n$-cell we require that the connection $1$-form $\Gamma=\sum_{i=1}^n\alpha_i$ is closed, hence locally the connection is flat. In order to impose global flatness of the considered bundle, we require that the parallel transport does not depend on the homotopic deformations of curves that cross different pieces of $C_n(\Gamma)$. This requirement imposes conditions on the parallel transport operators along certain edges ($1$-dimensional cells) of $C_n(\Gamma)$. To see this, we need the following theorem by Abrams \cite{AbramsPhD}.
\begin{theorem}\label{lemma:dn}
Fix $n$ -- the number of particles. If $\Gamma$ has the following properties: i) each path between distinct vertices of degree not equal to $2$ passes through at least $n-1$ edges, ii) each nontrivial loop passes through at least $n+1$ edges, then $C_n(\Gamma)$ deformation retracts to a $CW$-complex $D_n(\Gamma)$ which is a subspace of $C_n(\Gamma)$ and consists of the $n$-fold products of disjoint cells of $\Gamma$.
\end{theorem}
\noindent Complex $D_n(\Gamma)$ is called Abram's discrete configuration space and we elaborate on its construction in section \ref{chap:models}. For the construction of quantum kinematics, we only need the existence of the deformation retraction. This is because under this deformation, every loop in $C_n(\Gamma)$ can be deformed to a loop in $D_n(\Gamma)\subset C_n(\Gamma)$ which has a nicer structure of a $CW$-complex. Therefore, we only need to consider the parallel transport along loops in $D_n(\Gamma)$. Furthermore, every loop in $D_n(\Gamma)$ can be deformed homotopically to a loop contained in the one-skeleton of $D_n(\Gamma)$. The problem of gluing connections between different pieces of $C_n(\Gamma)$ becomes now discretised. Namely, we require that the unitary operators that describe parallel transport along the edges of $D_n(\Gamma)$ compose to the identity operator whenever the corresponding edges form a contractible loop. In other words,
\[U_{\sigma_1}U_{\sigma_2}\dots U_{\sigma_l}=\mathbbm{1}\ {\rm if}\ \sigma_1\rightarrow\sigma_2\rightarrow\dots\rightarrow\sigma_l\ {\rm is\ a\ contractible\ loop\ in}\ D_n(\Gamma).\]
By $\sigma_1\rightarrow\sigma_2\rightarrow\dots\rightarrow\sigma_l$ we denote the path constructed by travelling along $1$-cells $\sigma_i$ in $D_n(\Gamma)$. This is a closed path whenever $\sigma_l\cap\sigma_1\neq\emptyset$.

More formally, we classify all homomorphisms $\rho\in{\rm Hom}(\pi_1(C_n(\Gamma)),U(k))$ and consider the vector bundles that are induced by the action of $\rho$ on the trivial principal $U(k)$-bundle over the universal cover of $C_n(\Gamma)$. For more details, see section \ref{chap:bundles}.

Therefore, the classification quantum kinematics of rank $k$ on $C_n(\Gamma)$ is equivalent to the classification of the $U(k)$ representations of $\pi_1(D_n(\Gamma))$. The described method of classification of quantum kinematics in the case of rank $1$ becomes equivalent to the classification of discrete gauge potentials on $C_n(\Gamma)$ that were described in \cite{HKR11}.

\begin{example}\label{example:y-kinematics} 
{\bf Quantum kinematics of rank $1$ of two particles on graph $Y$.}
The two-particle discrete configuration space of graph $Y$ consists of $6$ edges that form a circle (Fig. \ref{fig:kinematics-y}). Therefore, any non-contractible loop in $C_2(Y)$ is homotopic with $D_2(Y)$.
 \begin{figure}[H]
\centering
\includegraphics[width=\textwidth]{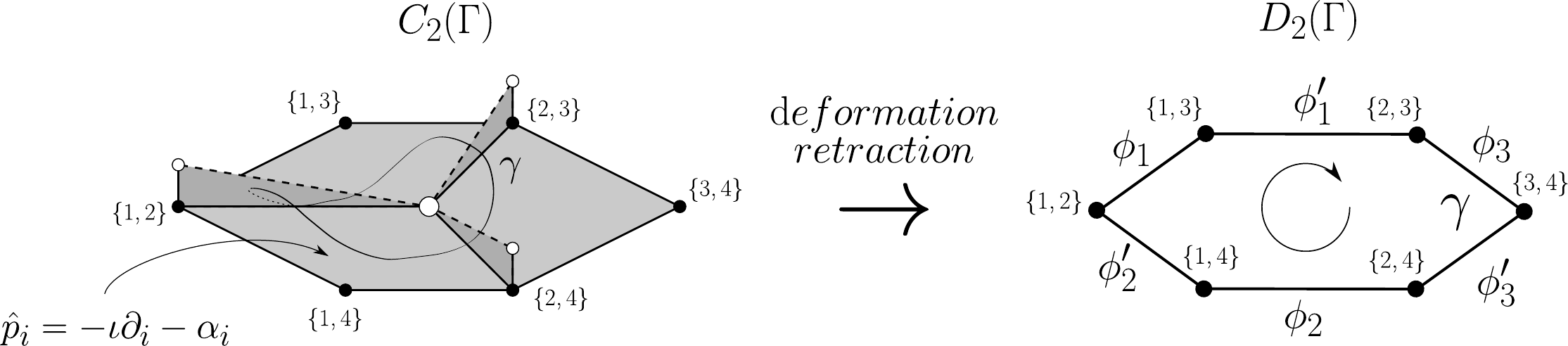}
\caption{Deformation of a loop from $C_2(Y)$ to $D_2(Y)$.}
\label{fig:kinematics-y}
\end{figure}
\noindent The classification of kinematics of rank $1$ boils down to writing down the consistency relation for $U(1)$ operators arising from the parallel transport along the edges in $D_2(Y)$. These operators are just phase factors
\[U_\sigma=e^{-i\phi_\sigma},\ \phi_\sigma=\int_\sigma \alpha_1.\]
The parallel transport of a wave function results with 
\[\hat T_\gamma\Psi=e^{-i\phi_0}\Psi,\ \phi_0=\phi_1+\phi_1'+\phi_2+\phi_2'+\phi_3+\phi_3'.\]
This is reflected in the fact that $\pi_1(C_2(Y))=\ZZ$.
\end{example}

\section{Methodology}
All topological spaces that are considered in this paper have the homotopy type of finite $CW$ complexes. This is due to the following two theorems.

\begin{theorem}
\cite{AbramsPhD,swiatkowski} The configuration space of any graph $\Gamma$ can be deformation retracted to a finite $CW$ complex which is a cube cumplex.
\end{theorem}

\begin{theorem}
\cite{Roth,FH01} The configuration space of $n$ particles in $\RR^k$ has the homotopy type of a finite $CW$-complex.
\end{theorem}
Using the structure of a $CW$-complex makes some computational problems more tractable. This is especially useful, while computing the homology groups of graph configuration spaces, because the corresponding $CW$-complexes have a simple, explicit form.

One of the central notions in the description of quantum statistics is the notion of the fundamental group. Importantly, the fundamental group of a finite $CW$ complex is finitely generated \cite{delaharpe}. This means that in all scenarios that are relevant in this paper, the fundamental group can be described by choosing a finite set of generators $a_1,\dots,a_r$ and considering all combinations of generators and their inverses, subject to certain relations 
\[\pi_1(X)=\langle a_1,a_2,\dots,a_r:\ W_1(a_1,\dots,a_r)=e,\dots,\ W_R(a_1,\dots,a_r)=e\rangle.\]
Relations $\{W_i\}$ have the form of words in $a_1,\dots,a_r$. The fundamental group of the $n$-particle configuration space of some topological space $X$ will be referred to as the $n$-strand braid group of $X$ and denoted by $Br_n(X)$. Notably, there is a wide variety of braid groups when the underlying topological space $X$ is changed. Let us next briefly review some of the flag examples.
\begin{enumerate}
\item The $n$-strand braid group of $\RR^3$ is the permutation group, $Br_n(\RR^3)=S_n$.
\item The $n$-strand braid group of $\RR^2$ is often simply called {\it braid group} and denoted by $Br_n$. It has $n-1$ generators denoted by $\sigma_1,\dots,\sigma_{n-1}$. One can illustrate the generators by arranging particles on a line. In such a setting, $\sigma_i$ corresponds to exchanging particles $i$ and $i+1$ in a clockwise manner. By composing such exchanges, one arrives at the following presentation of $Br_n(\RR^2)$
\begin{gather*}
Br_n(\RR^2)=\langle\sigma_1,\dots,\sigma_{n-1}:\ \sigma_i\sigma_{i+1}\sigma_i=\sigma_{i+1}\sigma_i\sigma_{i+1}\ {\rm for\ }i=1,\dots,n-2,\\ 
\sigma_i\sigma_j=\sigma_j\sigma_i{\rm\ for\ }|i-j|\geq 2\rangle.
\end{gather*}
\item The $n$-strand braid group of a sphere $S^2$ has the same set of generators and relations as $Br_n(\RR^2)$, but with one additional relation: $\sigma_1\sigma_2\dots\sigma_{n-1}\sigma_{n-1}\dots\sigma_2\sigma_1=e$.
\item The $n$-strand braid group of a torus $T^2$. Group $Br_n(T^2)$ is generated by i) generators $\sigma_1,\dots,\sigma_{n-1}$ where the relations are the same as in the case of $\RR^2$ and ii) generators $\tau_i,\ \rho_i$, $i=1,\dots,n$ that transport particle $i$ around one of the two fundamental loops on $T^2$ respectively. As the full set of relations defining $Br_n(T^2)$ is quite long, we refer the reader to \cite{torus-statistics}.
\item Fundamental groups of $n$-particle configuration spaces of graphs, also called graph braid groups \cite{FSbraid,FSpresentations}. The study of integral homology of graph braid groups is a central point of this paper. 
\end{enumerate}
Graph configuration spaces and $C_n(\RR^2)$ are Eilenberg-MacLane spaces of type $K(G,1)$, i.e. the fundamental group is their only non-trivial homotopy group. Such spaces are also called aspherical. In the following example we aim to provide some intuitive understanding of complications and difficulties that are met while dealing with graph braid groups.
\begin{example}[Braid groups for two or three particles on $\Theta$-graphs]
Consider graph $\Theta$ that consists of two vertices and three parallel edges that connect the vertices. As we show schematically in Fig. \ref{theta-braid}, group $Br_2(\Gamma_\Theta)$ is a free group that has three generators, $Br_2(\Gamma_\Theta)=\langle \alpha_D,\alpha_U,\gamma_L\rangle$. Generators $\alpha_U$ and $\alpha_D$ correspond to a single particle travelling around a simple cycle in $\Gamma_\Theta$ while generator $\gamma_L$ denotes a pair of particles exchanging on the left junction. Clearly, it is possible to have an analogous exchange on the right junction, $\gamma_R$. Such an exchange can be expressed by the above generators as 
\begin{equation}\label{theta-rel}
\gamma_R\sim \alpha_D\alpha_U\gamma_L^{-1}\left(\alpha_U\alpha_D\right)^{-1}
\end{equation}
 \begin{figure}[H]
\centering
\includegraphics[width=0.6\textwidth]{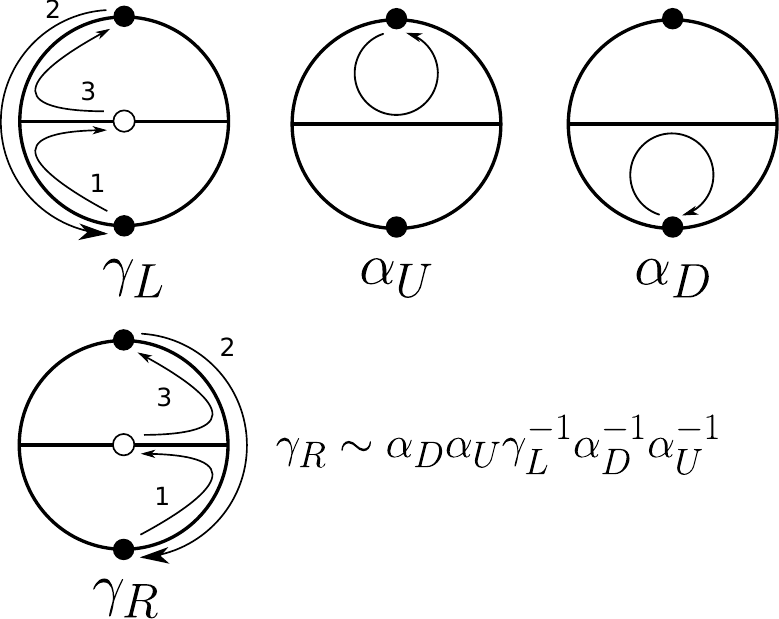}
\caption{Group $Br_2(\Gamma_\Theta)$ is a free group with three generators: $\alpha_U,\alpha_D,\gamma_L$.}
\label{theta-braid}
\end{figure}
A physical model for a $U(2)$ representation of $Br_2(\Gamma_\Theta)$ can be constructed using general theory of exchanging Majorana fermions on networks of quantum wires presented in \cite{alicea}. Here we only briefly sketch the main ideas of this construction. The role of particles is played by two Majorana fermions placed on the spots of black dots from Fig. \ref{theta-braid}. The two fermions are at the endpoints of the so-called topological region in a network of superconducting quantum wires. Majorana fermions are braided by adiabatically changing physical parameters of the quantum wire. 

An example of a graph whose braid group has a more complicated structure is graph $\Theta_4$ which has four parallel edges that connect two vertices. Space $C_3(\Gamma_{\Theta_4})$ has the homotopy type of a closed two-dimensional surface of genus $3$ \cite{KoPark}. Hence, the corresponding graph braid group has six generators subject to one relation 
\[Br_3(\Gamma_{\Theta_4})=\Big{\langle}\alpha_1,\gamma_1,\dots,\alpha_3,\gamma_3:\ \prod_{i=1}^3\alpha_i\gamma_i\alpha_i^{-1}\gamma_i^{-1}=e\Big{\rangle}.\]
\end{example}

In this paper we focus on calculating cellular homology of graph configuration spaces. It is done by assigning to $C_n(\Gamma)$ a finite chain complex $\mk{C}$ in the way which is described in section \ref{chap:models}. Homology groups of complex $\mk{C}$ are finitely generated abelian groups, i.e. have the following form
\[H_d(\mk{C},\ZZ)=\ZZ^K\oplus\bigoplus_{i=1}^L\ZZ_{p_i},\]
where $K, L\in \NN$, and $\{p_i\}_{i=1}^L$ are natural numbers such that $p_i$ divides $p_{i+1}$ for all $i$. 
\noindent Number $K$ is called the rank of $H_d(\mk{C},\ZZ)$, and is equal to the $d$th Betti number of complex $X$.
\[K={\rm rk} (H_d(\mk{C},\ZZ))=\beta_d(X).\]
The cyclic part of $H_d(\mk{C},\ZZ)$ is called the torsion part and denoted by $T(H_d(\mk{C},\ZZ))$ or $T_d(\mk{C},\ZZ)$. An important theorem that we will often use reads \cite{HatcherAT}:

\begin{theorem}
If $X$ has the homotopy type of a finite $CW$ complex, then ranks of $H^k(X,\ZZ)$ and $H_k(X,\ZZ)$ are equal and the torsion of $H^k(X,\ZZ)$ is equal to the torsion of $H_{k-1}(X,\ZZ)$.
\end{theorem}

\section{Vector bundles and their classification}\label{chap:bundles}

The main motivation for studying (co)homology groups of configuration spaces comes from the fact that they give information about the isomorphism classes of vector bundles over configuration spaces. In the following section, we review the main strategies of classifying vector bundles and make the role of homology groups more precise. Throughout, we do not assume that the configuration space is a differentiable manifold, as the configuration spaces of graphs are not differentiable manifolds. We only assume that $C_n(X)$ has the homotopy type of a finite $CW$-complex. This means that $C_n(X)$ can be deformation retracted to a finite $CW$-complex. As we explain in section \ref{chap:models}, configuration spaces of graphs are such spaces. The lack of differentiable structure means that the flat vector bundles have to be defined without referring the notion of a connection and all the methods that are used have to be purely algebraic. We provide such an algebraic definition of flat bundles in section \ref{sec:flat}.

In this paper, we consider only complex vector bundles $\pi:\  E\to B$, where $E$ is a total space and $B$ is the base. Two vector bundles are isomorphic iff there exists a homeomorphism between their total spaces which preserves the fibres. If two vector bundles belong to different isomorphism classes, there is no continuous function which transforms the total spaces to each other, while preserving the fibres. Hence, the wave functions stemming from sections of such bundles must describe particles with different topological properties. The classification of vector bundles is the task of classifying isomorphism classes of vector bundles. The set of isomorphism classes of vector bundles of rank $k$ will be denoted by $\mc{E}_k^\KK(B)$. 

%

Before we proceed to the specific methods of classification of vector bundles, we introduce an equivalent way of describing vector bundles which involves {\it principal bundles} (principal $G$-bundles). A principal $G$-bundle $\xi:\ P\to B$ is a generalisation of the concept of vector bundle, where the total space is equipped with a free action of group $G$ \footnote{The action of $G$ on $P$ can be left or right. In this work we pick up the convention of right action. This means that $g(h(p))=(gh)(p)$ for $g,h\in G$, $p\in P$. Group action is free iff for all $g\in G$ and $p\in P$, $gp\neq p$.} and the base space has the structure of the orbit space $B\cong P/G$. Fibre $\pi^{-1}(p)$ is isomorphic to $G$ is the sense that map $\pi:\ P\to B$ is $G$-invariant, i.e. $\pi(ge)=\pi(e)$. Moreover, all relevant morphisms are required to be $G$-equivariant. The set of isomorphism classes of principal $G$-bundles over base space $B$ will be denoted by 
$\mc{P}_G(B)$.

While interpreting sections of vector bundles as wave functions, we need the notion of a hermitian product on $E$. This means that we consider hermitian vector bundles, i.e. bundles with hermitian product $\langle\cdotp,\cdotp\rangle_p$ on fibres $\pi^{-1}(p),\ p\in B$ that depends on the base point and varies between the fibres in a continuous way.  Choosing sets of unitary frames, we obtain a correspondence between hermitian vector bundles and principal $U(k)$-bundles. If the base space is paracompact, any complex vector bundle can be given a hermitian metric \cite{Milnor}. Using the fact that principal $U(k)$-bundles corresponding to different choices of the hermitian structure are isomorphic \cite{Milnor}, we have the following bijection
\[\mc{P}_{U(k)}(B)\cong \mc{E}_k^\CC(B).\]
From now on, we will focus only on the problem of classification of principal $U(k)$-bundles.

\subsection{Universal bundles and Chern classes}\label{sec:universal-bundles}
Recall that all vector bundles of rank $k$ over a paracompact topological space can be obtained from a vector bundle which is universal for all base spaces. This is done in the following way. Any continuous map $f:\ B'\to B$ between base spaces induces a pullback map of vector bundles over $B$ to vector bundles over $B'$. The pullback bundle is defined as $f^*E=\{(p,e)\in B'\times E:\ f(p)=\pi(e)\}$. Similarly, one defines the pullback of principal $G$-bundles. For a fixed principal $G$-bundle $\xi:\ P\to B$, the pullback map induces a map from $[A,B]$, i.e. from the space of homotopy classes of continuous maps from $A$ to $B$, to the set of isomorphism classes of principal $G$-bundles over $A$ by $f\mapsto f^* \xi$. A space $B$ for which such a map is bijective regardless the choice of space $A$, is called a {\it classifying space} for $G$ and is denoted by $BG$. If this is the case, bundle $\xi$ is called a {\it universal bundle}. For principal $U(k)$-bundles, the classifying space is the infinite Grassmannian \cite{Milnor}
\[BU(k)=Gr_k(\CC^\infty),\]
and the corresponding universal bundle is denoted by $\gamma^k_\CC$. Therefore, any principal $U(k)$-bundle over a paracompact Hausdorff space $B$ can be written as $f^*(\gamma^k_\CC)$ for $f:\ B\to Gr_k(\CC^\infty)$. The isomorphism class of $f^*(\gamma^k_\CC)$ is determined uniquely by the homotopy class of $f$ and vice versa.  However, the classification of such homotopy classes of maps, as well as differentiating between different classes are difficult tasks. A more computable criterion for comparing isomorphism classes of vector bundles are invariants called  Chern characteristic classes. Let us next briefly introduce this notion. A characteristic class is a map that assigns to each principal $G$-bundle $\xi:\ P\to B$ an element of the cohomology ring of $B$ with some coefficients. Characteristic classes are invariant under isomorphisms of principal bundles, and those that  describe principal $U(k)$-bundles have values in $H^*(B,\ZZ)$. Such characteristic classes are called integral Chern classes. They are evaluated as follows. Let $a\in H^q(BU(k),\ZZ)$. We assign to this element a characteristic class $c_a$ which is defined defined by its values on an arbitrary principal bundle $\xi:\ P\to B$. By the classification theorem, we have $\xi=f^*_\xi(\gamma_\CC^k)$ for some continuous map $f_\xi:\ B\to BU(k)$. Hence, $c_a$ is evaluated as $c_a(\xi):=f^*_\xi(a)$, where $f^*_\xi:\ H^q(BU(k),\ZZ)\to H^q(B,\ZZ)$ is the pullback of cohomology rings via map $f_\xi$. Map $f^*_\xi$ is often called the {\it characteristic homomorphism}. It turns out that the only nonzero Chern classes are of even degree.
  
Chern classes are especially useful in classifying line bundles, as the set of homotopy classes of maps $[B,BU(1)]$ is in a bijective correspondence with $H^2(B,\ZZ)$. Hence, we arrive at the first direct application of the knowledge of cohomology ring of space $B$, namely
\[\mc{E}_1^\CC(B)\cong H^2(B,\ZZ).\]
More applications of Chern classes and cohomology ring $H^*(B,\ZZ)$ follow in the remaining parts of this section. In particular, they appear in $K$-theory and while studying characteristic classes of flat vector bundles.

\subsection{Reduced $K$-theory}\label{sec:k-theory}

We start with recalling the definition  of  {\it stable} equivalence of vector bundles. 
\begin{definition}\label{defin:stable-equivalence}
Vector bundles $\xi$ and $\xi'$ are stably equivalent $\xi\sim_{s}\xi'$ iff
\[\exists_{k_1,k_2\in\ZZ}\ [\xi\oplus\tau_{k_1}]=[\xi'\oplus\tau_{k_2}].\]
\end{definition}
\noindent The set of stable equivalence classes of vector bundles over a compact Hausdorff space has the structure of an abelian group which is called the reduced Grothendieck group $\tilde K(B)$. If the base space has the homotopy type of a finite $CW$-complex, group $\tilde K(B)$ fully describes isomorphism classes of vector bundles that have a sufficiently high rank \cite{Husemoller}. This concerns vector bundles, whose rank is in the {\it stable range}, i.e. is greater than or equal to
\[k_s:=\left\lceil\frac{1}{2}\dim B\right\rceil,\]
where $\lceil x \rceil$ denotes the smallest integer that is greater than or equal to $x$. The set of stable equivalence classes of $Vect^\CC(B)$ is equal to $\mc{E}_{k_s}^\CC(B)$. Moreover, $\mc{E}_k^\CC(B)$ is the same for all $k\geq k_s$ and equal to $\mc{E}_{k_s}^\CC(B)$. Therefore, 
\[\mc{E}_k^\CC(B)\cong \tilde K(B)\ {\rm for}\ k\geq k_s.\]

The relation between reduced $K$-theory and cohomology is phrased via the {\it Chern character} which induces isomorphism from $\tilde{K}(B)$ to $H^*(B,\QQ)$ when $B$ has the homotopy type of a finite $CW$-complex.
%

As a consequence, the classification of vector bundles in the stable range asserts that
\[\mc{E}^\CC_k\cong \bigoplus_{i=1}H^{2i}(B,\QQ),{\rm\ for}\ k\geq \frac{1}{2}\dim B,\]
on condition that the even integral cohomology groups of $B$ are torsion-free. In the case when there is non-trivial torsion in $H^*(B,\ZZ)$, torsion of $\tilde K(B)$ is determined by the Atiyah-Hirzebruch spectral sequence \cite{AH}. However, the correspondence between torsion of even cohomology and $\tilde K(B)$ is not an isomorphism. In particular, torsion in $\tilde K(B)$ can vanish, despite the existence of nonzero torsion in $H^{2i}(B,\ZZ)$. Finally, we note that stable equivalence of vector bundles is physically important in situations when one has access only to Chern classes or other topological invariants stemming from Chern classes, e.g. the Chern numbers. This is because Chern classes of stably equivalent vector bundles are equal.

\subsection{Flat bundles and quantum statistics}\label{sec:flat}
In this section, we describe the structure of the set of flat principal $G$-bundles over base space $B$. More precisely, we consider the set of pairs $(\xi,\mc{A})$, where $\xi$ is a principal $G$-bundle, and $\mc{A}$ is a connection $1$-form on $\xi$. We divide the set of such pairs into equivalence classes $[(\xi,\mc{A})]$ that consist of vector bundles isomorphic to $\xi$ and the set of flat connections that are congruent to $\mc{A}$ under the action of the gauge group. The quotient space with respect to this equivalence relation is called the {\it moduli space of flat connections} and is denoted by $\mc{M}(B,G)$. The culminating point of this section is to introduce the fundamental relation which says that $\mc{M}(B,G)$ is in a bijective correspondence with the set of conjugacy classes of homomorphisms of the fundamental group of $B$.
\begin{equation}\label{eq:flat-hom}
\mc{M}(B,G)\cong {\rm Hom}(\pi_1(B),G)/G.
\end{equation}
We use this relation to explain some key properties of quantum statistics that were sketched in the introduction of this paper.

Recall the description of the moduli space of flat connections in the case when $B$ is a smooth manifold. Having fixed a principal connection $H$ on $P$, we consider parallel transport of elements of $P$ around loops in $B$. Parallel transport around loop $\gamma\subset B$ is a morphism of fibres $\Gamma_\gamma:\ \pi^{-1}(b)\to\pi^{-1}(b)$ which assigns the end point of the horizontal lift of $\gamma$ (denote it by $\tilde\gamma$) to its initial point
\[\Gamma_\gamma:\ \tilde\gamma(0)\mapsto \tilde\gamma(1).\]
Because fibres are homogeneous spaces for the action of $G$, for every choice of the initial point $p=\tilde\gamma(0)$ there is a unique group element $g\in G$ such that $\tilde\gamma(1)=gp$. We denote this element by $\hol_p(H,\gamma)$ and call the {\it holonomy} of connection $H$ around loop $\gamma$ at point $p$. Moreover, by the $G$-equivariance of the connection, we get that 
\[\Gamma_\gamma(gp)=g\Gamma_\gamma(p),\ p\in P.\]
This means that $\hol_{gp}(H,\gamma)=g^{-1}\hol_p(H,\gamma)g$. If connection $H$ is flat, the parallel transport depends only on the topology of the base space \cite{KobNomizu}, i.e. i) $\Gamma_\gamma$ depends only on the homotopy class of $\gamma$, ii) parallel transport around a contractible loop is trivial,
iii) parallel transport around two loops that have the same base point is the composition of parallel transports along the two loops $\Gamma_{\gamma_1\circ\gamma_2}=\Gamma_{\gamma_1}\circ\Gamma_{\gamma_2}$.
These facts show that if $H$ is flat, map $\pi_1(B)\ni[\gamma]\mapsto \hol_p(H,\gamma)\in G$ is a homomorphism of groups. Because holonomies at different points from the same fibre differ only by conjugation in $G$, it is not necessary to specify the choice of the initial point. Instead, we consider map 
\[\mc{S}_H:\ \pi_1(B)\ni[\gamma]\mapsto \Hol(H,\gamma)\in Conj(G),\]
where $\Hol(H,\gamma)=\{\hol_p(H,\gamma):\ p\in \pi^{-1}(\gamma(0))\}$ is a conjugacy class of group $G$. There is one more symmetry of this map that we have not discussed so far, namely the gauge symmetry. A gauge transformation is a map $f:\ P\to G$ which is $G$-equivariant, i.e. $f(gp)=g^{-1}f(p)g$. A gauge transformation induces an automorphism of $P$ which acts as $p\to f(p) p$. Consequently, transformation $f$ induces a pullback of connection forms. It can be shown that map $\mc{S}_H$ is gauge invariant \cite{KobNomizu}, i.e. depends only on the gauge equivalence class of connection $H$.

An important conclusion regarding flat bundles on spaces that do not have a differential structure comes from the second part of correspondence (\ref{eq:flat-hom}). This is the reconstruction of a flat principal bundle from a given homomorphism $\Hom(\pi_1(B),G)$. It turns out that any flat bundle over $B$ can be realised as a particular quotient bundle of the trivial bundle over the {\it universal cover} of $B$. In order to formulate the correspondence, we first introduce the notion of a covering space and a universal cover \footnote{Universal covers of graph configuration spaces have a particularly nice structure, as they have the homotopy type of a $CAT(0)$ cube complex \cite{AbramsPhD} which is contractible.}. The following theorem is also a definition of a flat principal bundle for spaces that are not differential manifolds.
\begin{theorem}\label{thm:flat-bundles}
Any flat principal $G$-bundle $P\to B$ can be constructed as the following quotient bundle of the trivial bundle over the universal cover of $B$.
\[P=(\tilde B\times G)/\pi_1(B).\]
In the above formula, group $\pi_1(B)$ acts on $\tilde B$ via deck transformations. Action on $G$ is defined by picking a homomorphism $\rho: \pi_1(B)\to G$. Then the action reads $ag:=\rho(a)g$ for $a\in \pi_1(B)$, $g\in G$.
\end{theorem}

Summing up, in order to describe the moduli space of flat $G$-bundles, one has to classify conjugacy classes of homomorphisms $\pi_1(B)\to G$. All spaces that are considered in this paper have finitely generated fundamental group. This fact makes the classification procedure easier. Namely, one can fix a set of generators $a_1,\dots,a_r$ of $\pi_1(B)$ and represent them as group elements $g_1,\dots,g_r$. Matrices $g_1,\dots,g_r$ realise $\pi_1(B)$ in $G$ in a homomorphic way iff they satisfy the relations between the generators of $\pi_1(B)$. This way, the moduli space of flat connections can be given the structure of an algebraic variety. In other words, we consider map
\[\mc{Q}:\ G^{\times r}\to G^{\times n_R},\]
which returns the values of words describing the relations between generators of $\pi_1(B)$. Then, 
\[\mc{M}(B,G)=\mc{Q}^{-1}(e,\dots,e)/G.\]
We view $\mc{Q}^{-1}(e,\dots,e)$ as the zero locus of a set of multivariate polynomials. In general, such a zero locus has many path connected components. This reflects the topological structure of $\mc{M}(B,G)$. Namely, one can decompose the moduli space of flat connections into a number of disjoint components that are enumerated by the isomorphism classes of bundles
\[\mc{M}(B,G)=\bigsqcup_{[\xi]\in\mc{P}_G(B)} \mc{M}_{[\xi]}(B,G).\]
$\mc{M}_{[\xi]}(B,G)$ is the space of flat connections on principal bundles from the isomorphism class $[\xi]$ modulo the gauge group. The following fact gives a necessary condition for two flat structures to be non-isomorphic.
\begin{fact}\label{thm:flat-topology}
Two points in $\mc{M}(B,G)$ that correspond to two non-isomorphic flat bundles, belong to different path-connected components of $\mc{M}(B,G)$.
\end{fact}
\noindent Equivalently, if two flat structures, i.e. points in $\mc{M}(B,G)$, belong to the same path-connected component of $\mc{M}(B,G)$, then the corresponding vector bundles are isomorphic. A path connecting the two points in $\mc{M}(B,G)$ gives a homotopy between the corresponding flat structures.

\begin{example}{\bf -- The moduli space of flat $U(1)$ bundles over spaces with finitely generated fundamental group.} 
As conjugation in $U(1)$ is trivial, we have 
\[\mc{M}(B,U(1))\cong \Hom(\pi_1(B),U(1)).\]
Moreover, $\Hom(\pi_1(B),U(1))$ is the same as the space of homomorphisms from the abelianization of $\pi_1(B)$ to $U(1)$. A standard result from algebraic topology says that
\[\sfrac{\pi_1(B)}{[\pi_1(B),\pi_1(B)]}\cong H_1(B,\ZZ),\]
where $[\cdot, \cdot]$ is the group commutator. $H_1(B,\ZZ)$ as any finitely generated abelian group decomposes as the sum of a free component and a cyclic (torsion) part
\[H_1(B,\ZZ)=\ZZ^{p}\oplus\bigoplus_{i=1}^q\ZZ_{p_i}.\]
Therefore, we can generate $H_1(B,\ZZ)$ as
\[H_1(B,\ZZ)=\langle a_1,\dots,a_p,b_1,\dots,b_q:\ b_i^{p_i}=e\rangle.\]
We represent $a_i$ as $e^{\iota\phi_i},\ \phi_i\in[0,2\pi[$ and the cyclic generators as roots of unity $e^{\iota 2k_i\pi/p_i}$, where $k_i=0,1,2,\dots,p_i-1$. This way, we get $\prod_{i=1}^qp_i$ connected components in the space of homomorphisms $\Hom(H_1(B,\ZZ),U(1))$ that are enumerated by different choices of numbers $k_i$. Each connected component is homeomorphic to a $p$-torus, whose points correspond to phases $\phi_i$. In fact, the connected components are in a one-to-one correspondence with isomorphism classes of flat bundles. To see this, recall the fact that set of $U(1)$-bundles has the structure of a group which is isomorphic to $H^2(B,\ZZ)$. Moreover, as we explain in Remark \ref{rem:flat-torsion}, Chern classes of flat bundles are torsion. This means that flat $U(1)$-bundles form a subgroup of the group of all $U(1)$-bundles which is isomorphic to the torsion of $H^2(B,\ZZ)$. By the universal coefficient theorem \cite{spanier}, torsion of $H^2(B,\ZZ)$ is the same as torsion of $H_1(B,\ZZ)$. Note that there is exactly the same number of connected components in $\Hom(H_1(B,\ZZ),U(1))$ as the number of group elements in the torsion component of $H_1(B,\ZZ)$. In this case, fact \ref{thm:flat-topology} implies that each connected component represents one isomorphism class of flat bundles.

 Recall that for particles in $\RR^2$ and $\RR^3$, we had
\[H_1(C_n(\RR^2),\ZZ)=\ZZ,\ H_1(C_n(\RR^3),\ZZ)=\ZZ_2.\]
Hence, the moduli spaces read (see also Fig. \ref{fig:moduli-real})
\begin{gather*}
\mc{M}(C_n(\RR^2),U(1))\cong\Hom(\ZZ,U(1))\cong S^1,\\
\mc{M}(C_n(\RR^3),U(1))\cong\Hom(\ZZ_2,U(1))\cong \{*,*'\}\subset T^2.
\end{gather*}
 \begin{figure}[ht]
\centering
\includegraphics[width=0.7\textwidth]{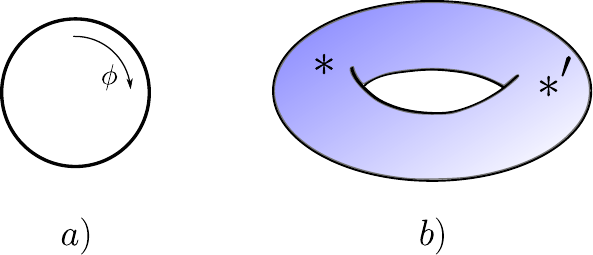}
\caption{The moduli space of flat $U(1)$ bundles a) for $n$ particles on a plane, b) $n$ particles in $\RR^3$. Homomorphisms from $\ZZ$ to $U(1)$ are parametrised by points from $S^1$ via the map $\phi\mapsto e^{\iota\phi}$. The corresponding homomorphism reads $n\mapsto e^{\iota n\phi}$. There is only one path-connected component in $\Hom(\ZZ,U(1))$ which reflects the fact that there is only one flat $U(1)$ bundle over $C_n(\RR^2)$ (the trivial one) and points form the circle parametrise different flat connections. For particles in $\RR^3$, there are two homomorphisms of $\ZZ_2=\{1,-1\}$ - the trivial one and $1\mapsto e^{2\pi\iota}$, $-1\mapsto e^{\iota\pi}$. They correspond to two isolated points on the torus $T^2=U(1)\times U(1)$. The trivial homomorphism corresponds to the bosonic bundle, while the other homomorphism corresponds to the fermionic bundle. The fundamental difference between these two types of quantum statistics is that anyons arise as different flat connections on the trivial bundle, whereas bosons and fermions arise as canonical flat connections on two non-isomorphic flat bundles.}
\label{fig:moduli-real}
\end{figure} 
\end{example}

\paragraph*{Characteristic classes of flat bundles}\ From this point, we can move away from considering connections and use the wider definition of flat $G$-bundles which makes sense for bundles over spaces that have a universal covering space. As stated in theorem \ref{thm:flat-bundles}, such flat bundles have the form
\[P=(\tilde B\times G)/\pi_1(B),\]
where we implicitly use a group homomorphism $\rho:\ \pi_1(B)\to G$ in the definition of the quotient. For such flat $U(n)$-bundles over connected $CW$-complexes we have the following general result about the triviality of rational Chern classes \cite{flat-chern}.
\begin{theorem}\label{thm:flat-chern}
Let $G$ be a compact Lie group, $B$ a connected $CW$-complex and $\xi: P\to B$ a flat $G$-bundle over $B$. Then, the characteristic homomorphism 
\[f_\xi^*:\ H^*(BG,\QQ)\to H^*(B,\QQ)\]
is trivial.
\end{theorem}
\begin{rem}\label{rem:flat-torsion}
Theorem \ref{thm:flat-chern} in particular means that if $B$ is a finite $CW$-complex, then by the universal coefficient theorem for cohomology (see e.g. \cite{spanier}), the image of the characteristic map $f_\xi^*:\ H^*(BG,\ZZ)\to H^*(B,\ZZ)$ consists only of torsion elements of $H^*(B,\ZZ)$.
\end{rem}
Specifying the above results for $U(n)$-bundles, we get that the lack of nontrivial torsion in $H^{2i}(B,\ZZ)$ has the following implications for the stable equivalence classes of flat vector bundles.
\begin{prop}\label{prop:torsion-free-flat}
Let $B$ be a finite $CW$ complex. If the integral homology groups of $B$ are torsion-free, then every flat complex vector bundle over $B$ is stably equivalent to a trivial bundle.
\end{prop}
\begin{proof}
If the integral cohomology of $B$ is torsion-free, then by the Chern character we get that the reduced Grothendieck group is isomorphic to the direct sum of even cohomology of $B$. Thus, if all Chern classes of a given bundle vanish, this means that this bundle represents the trivial element of the reduced Grothendieck group, i.e. is stably equivalent to a trivial bundle. 
\end{proof}
\noindent Interestingly, in the following standard examples of configuration spaces, there is torsion in cohomology.
\begin{enumerate}
\item Configuration space of $n$ particles on a plane. Space $C_n(\RR^2)$ is aspherical, i.e. is an Eilenberg-Maclane space of type $K(\pi_1,1)$, where the fundamental group is the braid group on $n$ strands $Br_n$. Cohomology ring $H^*(C_n(\RR^2),\ZZ)=H^*(Br_n,\ZZ)$ is known \cite{Arnold,Vainshtein}. Its key properties are i) {\bf finiteness} -- $H^{i}(Br_n,\ZZ)$ are cyclic groups, except $H^{0}(Br_n,\ZZ)=H^{1}(Br_n,\ZZ)=\ZZ$, ii) {\bf repetition} -- $H^{i}(Br_{2n+1},\ZZ)=H^{i}(Br_{2n},\ZZ)$, iii) {\bf stability} -- $H^{i}(Br_{n},\ZZ)=H^{i}(Br_{2i-2})$ for $n\geq 2i-2$. Description of nontrivial flat $U(n)$ bundles over $C_n(\RR^2)$ for $n> 2$ is an open problem.

\item Configuration space of $n$ particles in $\RR^3$. Much less is known about $H^*(C_n(\RR^3))$. Some computational techniques are presented in \cite{cohen-iterated-loop-spaces,bloore-homology}, but little explicit results are given. Ring $H^*(C_3(\RR^3)$ is equal to $\ZZ,0,\ZZ_2,0,\ZZ_3$ \cite{bloore-bundles} and $H^q(C_3(\RR^3))=0$ for $q>4$. However, it has been shown that there are no nontrivial flat $SU(n)$ bundles over $C_3(\RR^3)$. 

\item Configuration space of $n$ particles on a graph (a $1$-dimensional $CW$-complex $\Gamma$). Spaces $C_n(\Gamma)$ are Eilenberg-Maclane spaces of type $K(\pi_1,1)$. The calculation of their homology groups is a subject of this paper. Group $H_1(C_n(\Gamma),\ZZ)$ is known \cite{HKRS,KoPark} for an arbitrary graph. We review the structure of $H_1(C_n(\Gamma))$ in section \ref{sec:first-homology}. By the universal coefficient theorem, the torsion of $H^2(C_n(\Gamma))$ is equal to the torsion of $H_1(C_n(\Gamma))$ which is known to be equal to a number of copies of $\ZZ_2$, depending on the structure of $\Gamma$. We interpret this result as the existence of different bosonic or fermionic statistics in different parts of $\Gamma$. The existence of torsion in higher (co)homology groups of $C_n(\Gamma)$ which is different than $\ZZ_2$, is an open problem. In this paper, we compute homology groups for certain canonical families of graphs. However, the computed homology groups are either torsion-free, or have $\ZZ_2$-torsion.
\end{enumerate}
As we have seen while studying the example of anyons, the parametrisation of different path-connected components of the moduli space of flat bundles corresponds physically to changing some fields. On the other hand, while studying the example of particles in $\RR^3$, we learned that on each path-connected component of $\mc{M}(B,G)$ there may exist points that correspond to nontrivial action of the holonomy without the requirement of introducing any additional fields in the physical model. Such points are for example the isolated points of $\mc{M}(B,G)$. It is worthwhile to pursue the search of such canonical points in $\mc{M}(B,G)$, as they may lead to some new spontaneously occurring quantum statistical phenomena.

\section{Configuration spaces of graphs}\label{chap:models}
The general structure of configuration spaces of graphs has been introduced in section \ref{sec:qkin}. For computational purposes, we use discrete models of graph configuration spaces. By a discrete model we understand a $CW$-complex which is a deformation retract of $C_n(\Gamma)$. The existence of discrete models for graph configuration spaces enables us to use standard tools from algebraic topology to compute homology groups of graph configuration spaces. In particular, we use different kinds of homological exact sequences. There are two discrete models that we use. 
\begin{enumerate}
\item Abram's discrete configuration space \cite{AbramsPhD}. The Abram's deformation retract of $C_n(\Gamma)$ is denoted by $D_n(\Gamma)$. We use Abram's discrete model mainly in the first part of this paper, where we apply discrete Morse theory to the computation of homology groups of some small canonical graphs (section \ref{sec:morse}).

\item The discrete model by \'{S}wi\k{a}tkowski \cite{swiatkowski} that we denote by $S_n(\Gamma)$. We use this model in sections \ref{sec:wheel}-\ref{sec:K_2p} to compute homology groups of configuration spaces of wheel graphs and some families of complete bipartite graphs. 
\end{enumerate}
\'{S}wi\k{a}tkowski model has an advantage over Abram's model in the sense that its dimension agrees with the homological dimension of $C_n(\Gamma)$, and as such, stabilises for sufficiently large $n$. The dimension of Abram's model is equal to $n$ for sufficiently large $n$. Hence, the \'{S}wi\k{a}tkowski model is more suitable for rigorous calculations. However, sometimes it is more convenient to use Abram's model with the help of discrete Morse theory. The computational complexity of numerically calculating the homology groups of $C_n(\Gamma)$ for a generic graph is comparable in both approaches.

\paragraph*{Abrams discrete model}\label{sec:abrams}\ Let us next describe in detail the discrete configuration spaces $D_n(\Gamma)$ by Abrams. For the deformation retraction from $C_n(\Gamma)$ to $D_n(\Gamma)$ to be valid, the graph must be simple and {\it sufficiently subdivided} which means that 
\begin{itemize}
\item each path between distinct vertices of degree not equal to 2 passes through at least $n-1$ edges,
\item each nontrivial loop passes through at least $n+1$ edges. 
\end{itemize}
The discrete configuration space $D_n(\Gamma)$ is a cubic complex. The $n$-dimensional cells in $D_n(\Gamma)$ are of the following form.
\begin{eqnarray*}
\Sigma^{n}(D_n(\Gamma))=\{\{e_1,\dots,e_n\}:\ e_i\in E(\Gamma),\ e_i\cap e_j=\emptyset\ {\rm\ for\ }\ i\neq j\}.
\end{eqnarray*}
We denote cells of $D_n(\Gamma)$ by the set notation using curly brackets. Lower dimensional cells are described by sets of edges and vertices from $\Gamma$ that are mutually disjoint. A $d$-dimensional cell consists of $d$ edges and $n-d$ vertices. In other words, cells from $\Sigma^{d}(D_n(\Gamma))$ are of the form
\[\Sigma^{d}(D_n(\Gamma))=\{\sigma \subset E(\Gamma)\cup V(\Gamma):\ |\sigma|=n,\ |\sigma\cap E(\Gamma)|=d,\ \epsilon\cap \epsilon'=\emptyset\ \forall_{\epsilon,\epsilon'\in\sigma}\}.\]
In particular when there are not enough pairwise disjoint edges in the sufficiently subdivided $\Gamma$, the dimension of the discrete configuration space can be smaller than $n$.

In order to define the boundary map, we introduce a suitable order on vertices of $\Gamma$, following \cite{FSbraid,KoPark}. To this end, we choose a spanning tree $T\subset \Gamma$ and fix its planar embedding. We also fix the root $*$ of $T$ by picking a vertex of degree $1$ in $T$.  For every $v\in V(\Gamma)$ there is the unique path in $T$ that joins $v$ and $*$, called the geodesic $g_{v,*}$. For every vertex with $d(v)\geq 2$ we enumerate the edges adjacent to $v$ with numbers $0,1,\dots,d(v)-1$. The edge contained in $g_{v,*}$ has label $0$. The remaining edges are labelled increasingly, according to their clockwise order starting from edge $0$. The enumeration procedure for vertices goes in an inductive manner. The root has number $1$. If vertex $v$ has label $k$ and $d(v)=2$, the vertex adjacent to $v$ is given label $k+1$. Otherwise, if $d(v)\geq 2$, the vertex adjacent to $v$ in the lowest direction with vertices that have not been yet labelled is given label $k_{max}+1$, where $k_{max}$ is the maximal label among all of the already labelled vertices. If $d(v)=1$, we look for essential vertices in $g_{v,*}$ and go back to the closest essential vertex that contains a direction with unlabelled vertices. In other words, the vertices are labelled in the clockwise direction. This way every edge is given an initial and terminal vertex that we denote by $\iota(e)$ and $\tau(e)$ respectively. The terminal vertex is the vertex with the lower index, i.e. $\tau(e)<\iota(e)$. We can unambiguously specify an edge by calling its initial and terminal vertices, hence we denote the edges by $e_{\tau}^{\iota}$. Given a cell from $D_n(\Gamma)$
\[\sigma=\{e_1,\dots,e_d,v_1,\dots,v_{n-d}\},\]
we order the edges from $\sigma$ according to their terminal vertices, i.e. $\tau(e_1)<\tau(e_2)<\dots<\tau(e_d)$. The $i$th pair of faces from the boundary of $\sigma$ reads
\begin{eqnarray*}
\left(\partial^\iota\sigma\right)_i:=\{e_1,\dots,e_{i-1},e_{i+1},\dots,e_d,v_1,\dots,v_{n-d},\iota(e_i)\}, \\
\left(\partial^\tau\sigma\right)_i:=\{e_1,\dots,e_{i-1},e_{i+1},\dots,e_d,v_1,\dots,v_{n-d},\tau(e_i)\}.
\end{eqnarray*}
The full boundary of $\sigma$ is given by the following alternating sum of faces.
\begin{equation}\label{eq:boundary}
\partial\sigma=\sum_{i=1}^k(-1)^i\left(\left(\partial^\iota\sigma\right)_i-\left(\partial^\tau\sigma\right)_i\right).
\end{equation}
For examples, see section \ref{sec:first-homology} and section \ref{chap:computations}.

\paragraph*{\'{S}wi\k{a}tkowski discrete model}\ \ \'{S}wi\k{a}tkowski complex is denoted by $S_n(\Gamma)$. In order to define it, we regard graph $\Gamma$ as a set of edges $E$, vertices $V$ and half-edges $H$. A half-edge of $e\in E(\Gamma)$ assigned to vertex $v$, $h(v)\subset e$, is the part $e$ which is an open neighbourhood of vertex $v$. Intuitively, the half-edges are places, where the particles are allowed to `slide'. By $e(h)$ we will denote the unique edge, for which $e\cap h\neq \emptyset$. Similarly, we have vertex $v(h)$ as the vertex, for which $h$ is a neighbourhood. By $H(v)$ we will denote all half edges that are incident to vertex $v$. Chain complex $S(\Gamma)=\bigoplus_n S_n(\Gamma)$ reads
\[S(\Gamma)=\ZZ[E]\otimes\bigotimes_{v\in V}S_v,\]
where $S_v=\ZZ\langle v, h\in H(v), \emptyset\rangle$. This is a bigraded module with respect to the multiplication by $E(\Gamma)$ (a bigraded $\ZZ[E]$ module). The degrees of the components are
\[|v|=(0,1),\ |e|=(0,1),\ |h|=(1,1).\]
The boundary map reads
\[\partial v=\partial e=0,\ \partial h=e(h)-v(h).\]
The boundary map for elements of a higher degree is determined by the Eilenberg-Zilber theorem:
\[\partial(\chi\otimes\eta)=(\partial\chi)\otimes\eta+(-1)^{d}\chi\partial\eta\]
for $d$-chain $\chi$. There is a canonical basis for $S(\Gamma)$, whose elements of degree $(d,n)$ are of the form
\begin{gather}\label{cells}
h_1\dots h_dv_1\dots v_ke_1^{n_1}\dots e_l^{n_l},\ \{v_1\dots v_k\}\cap\{h_1,\dots h_d\}=\emptyset,\\\nonumber  d+k+n_1+\dots+n_l=n.
\end{gather}
The basis elements form a cube complex. In calculations we use the notion of support of a given cell or a chain.
\begin{definition}\label{def:supp}
The support of $d$-cell $c=h_1\dots h_dv_1\dots v_ke_1^{n_1}\dots e_l^{n_l}\in S_n(\Gamma)$ is the set of the corresponding edges and vertices of $\Gamma$
\[{\rm Supp}(c):=\left(\bigcup_{i=1}^d \{e(h_i),v(h_i)\}\right)\cup\{v_1,\dots, v_k\}\cup\{e_1,\dots,e_l\}\subset E(\Gamma)\cup V(\Gamma).\]
The support of a chain $b=\sum_i p_i c_i$, $p_i\in\ZZ$ is given by
\[{\rm Supp}(b):=\bigcup_{i=1}^d{\rm Supp}(c_i).\]
\end{definition}

In this paper we will also use a variation of $S(\Gamma)$ which we will call the reduced \'{S}wi\k{a}tkowski complex with respect to a subset of vertices $U\subset V(\Gamma)$ and denote by $\tilde S^U(\Gamma)$. In most cases, the reduced complexes lack a canonical basis, however they have a smaller number of generators than $S(\Gamma)$. The reduction is done by changing the generators at vertex $v$ to differences of half edges $h_{ij}:=h_i-h_j,\ h_i,h_j\in H(v)$, $\tilde S_v:=\ZZ\langle \emptyset, h_{ij}\rangle$.
\[\tilde S^U(\Gamma)=\ZZ[E]\otimes\bigotimes_{v\in V\backslash U}S_v\otimes\bigotimes_{v\in U}\tilde S_v.\]
Intuitively, this means that effectively, the particles always slide from one half-edge to another without staying at the central vertex. Both reduced and the non-reduced \'{S}wi\k{a}tkowski complexes have the same homology groups \cite{Knudsen}. From now on, the default complex we will work with is the complex which is reduced with respect to all vertices of degree one. Intuitively, this means that we do not consider redundant cells, where particles move from an edge to some vertex of valency one. Such complexes have the canonical basis which corresponds to cells of a cube complex of the form (\ref{cells}). By a slight abusion of notation, we will denote such a default reduced complex by $S(\Gamma)$. In other words, from now on
\[S(\Gamma):=\ZZ[E]\otimes\bigotimes_{v\in V: d(v)>1}S_v.\]
For examples, see figure \ref{sn-examples}.
\begin{figure}[ht]
\centering
\includegraphics[width=0.9\textwidth]{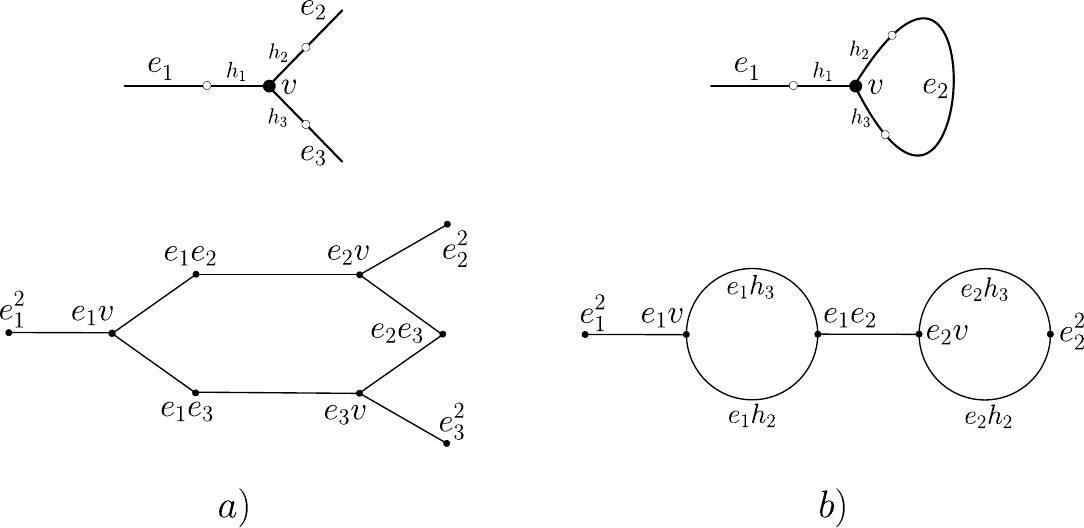}
\caption{\'{S}wi\k{a}tkowski complex of the $Y$-graph and of the lasso graph, where vertices of degree $1$ have been reduced. a) \'{S}wi\k{a}tkowski complex of $C_2(Y)$. Only vertices of $S_2(\Gamma)$ are captioned. The $Y$-cycle reads $e_1(h_2-h_3)+e_2(h_3-h_1)+e_3(h_1-h_2)$. b) \'{S}wi\k{a}tkowski complex of $C_2(\Gamma)$ for the lasso graph. Vertices and some chosen edges of $S_2(\Gamma)$ are captioned. The $O$-cycles are $e_1(h_2-h_3)$ and $e_2(h_2-h_3)$. The $Y$-cycle is their sum, hence can be written as $(e_1-e_2)(h_2-h_3)$.}
\label{sn-examples}
\end{figure}
As a direct consequence of the dimension of $S_n(\Gamma)$, we get the following fact.
\begin{fact}\label{fact:dimensions}
Let $\Gamma$ be a graph. Then, the following homology groups of $C_n(\Gamma)$ vanish.
\[H_d(C_n(\Gamma))=0{\rm\ if\ } d<n{\rm\ or\ }d>N_\Gamma,\]
where $N_\Gamma=|\{v\in V(\Gamma):\ d(v)\geq 3\}|$.
\end{fact}

\paragraph*{Vertex blowup} \ In the following, we will explore relations on homology groups that stem from blowing up a vertex of $\Gamma$: $\Gamma\to\Gamma_v$ (Fig. \ref{blowup}). 
 \begin{figure}[H]
\centering
\includegraphics[width=0.4\textwidth]{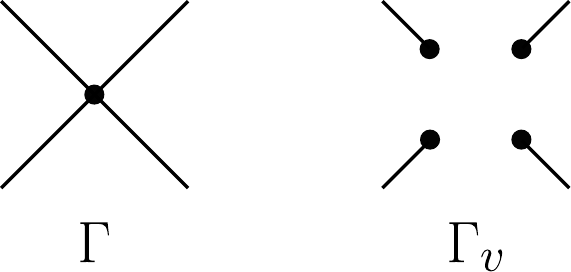}
\caption{Vertex blow up at vertex $v$ in $\Gamma$.}
\label{blowup}
\end{figure}
We borrow this nomenclature and the methodology of this subsection from \cite{Knudsen}.
We start with the reduced complex with respect to vertex $v$, $\tilde S^v(\Gamma)$. Any chain $b\in\tilde S^v(\Gamma)$ can be decomposed in a unique way by extracting the part that involves generators from $\tilde S_v$. In order to do it, we fix a half-edge $h_0\in H(v)$ and write $b$ as
\[b=b_0+\sum_{h\in H(v)\backslash h_0} (h_0-h)b_h.\]
Note that chains $b_0$ and $b_h$ belong to $S(\Gamma_v)$. We associate two chain maps to the above decomposition. The first map $\phi$ is the embedding of any chain $b_0$ from $S(\Gamma_v)$ to $\tilde S^v(\Gamma)$. Clearly, this map is injective and commutes with the boundary operator.
\begin{gather*}
\phi_n:\ S_n(\Gamma_v)\to \tilde S_n^v(\Gamma), \phi(b_0)=b_0\in \tilde S^v(\Gamma), \\
\end{gather*}
The other map $\psi$ is the projection of $b\in\tilde S^v(\Gamma)$ to its $h$-components. It assigns a number of $n-1$-particle $d-1$-chains to a $n$-particle $d$-chain in the following way
\begin{gather*}
\psi_n: \tilde S_n^v(\Gamma)\to\bigoplus_{h\in H(v)\backslash h_0}S_{n-1}(\Gamma_v),\ \psi(b)=\bigoplus_{h\in H(v)\backslash h_0} b_h.
\end{gather*}
Map $\psi$ is surjective, because any chain $b'\in S_{n-1}(\Gamma_v)$ can be obtained by $\psi$ for exmaple from chain $(h_0-h)b'\in\tilde S_n^v(\Gamma)$. In order to see that $\psi$ is a chain map, consider a cycle $c\in \tilde S_n^v(\Gamma)$. We have
\[0=\partial c=\partial c_0+\sum_{h\in H(v)\backslash h_0} \left((e(h_0)-e(h))c_h-(h_0-h)\partial c_h\right).\]
Grouping the summands that entirely belong to $S_{n-1}(\Gamma_v)$, we get
\begin{gather*}
\partial c_0+\sum_{h\in H(v)\backslash h_0} (e(h_0)-e(h))c_h=0, \\
\sum_{h\in H(v)\backslash h_0} (h_0-h)\partial c_h=0.
\end{gather*}
By the same argument, the second equation implies that $\partial c_h=0$ for all $h\in H(v)\backslash h_0$. We can write down the two maps as a short exact sequence
\begin{equation}\label{short-blowup}
0\rightarrow S_n(\Gamma_v)\xrightarrow{\phi_n}\tilde S_n^v(\Gamma)\xrightarrow{\psi_n}\bigoplus_{h\in H(v)\backslash h_0}S_{n-1}(\Gamma_v)\rightarrow0.
\end{equation}
Short exact sequence (\ref{short-blowup}) of chain maps implies the long exact sequence of homology groups
\begin{gather}\label{long-blowup}
\dots\xrightarrow{\Psi_{n,d+1}} \bigoplus_{h\in H(v)\backslash h_0}H_d\left(S_{n-1}(\Gamma_v)\right)\xrightarrow{\delta_{n,d}}H_d\left(S_n(\Gamma_v)\right)\xrightarrow{\Phi_{n,d}}H_d\left(\tilde S^v_n(\Gamma)\right)\xrightarrow{\Psi_{n,d}} \\ \nonumber \xrightarrow{\Psi_{n,d}}\bigoplus_{h\in H(v)\backslash h_0}H_{d-1}\left(S_{n-1}(\Gamma_v)\right)\xrightarrow{\delta_{n,d-1}}H_{d-1}\left(S_n(\Gamma_v)\right)\xrightarrow{\Phi_{n,d-1}}\dots,
\end{gather}
where the connecting homomorphism reads
\[\delta [b_h]=[\partial\left((h_0-h)b_h\right)]=e(h_0)[b_h]-e(h)[b_h].\]
Long exact sequence (\ref{long-blowup}) implies a collection of short exact sequences
\[0\rightarrow \coker\left(\delta_{n,d}\right)\xrightarrow{}H_d\left(\tilde S^v_n(\Gamma)\right)\xrightarrow{}\ker\left(\delta_{n,d-1}\right)\rightarrow 0.\]
Intuitively, the $\coker\left(\delta_{n,d}\right)$ identifies different distributions of free particles in $S_n(\Gamma_v)$ on the two sides of the junction $h_0-h$ and $\ker\left(\delta_{n,d-1}\right)$ is responsible for creating new cycles at vertex $v$ (for example, the $c_Y$ cycles).

\subsection{$O$-cycles and $Y$-cycles}\label{sec:first-homology}

There are some particular types of cycles that play an important role in this work. These are $O$-cycles and $Y$-cycles. We specify them for the Abram's model. The construction for $S_n(\Gamma)$ is fully analogous.

\begin{definition}\label{definition:o-cycle}
Let $O\subset \Gamma$ be a simple cycle (an embedding of $S^1$ in $\Gamma$). Choose sign coefficients $s_e\in\{-1,1\},\ e\in O$ such that $\partial \sum_{e\in O} s_e {e}=0$ in $D_1(\Gamma)$. An $O$-cycle in $D_n(\Gamma)$ is a $1$-chain of the form
\[c_O:=\sum_{e\in O}s_e\{e,v_1,\dots,v_{n-1}\},\]
where $\{v_1,\dots,v_{n-1}\}\cap O=\emptyset$ is some choice of vertices. In order to define an $O$-cycle in $S_n(\Gamma)$, note that for all $v\in V(\Gamma)\cap O$, set $H(v)\cap O$ contains exactly two half-edges. We denote these half-edges by $h_v, h_v'$, where the labels are such that $\partial \sum_{v\in V(\Gamma)\cap O}(h_v'-h_v)=0$. Then, 
\begin{gather*}c_O=\left(\sum_{v\in V(\Gamma)\cap O}(h_v'-h_v)\right)\otimes \left(\bigotimes_{w\in W}w\right)\otimes \left(\bigotimes_{e\in E(\Gamma)}e^{n_e}\right),\\ W\subset (V(\Gamma)-V(\Gamma)\cap O),\ \#W+\sum_{e\in E(\Gamma)}n_e=n-1.
\end{gather*}
\end{definition}

\begin{definition}\label{definition:Y-cycle}
Let $Y\subset \Gamma$ be a $Y$-subgraph of $\Gamma$ spanned on vertices $u_0,u_h,u_1,u_2$ such that $u_0,\ u_1,\ u_2$ are adjacent to $u_h$ and $u_0<u_h<u_1<u_2$. The $Y$-cycle in $D_2(\Gamma)$ associated to subgraph $Y$ is of the following form 
\[c_Y:=\{e_{u_h}^{u_1},u_0\}+\{e_{u_0}^{u_h},u_1\}+\{e_{u_h}^{u_2},u_1\}-\{e_{u_h}^{u_1},u_2\}-\{e_{u_0}^{u_h},u_2\}-\{e_{u_h}^{u_2},u_0\}.\]
A $Y$-cycle in $D_n(\Gamma)$ is formed by distributing the free particles outside of subgraph $Y$, i.e. 
\[c_Y^{(n)}:=\sum_{\sigma\in c_Y}s_\sigma \left(\sigma\cup \{v_1,\dots,v_{n-2}\}\right),\]
where $\{v_1,\dots,v_{n-2}\}\cap Y=\emptyset$ and $s_\sigma$ is the sign of cell $\sigma$ in cycle $c_Y$. In order to define the $Y$-cycle in $S_n(\Gamma)$, denote the half edges of subgraph $Y$ as $\{h_i\}_{i=0}^2$, where $h_i\in H(u_h)$ are such that $e(h_0)=e_{u_0}^{u_h}$, $e(h_1)=e_{u_h}^{u_1}$, $e(h_2)=e_{u_h}^{u_2}$. Then,
\[c_{Y}=e_{u_0}^{u_h}(h_2-h_3)+e_{u_h}^{u_1}(h_3-h_1)+e_{u_h}^{u_2}(h_1-h_2).\]
Cycle $c_Y^{(n)}\in S_n(\Gamma)$ is formed by multiplying $c_Y$ by a suitable polynomial in $V(\Gamma)$ and $E(\Gamma)$.
\[c_Y^{(n)}=c_Y\otimes \left(\bigotimes_{w\in W}w\right)\otimes \left(\bigotimes_{e\in E(\Gamma)}e^{n_e}\right),\ W\subset (V(\Gamma)-\{u_h\}),\ \#W+\sum_{e\in E(\Gamma)}n_e=n-2.\]
\end{definition}

\begin{figure}[H]
\centering
\includegraphics[width=0.9\textwidth]{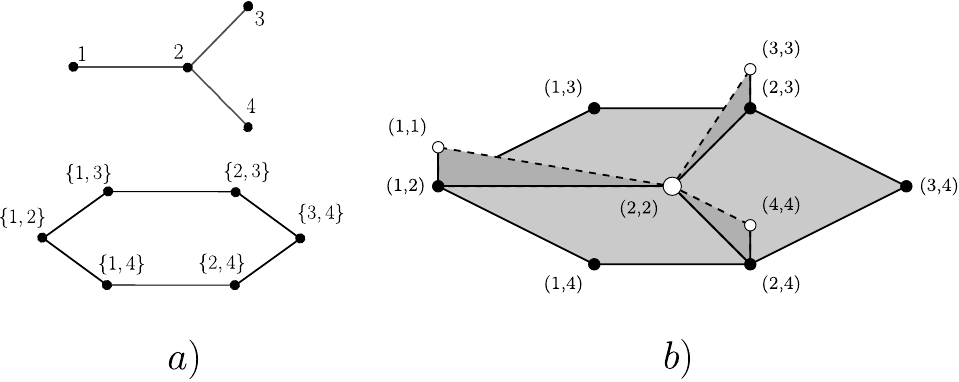}
\caption{A $Y$-graph, its configuration space (b) and its discrete configuration space $D_2(\Gamma)$ (a).}
\label{y_graph}
\end{figure}

It has been shown in \cite{HKRS} that subject to certain relations, cycles $c_O$ and $c_Y^{(n)}$ generate $H_1(D_n(\Gamma))$ (see also \cite{Knudsen} for the proof of an analogous fact for $H_1(S_n(\Gamma))$). The fundamental relation between $Y$-cycles is shown on Fig. \ref{y_relation} and Fig. \ref{rel_cycle_y}.
 \begin{figure}[H]
\centering
\includegraphics[width=0.7\textwidth]{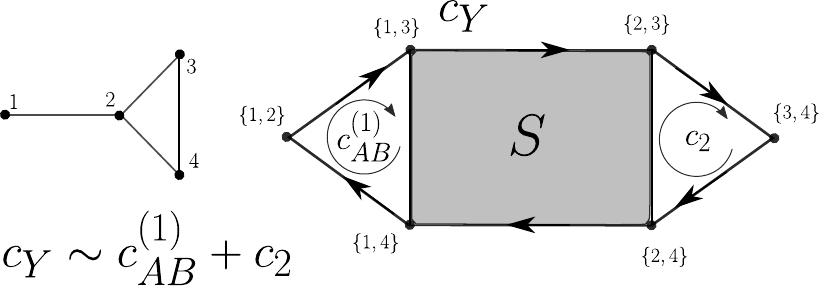}
\caption{The fundamental relation between the two-particle cycle on a $Y$-graph and the $AB$-cycle and a two-particle cycle $c_2$ in the lasso graph.}
\label{y_relation}
\end{figure} 
\noindent Cycle $c_{AB}^{(1)}$ is the cycle, where one particle goes around the cycle in the lasso graph and the other particle occupies vertex $1$.
\[c_{AB}^{(1)}=c_O\times\{1\}=\{e_2^3,1\}+\{e_3^4,1\}-\{e_2^4,1\}.\]
Cycle $c_2$ is the cycle, where two particles go around the cycle in lasso.
\[c_2=\{e_2^4,3\}-\{e_2^3,4\}-\{e_3^4,2\}.\]
It is straightforward to check that 
\begin{equation}\label{y-ab-rel}
c_{AB}^{(1)}+c_2-c_Y=\partial S,
\end{equation}
where $S=\{e_1^2,e_3^4\}$. Consider next a situation, where two disjoint $Y$-graphs share one cycle $c_O$ and their free ends are connected by a path $p_{v_1,v_2}$ which is disjoint with $c_O$ (Fig. \ref{rel_cycle_y}). In other words, consider an embedding of a graph which is isomorphic to the $\Theta$- graph\footnote{The $\Theta$ graph consists of two vertices which are connected by three edges. It can be also viewed as complete bipartite graph $K_{2,3}$.}.
 \begin{figure}[H]
\centering
\includegraphics[width=0.7\textwidth]{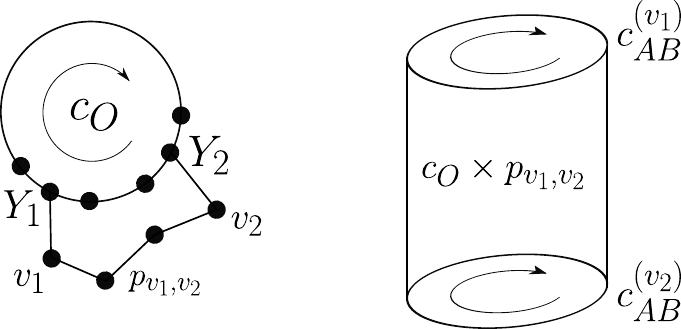}
\caption{Cycles $c_{Y_1}$ and $c_{Y_2}$ are homologically equivalent.}
\label{rel_cycle_y}
\end{figure}
\noindent Then,
\begin{gather*}
c_{AB}^{(v_1)}+c_2-c_{Y_1}=\partial S_1, \\
c_{AB}^{(v_2)}+c_2-c_{Y_2}=\partial S_2.
\end{gather*}
Subtracting both equations, we get 
\begin{equation}\label{ab-y-rel}
c_{Y_1}-c_{Y_2}=\partial(S_2-S_1)+c_{AB}^{(v_1)}-c_{AB}^{(v_2)}.
\end{equation}
But the existence of $p_{v_1,v_2}$ gives us that $c_{AB}^{(v_1)}-c_{AB}^{(v_2)}=\partial\left(c_O\times p_{v_1,v_2}\right)$. This in turn means that $c_{Y_1}$ and $c_{Y_2}$ are homologically equivalent. Relation
\begin{equation}\label{theta3-rel}
c_{Y_1}-c_{Y_2}=\partial\left(S_2-S_1+c_O\times p_{v_1,v_2}\right)
\end{equation}
will be called a $\Theta$-relation. It turns out that considering all $\Theta$-relations stemming from different $\Theta$-subgraphs and relations (\ref{ab-y-rel}) that express different distributions of particles in the $O$-cycles as differences of $Y$-cycles, one can compute the first homology group of $D_n(\Gamma)$. Let us next summarise the results concerning the structure of the first homology group of graph configuration spaces. We formulate the results assuming that the considered graphs are simple. The general form of the first homology group reads
\begin{equation}\label{eq:1st-hom}
H_1(D_n(\Gamma),\ZZ)=(\ZZ)^N\oplus (\ZZ_2)^L,
\end{equation}
where $N$ and $L$ are the numbers of copies of $\ZZ$ and $\ZZ_2$ respectively. Numbers $N$ and $L$ depend on the planarity and some combinatorial properties of the given graph \cite{HKRS,KoPark}. The $\ZZ_2$-components appear when $\Gamma$ is non-planar and have the interpretation of different fermionic/bosonic statistics that may appear locally in different parts of a given graph (see \cite{HKRS}).

\section{Calculation of homology groups of graph configuration spaces}\label{chap:computations}
This section contains the techniques that we use for computing homology groups of graph configuration spaces. We tackle this problem from the `numerical' and the `analytical' perspective. The numerical approach means using a computer code for creating the boundary matrices and then employing the standard numerical libraries for computing the kernel and the elementary divisors of given matrices. The procedures for calculating the boundary matrices of $D_n(\Gamma)$, $S_n(\Gamma)$ and the Morse complex (see section \ref{sec:morse}) were written by the authors of this paper, based on papers \cite{FSbraid,KoPark}. The analytical approach means computing the homology groups for certain families of graphs by suitably decomposing a given graph into simpler components and using various homological exact sequences. Recently in the mathematical community, there has been a growing interest in computing the homology groups of graph configuration spaces. A significant part of the recent work has been devoted to explaining certain regularity properties of the homology groups of $C_n(\Gamma)$ \cite{Ramos,RW17,Ramos2,Lut17,LC18,Knudsen-stabilisation}. 

 \subsection{Product cycles}\label{sec:prod-cycles}
 Considering simultaneous exchanges of pairs of particles on disjoint $Y$-subgraphs of $\Gamma$ and the $O$-type cycles with the remaining particles distributed on the free vertices of $\Gamma$, one can construct some generators of $H_*(D_n(\Gamma))$ or $H_*(S_n(\Gamma))$. Such cycles are products of $1$-cycles, hence are isomorphic to tori embedded in the discrete configuration space. To construct a product $d$ cycle in $D_n(\Gamma)$, we choose $Y$-subgraphs of $\Gamma$ $\{Y_i\}_{i=1}^{d_Y}$ and cycles in $\Gamma$ ($O$-subgraphs of $\Gamma$) $\{O_i\}_{i=1}^{d_O}$, where $d_Y+d_O=d$. All the chosen subgraphs must be mutually disjoint.
 \[Y_i\cap Y_j=O_i\cap O_j=\emptyset\ {\rm for}\ i\neq j,\ Y_i\cap O_j=\emptyset\ {\rm for\ all}\ i,j.\]
 Moreover, we choose vertices $\{v_1,\dots,v_{n-2d_Y-d_O}\}\subset V(\Gamma)$, so that $v_i\cap O_j=v_i\cap Y_j=\emptyset$ for all $i,j$. Product cycle on $Y_1\times\dots \times Y_{d_Y}\times O_1\times\dots \times O_{d_O}$ with the free particles distributed on $\{v_1,\dots,v_{n-2d_Y-d_O}\}$ is the following chain.
 \[c_{Y_1}\otimes\dots \otimes c_{Y_{d_Y}}\otimes c_{O_1}\otimes\dots\otimes c_{O_{d_O}}\otimes \{v_1,\dots,v_{n-2d_Y-d_O}\}.\]
In an analogous way, we form product cycles in $S_n(\Gamma)$.

We study such product cycles for configuration spaces of different graphs and describe relations between them. So far, it has been known that product cycles generate the second homology of the two particle configuration space of a simple graph \cite{BF09} and all homology groups for an arbitrary number of particles on tree graphs \cite{MS17} (see also \cite{FStree}). In this section, we find new families of graphs, for which product cycles generate some homology groups of their configuration spaces. These cases are
\begin{itemize}
\item all homology groups of the configuration spaces of wheel graphs (section \ref{sec:wheel}),
\item all homology groups of the configuration space of graph $K_{3,3}$, except the third homology group (section \ref{sec:k33}), 
\item the second homology group of a simple graph which has at most one vertex of degree greater than $3$.
\end{itemize}
 In sections \ref{sec:k33} and \ref{sec:K_2p} we also discuss examples of cycles that are different than tori. In particular, we compute all homology groups of configuration spaces of complete bipartite graphs $K_{2,p}$ that are often pointed out in the literature as an unsolved example, where the simple use of product cycles is not sufficient to generate the homology groups. We show that some of the generators of $H_*(S_n(K_{2,p}))$ are cycles of a new type that have the homotopy type of triple tori.

\subsection{Discrete Morse theory for Abrams model}\label{sec:morse}
In this subsection, we apply a version of Forman's discrete Morse theory \cite{Forman} for Abram's discrete model that was formulated in \cite{FSbraid} (see also \cite{ASmorse}). The results are listed in tables \ref{morse_table} and \ref{morse-petersen}.

The discrete Morse theory relies on constructing a discrete gradient flow $F$ which is a linear map mapping $d$-chains to $d$-chains. Moreover, map $F$ has the property that for any chain $c$, we have $F^{r+1}(c)=F^{r}(c)$ for some $r$. The Morse complex is the chain complex of chains invariant under $F$. The basis of such invariant chains consists of {\it critical cells}. There are {\it a priori} different ways to explicitly realise the discrete gradient flow for graph configuration spaces. We have chosen the realisation introduced in \cite{FSbraid}. Here, we do not review the details of this construction, but only present a pseudocode which shows schematically how to compute $H_d(D_n(\Gamma))$ using the knowledge of the boundary map in $D_n(\Gamma)$ and the list of critical cells of $F$ as cells in $D_n(\Gamma)$. We also direct the reader to public repository \cite{code} where we uploaded a Python implementation of the discrete Morse theory that we used in our work. The results of running the code for different graphs are collected in tables \ref{morse_table} and \ref{morse-petersen}.

\begin{algorithm}
\caption{Main steps of the algorithm for computing $H_d(D_n(\Gamma))$ via discrete Morse theory}\label{morse-alg}
\begin{algorithmic}[1]
\State \textbf{Input:} \text{Sufficiently subdivided graph $\Gamma$, number of particles $n$.}
\State \textbf{Output:} \text{$\beta_d(D_n(\Gamma))$, $T_d(D_n(\Gamma))$}
\State $\textit{F} \gets \text{flow of the discrete gradient vector field}$
\State $\partial \gets \text{the boundary map in $\mk{C}(D_n(\Gamma))$}$
\Procedure{MorseBoundaryMap}{d}
\State $\textit{critcellsd} \gets \text{list of critical $d$-cells}$
\State $\textit{critcellsdminus} \gets \text{list of critical $d-1$-cells}$
\State $D_{\mc{M}}\gets \text{integer matrix of size Length{\it(critcellsd)}$\times$Length{\it(critcellsdminus)}}$
\For{$i=0\ \bf{to}\ \text{Length({\it critcellsd})}$}
\State $\textit{b} \gets \text{$\partial(critcellsd[i])$}$
\Repeat 
\State $\textit{b} \gets \text{$F(b)$}$
\Until{$F(b) == b$}
\For{$\sigma'\ \bf{in}\ \textit{b}$}
\State $D_{\mc{M}}[i][{\rm Index}(\sigma', critcellsdminus)]\gets {\rm Coefficent}(\sigma',b)$
\EndFor
\EndFor
\State \textbf{return} $D_{\mc{M}}$
\EndProcedure
\State $D_d\gets {\rm MorseBoundaryMap}(d)$
\State $\textit{dimker} \gets \text{${\rm Length}(D_d[0])-{\rm MatrixRank}({\rm MorseBoundaryMap}(d))$}$
\State $D_{d+1}\gets {\rm MorseBoundaryMap}(d+1)$
\State $divisors\gets {\rm ElementaryDivisors}(D_{d+1})$
\State $nonzerodivisors\gets \text{number of nonzero elements of {\it divisors}}$
\State $torsion\gets \text{list of elements of {\it divisors} that are greater than $1$}$
\State \textbf{return} $(dimker-nonzerodivisors)$,\ torsion
\end{algorithmic}
\end{algorithm}

\begin{table}[ht]
\centering
\begin{tabular}{c|c|c|c|c}
$\Gamma$ & $n$ & $\beta_2(C_n(\Gamma))$ & $\beta_3(C_n(\Gamma))$ & $\beta_4(C_n(\Gamma))$ \\
\hline \hline
\multirow{4}{*}{$K_4$} & 3 & 3 & 0 & - \\
				& 4 & 9 & 0 & 0 \\
				& 5 & 15 & 0 & 0 \\
				& 6 & 21 & 4 & 0 \\
				& 7 & 27 & 16 & 0 \\
				& 8 & 33 & 40 & 1 \\
				& 9 & 39 & 80 & 6 \\
\hline \multirow{4}{*}{$K_{3,3}$} & 2 & 0 & - & - \\
				 & 3 & 8 & 0 & - \\
				& 4 & 19 & 1 & 0 \\
				& 5 & 28 & 10 & 0 \\
				& 6 & 37 & 39 & 0 \\
				& 7 & 46 & 88 & 0 \\
				& 8 & 55 & 157 & 15 \\
\hline
\multirow{5}{*}{$K_5$}  & 2 & 0 & - & - \\
 				& 3 & 30 & 0 & - \\
				& 4 & 76 & 1 & 0 \\
				& 5 & 116 & 77 & 0 \\
				& 6 & 156 & 381 & 0 \\
				& 7 & 196 & 961 & 0 \\

\hline
\end{tabular}
\caption{Betti numbers for chosen graphs computed using the discrete Morse theory \cite{FSbraid}. The calculated groups were torsion-free.}
\label{morse_table}
\end{table}

Table \ref{morse-petersen} presents the results for the second and third homology groups for graphs from the Petersen family (fig. \ref{fig:petersen}). 
\begin{figure}[ht]
\centering
\includegraphics[width=0.8\textwidth]{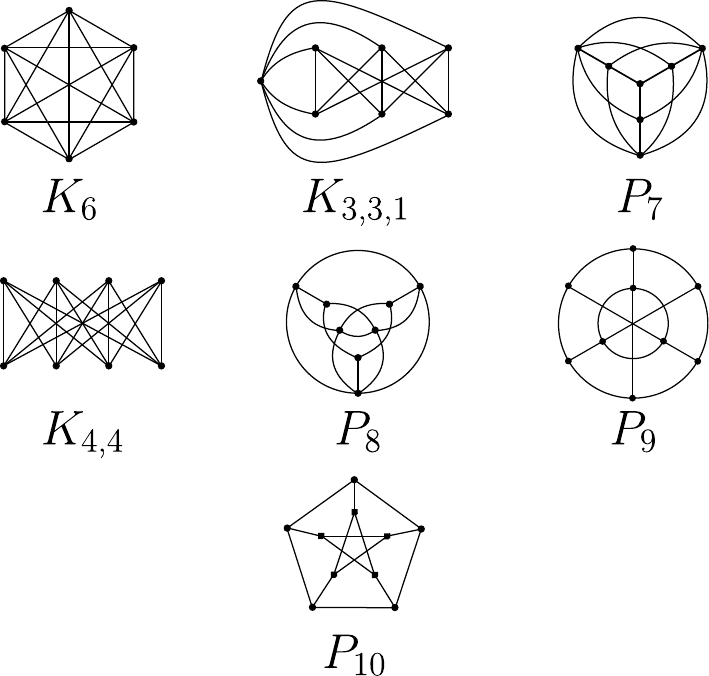}
\caption{Graphs that form the Petersen family.}
\label{fig:petersen}
\end{figure}
These graphs serve as examples, where torsion in higher homology groups appears. Interestingly, the torsion subgroups are always equal to a number of copies of $\ZZ_2$. This phenomenon can be explained by embedding a nonplanar graph in $\Gamma$ and considering suitable product cycles. The question about the existence of torsion different than $\ZZ_2$ in higher homologies remains open.

\begin{table}[H]
\centering
\begin{tabular}{c||c|c|c|c|c|c|c}
 & $K_6$ & $P_7$ & $K_{3,3,1}$ & $K_{4,4}$ & $P_8$ & $P_9$ & $P_{10}$ \\
\hline 
\hline\\[-1em]
$\beta_2(C_4(\Gamma))$ & $264$ & $177$ & $172$ & $144$ & $114$ & $70$ & $40$ \\
$T_2(C_4(\Gamma))$ & $\ZZ_2$ & $\ZZ_2$ & $\ZZ_2$ & $\left(\ZZ_2\right)^2$ & $\ZZ_2$ & $\ZZ_2$ & $\ZZ_2$ \\
\hline\\[-1em]
$\beta_3(C_6(\Gamma))$ & $4137$ & $2058$ & $1919$ & $1460$ & $986$ & $452$ & $191$ \\
$T_3(C_6(\Gamma))$ & $0$ & $0$ & $0$ & $\left(\ZZ_2\right)^{73}$ & $0$ & $0$ & $0$ \\
\hline
\end{tabular}
\caption{The first regular homology groups of order $2$ and $3$ for the Petersen family.}
\label{morse-petersen}
\end{table}

\subsection{Wheel graphs}\label{sec:wheel}
In this section, we deal with the class of wheel graphs. A wheel graph of order $m$ is a simple graph that consists of a cycle on $m-1$ vertices, whose every vertex is connected by an edge (called a spoke) to one central vertex (called the hub). We provide a complete description of the homology groups of configuration spaces for wheel graphs. In particular, we show that all homology groups are free. Therefore, in addition to tree graphs, wheel graphs provide another family of configuration spaces with a simplified structure of the set of flat complex vector bundles. The general methodology of computing homology groups for configuration spaces of wheel graphs is to consider only the product cycles and describe the relations between them. We justify this approach in subsection \ref{sec:wheel-swiatkowski}.

The simplest example of a wheel graph is graph $K_4$ which is the wheel graph of order $4$. Let us next calculate all homology groups of graph $K_4$ and then present the general method for any wheel graph. 

\subsubsection{Graph $K_4$}\label{sec:K4}
Graph $K_4$ is shown on figure \ref{K4}. It is the $3$-connected, complete graph on $4$ vertices.
 \begin{figure}[ht]
\centering
\includegraphics[width=0.3\textwidth]{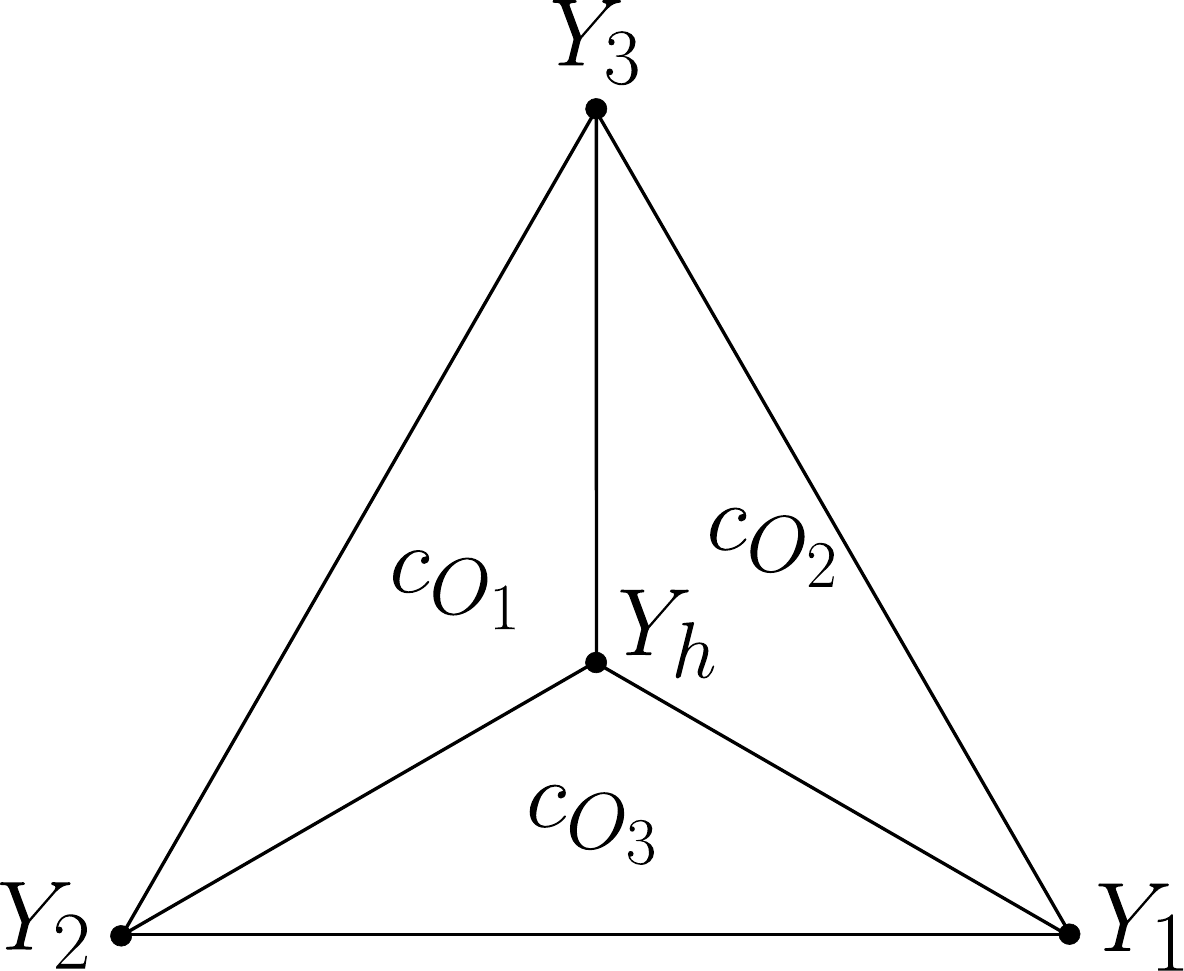}
\caption{Graph $K_4$ and the relevant $Y$-subgraphs and cycles. We omit the subdivision of edges in the picture.}
\label{K4}
\end{figure} 

\paragraph*{Second homology group}\ There are three independent cycles in $K_4$ graph. These are the cycles that contain the hub and two neighbouring vertices from the perimeter. However, any two such cycles always share some vertices. Hence, there are no tori that come from the products of $c_O$ cycles. Hence, the product $2$-cycles are either $c_Y\otimes c_O$ or $c_Y\otimes c_{Y'}$. There are four cycles of the first kind: $c_{Y_1}\otimes c_{O_1}$, $c_{Y_2}\otimes c_{O_2}$, $c_{Y_3}\otimes c_{O_3}$ and $c_{Y_h}\otimes c_{O}$, where $c_{O}$ is the outermost cycle. However, cycle $c_{Y_h}\otimes c_{O}$ can be expressed as a linear combination of cycles $c_{Y_1}\otimes c_{O_1}$, $c_{Y_2}\otimes c_{O_2}$, $c_{Y_3}\otimes c_{O_3}$. Therefore, the second homology of the three-particle configuration space is 
\[H_2(D_3(K_4))=\ZZ^3.\]
If $n>3$, there are still three independent $O\times Y$-cycles, as the differences between distributions of free particles in such cycles can always be expressed as combinations of $Y\times Y$-cycles. To see this, consider the following example. For $n=4$, consider the $O\times Y$-cycles that involve cycle $c_{O_1}$, subgraph $Y_1$ and one of three possible free vertices (Fig. \ref{K4-rel-ab}). The cycles are $c_{Y_1}\otimes c_{AB}^u$, $c_{Y_1}\otimes c_{AB}^v$, $c_{Y_1}\otimes c_{AB}^w$, where $c_{AB}^v:=c_{O_1}\times v$. From (\ref{y-ab-rel}) we have 
\begin{gather*}
c_{Y_2}\sim c_2+c_{AB}^v,\ c_{Y_3}\sim c_2+c_{AB}^w,\ c_{Y_h}\sim c_2+c_{AB}^u.
\end{gather*}
Subtracting the above equations and multiplying the results by $c_{Y_1}$, we get
\begin{gather*}
c_{Y_1}\otimes c_{Y_h}-c_{Y_1}\otimes c_{Y_2}\sim c_{Y_1}\otimes c_{AB}^u-c_{Y_1}\otimes c_{AB}^v, \\
c_{Y_1}\otimes c_{Y_h}-c_{Y_1}\otimes c_{Y_3}\sim c_{Y_1}\otimes c_{AB}^u-c_{Y_1}\otimes c_{AB}^w.
\end{gather*}
This means that the differences between distribution of particles in $AB$-cycles can be expressed as combinations of $Y\times Y$ cycles. This fact generalises to $n>4$ in a straightforward way.

 \begin{figure}[ht]
\centering
\includegraphics[width=0.4\textwidth]{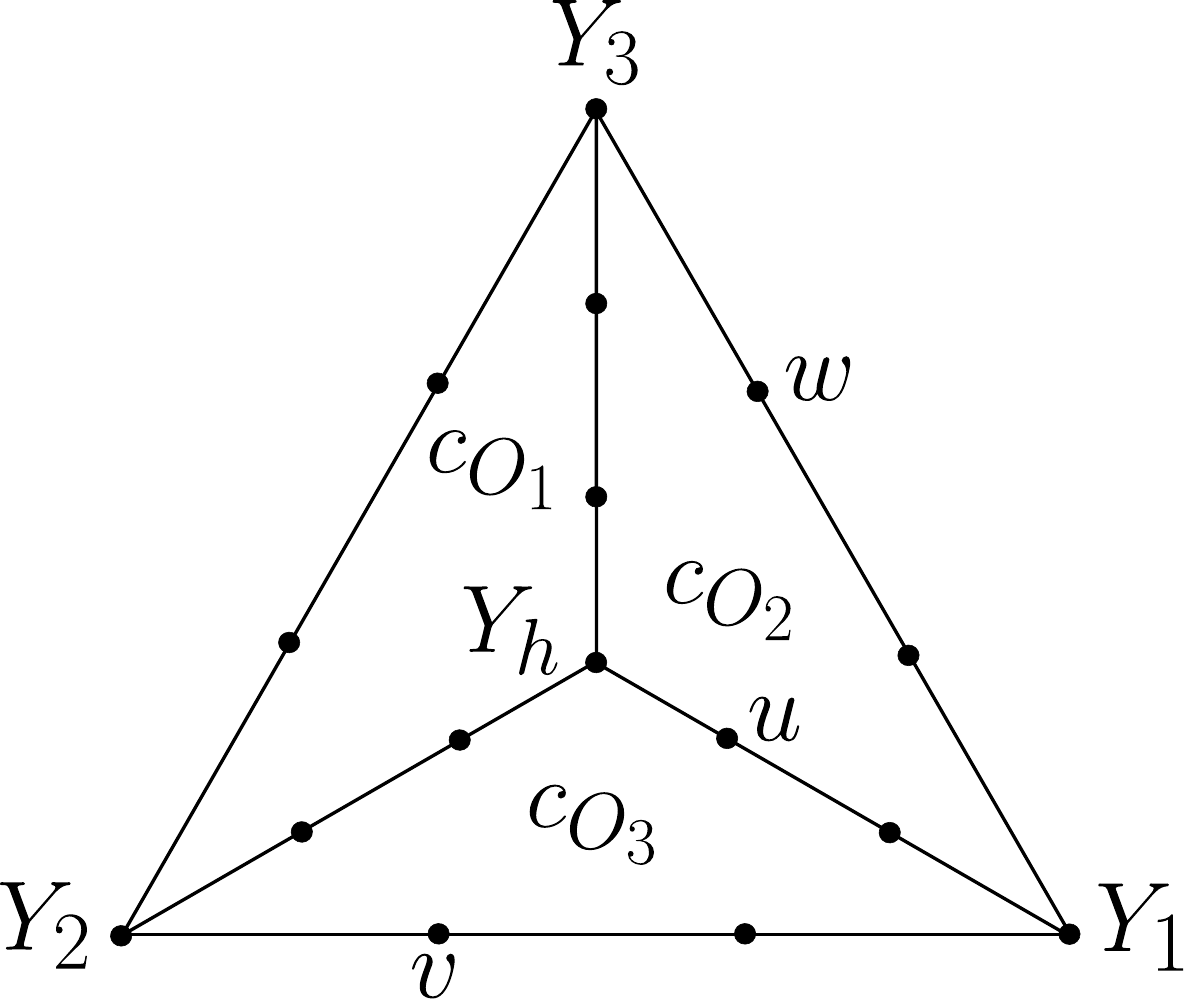}
\caption{Graph $K_4$ subdivided for $n=4$. Differences $c_{AB}^u-c_{AB}^v$ and $c_{AB}^u-c_{AB}^w$ are homologically equivalent to combinations of $Y\times Y$-cycles. $c_{Y_1}\otimes c_{Y_h}-c_{Y_1}\otimes c_{Y_2}$ and $c_{Y_1}\otimes c_{Y_h}-c_{Y_1}\otimes c_{Y_3}$ respectively. }
\label{K4-rel-ab}
\end{figure} 

Consider next all possible ways of choosing two $Y$-subgraphs. There are six $Y\times Y$-cycles modulo the distribution of free particles. Hence, if there are no free particles, i.e. when $n=4$, we have
\[H_2(D_4(K_4))=\ZZ^3\oplus\ZZ^6.\]
If $n>4$, we have to take into account the distribution of free particles in $\Gamma-(Y\cup Y')$. For a sufficiently subdivided graph one always ends up with two connected components (Fig. \ref{K4_removal}).
 \begin{figure}[ht]
\centering
\includegraphics[width=0.5\textwidth]{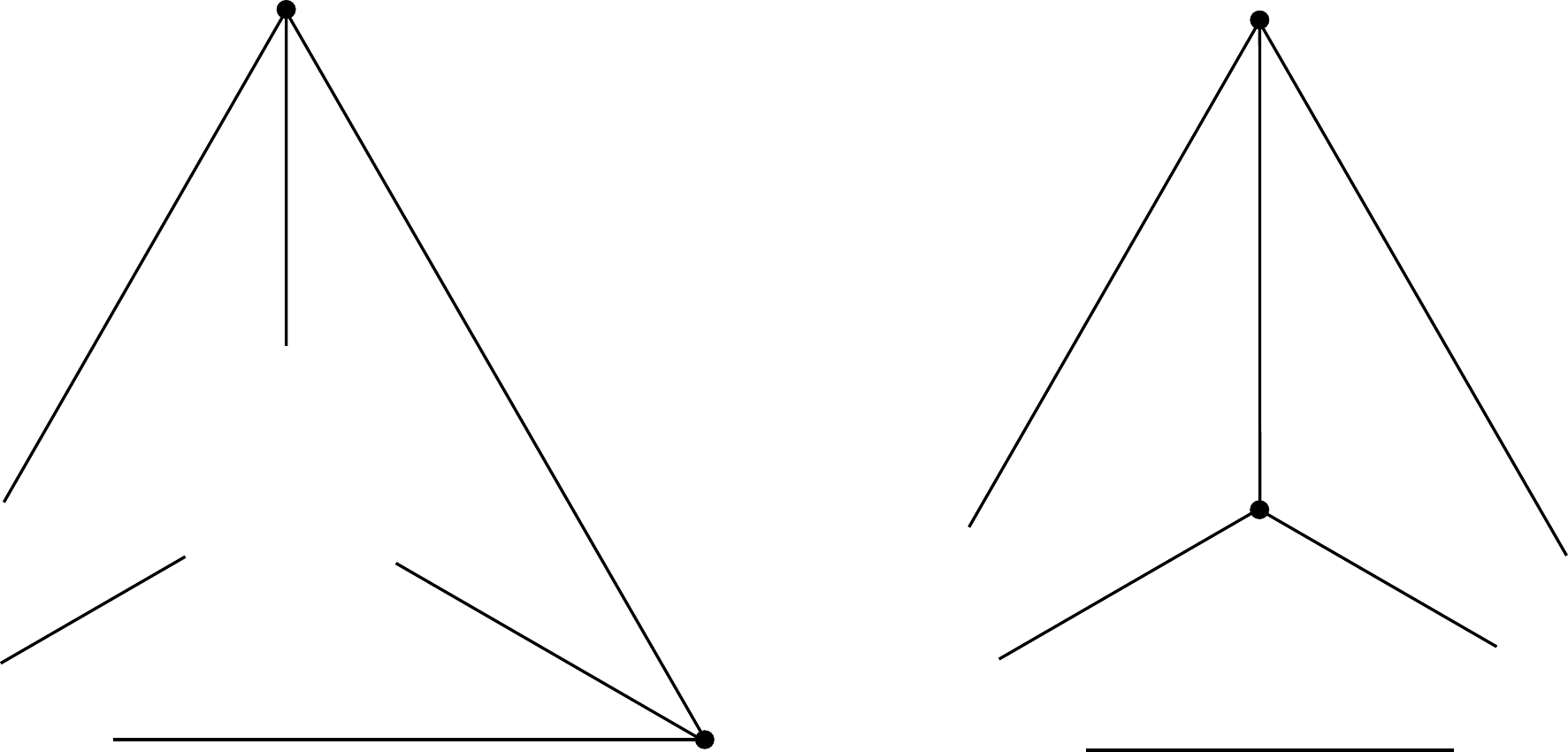}
\caption{Graph $K_4$ after removing two $Y$-subgraphs.}
\label{K4_removal}
\end{figure}
A $Y\times Y$-cycle involves $4$ particles, hence one has to calculate the number of all possible distributions of $n-4$ particles on those two components times the number of possible choices of the two $Y$-subgraphs. The number of all choices of the $Y$-subgraphs is ${4}\choose{2}$, while the number of possible distributions of $n-4$ particles on $2$ components is ${{n-4+2-1}\choose{2-1}}=n-3$. Hence, the contribution from $Y\times Y$ cycles reads
\begin{equation*}
{{4}\choose{2}}(n-3)=6(n-3),\ n\geq4.
\end{equation*}
Adding the contribution from $O\times Y$-cycles, the rank of the second homology group is then given by
\begin{equation*}
\beta_2(C_n(K_4))=3+6(n-3)=6n-15,\ n\geq3.
\end{equation*}

\paragraph*{Higher homology groups}\ The product generators of higher homologies are even simpler than in the case of the second homology. There are only basis cycles of $Y\times Y\times\dots\times Y$-type. After removing three and four $Y$-graphs, $K_4$ graph always disintegrates into $4$ and $6$ parts respectively. Taking into account the distributions of free particles, we get the following formulae for the Betti numbers.
\begin{gather*}
\beta_3(C_n(K_4))={{4}\choose{3}}{{n-6+4-1}\choose{4-1}}=4{{n-3}\choose{3}},\ n\geq6\\
\beta_4(C_n(K_4))={{4}\choose{4}}{{n-8+6-1}\choose{6-1}}={{n-3}\choose{5}},\ n\geq8.
\end{gather*}
Because there are maximally four $Y$-graphs, group $H_5(C_n(K_4),\mathbb{Z})$ is zero.

\subsubsection{General wheel graphs}\label{sec:other-wheel}
In Table \ref{morse_wheel} we list Betti numbers of configuration spaces of wheel graphs of order $5,\ 6$ and $7$ that were calculated using the discrete Morse theory.
\begin{table}[ht]
\centering
\begin{tabular}{c|c|c|c|c}
$\Gamma$ & $n$ & $\beta_2(D_n(\Gamma))$ & $\beta_3(D_n(\Gamma))$ & $\beta_4(D_n(\Gamma))$ \\
\hline \hline
\multirow{4}{*}{$W_5$} & 3 & 8 & 0  & -\\
				& 4 & 22 & 0 & 0  \\
				& 5 & 34 & 4 & 0  \\
				& 6 & 46 & 30 & 0 \\
				& 7 & 58 & 90 & 0 \\
				& 8 & 70 & 196 & 13 \\
\hline
\multirow{4}{*}{$W_6$} & 3 & 15 & 0 & - \\
				& 4 & 40 & 0 & 0 \\
				& 5 & 60 & 15 & 0 \\
				& 6 & 80 & 90 & 0 \\
				& 7 & 100 & 250 & 5 \\
\hline
\multirow{4}{*}{$W_7$} & 3 & 24 & 0 & - \\
				& 4 & 63 & 0 & 0 \\
				& 5 & 93 & 36 & 0 \\
				& 6 & 123 & 197 & 0  \\
				& 7 & 153 & 527 & 24 \\ 
\hline
\end{tabular}
\caption{Betti numbers of configuration spaces for chosen wheel graphs computed using the discrete Morse theory. In all cases the calculated groups were torsion-free.}
\label{morse_wheel}
\end{table}

\paragraph*{Second homology}\ Since there are no pairs of disjoint $O$-cycles in wheel graphs, we have 
\[\beta_2(D_2(W_m))=0.\]
When $n=3$, all product cycles are the $O\times Y$-cycles. Their number is $(m-1)(m-3)$, because there are $m-1$ choices of $Y$-subgraphs and $m-3$ cycles that are disjoint with a fixed $Y$-subgraph. Hence,
\[\beta_2(D_3(W_m))=(m-1)(m-3).\]
When $n=4$, we have to count the $Y\times Y$ cycles in. Let us divide the $Y\times Y$ cycles into two groups: i) cycles, where one of the subgraphs is $Y_h$ and ii) cycles, where both subgraphs lie on the perimeter. There are no relations between the cycles within group i) and no relations between the cycles within group ii). However, there are some relations between the cycles of type i) and type ii). The relations occur between cycles $Y_h\times Y$ and $Y'\times Y$ when subgraphs $Y_h$ and $Y$ do not share any edges of the graph (like on Fig. \ref{wheel-rel}b)). Then, as on Fig. \ref{rel_cycle_y}, cycles $c_{Y_h}$ and $c_{Y'}$ are in the same homology class in $D_2(W_m-Y)$, because they share the same $O$-cycle and they are connected by a path that is disjoint with $Y$. Therefore, by multiplying the relation by $c_{Y}$ we get that 
\[c_{Y_h}\times c_Y\sim c_{Y'}\times c_Y.\]
 \begin{figure}[ht]
\centering
\includegraphics[width=0.7\textwidth]{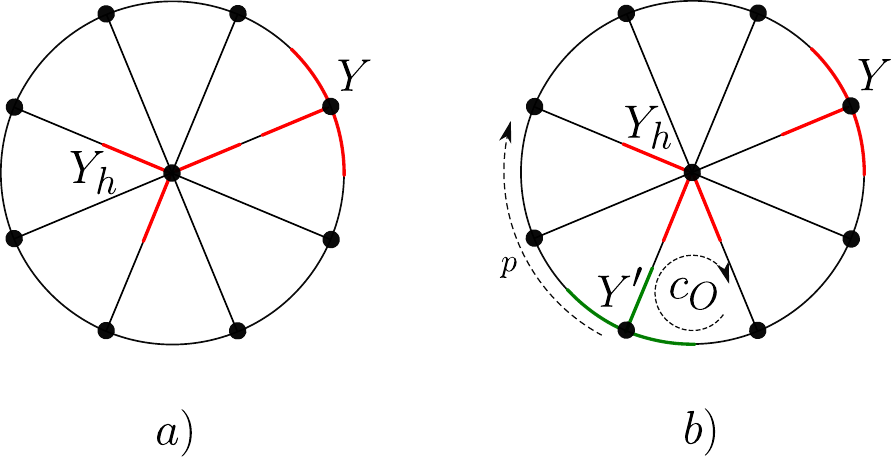}
\caption{Relations between different pairs of $Y_h\times Y$-cycles in a wheel graph. a) Cycles, where $Y_h$ and $Y$ share an edge of the graph are independent. b) Cycle, where $Y_h$ and $Y$ do not share any edges is in the same homology class as cycle $Y'\times Y$.}
\label{wheel-rel}
\end{figure}
If $m>4$, then for every pair $Y\times Y_h$ that does not share an edge, one can find subgraph $Y'$ on the perimeter which gives rise to such a relation. There are ${{m-1}\choose{2}}$ tori coming from $Y$-subgraphs from the perimeter. For a fixed $Y$-subgraph, the contribution from $Y\times Y_h$-cycles turns out to be equal to the number of independent cycles in the fan graph which is formed by removing subgraph $Y$ from the wheel graph \cite{HKRS}. This number is equal to $m-3$. Hence,
\[\beta_2(D_4(W_m))=2(m-1)(m-3)+{{m-1}\choose{2}}=\frac{(m-1)(5m-14)}{2}.\]

For numbers of particles greater than $4$, we have to take into account the distribution of free particles. Removing two $Y$-subgraphs from the perimeter may result with the decomposition of the wheel graph into at most two components. This happens iff two neighbouring $Y$-subgraphs have been removed. The number of nonequivalent ways of distributing the particles is $n-3$. The number of ways one can choose two neighbouring $Y$-subgraphs from the perimeter is $m-1$. This gives us the contribution of $(n-3)(m-1)$. Furthermore, removing a $Y$-subgraph from the hub and a subgraph from the perimeter always yields two nonequivalent ways of distributing the free particles. The first one being the edge $e$ joining the hub and the central vertex of $Y$, the second one being the remaining part of the graph, i.e. $W_m-(Y\sqcup Y_h\sqcup e)$. The contribution is $(n-3)(m-1)(m-3)$. Adding the contribution from $O\times Y$-cycles and from non-neighbouring $Y_p\times Y_p$-cycles, we get that the final formula for the second Betti number reads
\[\beta_2(D_n(W_m))=(n-2)(m-1)(m-3)+(m-1)(n-4)+{{m-1}\choose{2}},\ n\geq 4.\]

\paragraph*{Higher homologies}\ In computing the higher homology groups, we proceed in a similar fashion as in the previous section. However, the combinatorics becomes more complicated and in most cases it is difficult to write a single formula that works for all wheel graphs. Let us start with an example of $H_3(D_n(W_5))$. The possible types of product cycles are $O\times Y\times Y'$ and $Y\times Y'\times Y''$. Cycles of the first type arise in $W_5$ only when graphs $Y$ and $Y'$ are neighbouring subgraphs from the perimeter. There are four possibilities for such a choice of $Y$-subgraphs, hence 
\[\beta_3(D_5(W_5))=4.\]
When $n>5$, the free particles can be placed either on the edge joining the $Y$-subgraphs or on the connected part of $W_5$ that is created by removing subgraphs $Y$ and $Y'$. By arguments analogous to the ones presented in section \ref{sec:K4}, the distribution of free particles on the connected component containing cycle $O$ does not play a role. Hence, the contribution to $\beta_3$ is equal to the number of different distributions of free particles on the edge connecting $Y$ and $Y'$ and on the connected component. In other words, there are two bins and $n-5$ free particles. Hence, the total contribution from $O\times Y\times Y'$-cycles is $4(n-4)$. We split the contribution from $Y\times Y'\times Y''$-cycles into two groups. The first group consists of cycles only from perimeter ($Y_p\times Y'_p\times Y''_p$), for whom the combinatorial description is straightforward. The number of possible choices of $Y$-subgraphs is $4\choose{3}$ and it always results with the decomposition of $W_5$ into $3$ components. Hence, with $n-6$ free particles the number of independent $Y_p\times Y'_p\times Y''_p$-cycles is $4{{n-4}\choose{2}}$. In order to determine the number of independent cycles $Y_p\times Y'_p\times Y_h$ (two subgraphs from the perimeter and one from the hub), one has to consider different graphs that arise after removing two $Y$-subgraphs from the perimeter of $W_5$. The number of independent $Y_h$-cycles for a fixed choice of $Y_p$ and $Y_p'$ is the same as in a certain fan graph which is determined by the choice of the $Y_p$-subgraphs. Choosing $Y_p$ and $Y_p'$ to lie on the opposite sides of the diagonal of $W_5$, the resulting fan graph is the star graph $S_4$. The free  particles outside $Y_p$ and $Y_p'$ can always be moved to the $S_4$-subgraph. Hence, the contribution from such cycles is given by the number of independent $Y$-cycles in $S_4$ for $n-4$ particles. We denote this number by $\beta_1^{(n-4)}(S_4)$. The last group of cycles that we have to take into account are $Y_p\times Y'_p\times Y_h$, where $Y_p$ and $Y_p'$ are neighbouring subgraphs. The resulting fan graph is shown on Fig. \ref{w5-fan}. The $n-4$ particles that do not exchange on the perimeter subgraphs are distributed between the fan graph and the edge joining $Y_p$ and $Y_p'$. There have to be at least $2$ particles exchanging on a $Y_h$-subgraph of the fan graph. The number of independent $Y_h$-cycles for $k+2$ particles on the fan graph is given in the caption under Fig. \ref{w5-fan}. After summing all the above contributions, the final formula for the third Betti number reads
\begin{gather*}
\beta_3(D_n(W_5))=4(n-4)+4{{n-4}\choose{2}}+2\beta_1^{(n-4)}(S_4)+\\+4\sum_{k=0}^{n-6}\left(\beta_1^{(k+2)}(S_3)+{{k+3}\choose{k+1}}-1\right).
\end{gather*}
 \begin{figure}[ht]
\centering
\includegraphics[width=0.5\textwidth]{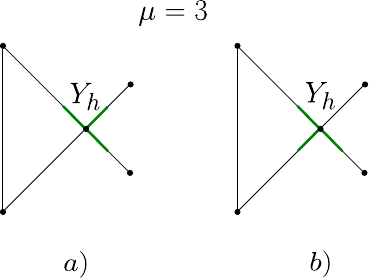}
\caption{The fan graph that is created after removing two neighbouring $Y$-subgraphs from the perimeter of $W_5$. It has $\mu=3$ leaves. There are two types of $Y$-cycles at the hub: a) cycles, where the $Y$-graph is spanned in three different leaves - the number of such cycles for $k+2$ particles is $\beta_1^{(k+2)}(S_3)$, b) cycles, where the $Y$-graph is spanned in two different leaves - the number of such cycles for $k+2$ particles is ${{k+3}\choose{k+1}}-1$, see \cite{HKRS}. }
\label{w5-fan}
\end{figure}
The fourth Betti number is easier to compute, because removing three $Y_p$-subgraphs always results with the same type of fan graph. This fan graph has no cycles, hence there are no $O\times Y\times Y\times Y$-cycles. Moreover, there is only one possible choice of four $Y$-subgraphs from the perimeter. This always results with the decomposition of $W_5$ into $5$ components. Choosing three $Y$-subgraphs from perimeter results with the decomposition of $W_5$ into $3$ components: a fan graph and $2$ edges. The number of independent $Y_h$ cycles in the fan graph is the same as in $S_4$. Taking into account the distribution of $n-6$ particles between the two edges and the fan graph, we have
\[\beta_4(D_n(W_5))={{n-4}\choose{4}}+4\sum_{k=0}^{n-8}(n-k-7)\beta_1^{(k+2)}(S_4),\ n\geq 8.\]
The top homology for $D_n(W_5)$ is $H_5$. Distributing $k+2$ particles on the central $S_4$ graph and the remaining particles on four free edges joining $Y_p$-subgraphs, we get
\[\beta_5(D_n(W_5))=\sum_{k=0}^{n-10}{{n-k-7}\choose{3}}\beta_1^{(k+2)}(S_4),\ n\geq 10.\]

Let us next generalise the above procedure to an arbitrary wheel graph $W_m$. The $d$th Betti number is zero whenever the number of particles is less than $2(d-1)+1=2d-1$. If $n=2d-1$ the only possible tori come from the products of $d-1$ $Y$-cycles and one $O$-cycle. The graph also cannot be too small, i.e. the condition $m-3\geq d-1$ must be satisfied. Otherwise, there is no cycle that is disjoint with $d-1$ $Y$-subgraphs. Hence,
\[\beta_d(D_{n}(W_m))=0\ {\rm if}\ n<2d-1\]
and 
\[\beta_d(D_{2d-1}(W_m))=0\ {\rm if}\ m<d+2.\]
Otherwise, for $n=2d-1$, if the graph is large enough, one has to look at all the possibilities of removing $Y$-subgraphs from the perimeter and what fan graphs are created. We are interested in the number of leaves ($\mu$) of the resulting fan graph. The number of cycles in such a fan graph with $\mu$ leaves is $m-1-\mu$. It is a difficult task to list all possible fan graphs for any $W_m$ in a single formula. The results for graphs up to $W_7$ are shown in Table \ref{tab:wheel-fan}. Using the notation from Table \ref{tab:wheel-fan}, the general formula for $\beta_d$ reads 
\[\beta_d(D_{2d-1}(W_m))=\sum_{{\bf n}:|{\bf n}|=d-1} N_{\bf n}(m-1-\mu_{\bf n}),\]
where $|{\bf n}|:=\sum_{i=1}^l n_i$.

\begin{table}[ht]
\centering
\begin{tabular}{c|c|c|c}
\multirow{2}{*}{$\Gamma$} & Groups of & Number of & Number of \\
					& $Y$-subgraphs - $\bf n$ & possible choices - $N_{\bf n}$ & leaves - $\mu_{\bf n}$ \\
\hline \hline
\multirow{5}{*}{$W_5$} & (1) & 4 & 1  \\
				& (1,1) & 2 & 4  \\
				& (2) & 4 & 3  \\
				& (3) & 4 & 4 \\
				& (4) & 1 & 4 \\
\hline
\multirow{7}{*}{$W_6$} & (1) & 5 & 2  \\
				& (1,1) & 5 & 4  \\
				& (2) & 5 & 3  \\
				& (2,1) & 5 & 5 \\
				& (3) & 5 & 4 \\
				& (4) & 5 & 5 \\
				& (5) & 1 & 5 \\
\hline
\multirow{11}{*}{$W_7$} & (1) & 6 & 2 \\
				& (1,1) & 9 & 4 \\
				& (2) & 6 & 3 \\
				& (1,1,1) & 2 & 6 \\
				& (2,1) & 12 & 5 \\
				& (3) & 6 & 4 \\
				& (2,2) & 3 & 6 \\
				& (3,1) & 6 & 6 \\
				& (4) & 6 & 5 \\
				& (5) & 6 & 6 \\
				& (6) & 1 & 6 \\
\hline
\end{tabular}
\caption{The possibilities of choosing a number of $Y$-subgraphs from the perimeter of a wheel graph. The groups of $Y$-subgraphs are denoted by sequences $(n_1,n_2,\dots,n_l)$, where $l+\sum_{i=1}^l n_i\leq m-1$. A group $n_i$ means that $n_i$ neighbouring $Y$-subgraphs were chosen. The groups have to be separated by at least one spoke. For a fixed set of groups there are many possibilities for distributing the remaining $Y$-subgraphs. The number of possibilities is written in the third column. The number of leaves of the resulting fan graph is written in the fourth column. It is independent on the distribution of the remaining $Y$-subgraphs and is given by $\mu_{\bf n}=\min\left(m-1,l+\sum_{i=1}^l n_i\right)$.}
\label{tab:wheel-fan}
\end{table}

\noindent For higher numbers of particles, one has to take into account the $Y\times Y\times\dots\times Y$ cycles and distribution of free particles. If $n=2d$, the free particles are only in $O\times Y\times Y\times\dots\times Y$-cycles, where they are distributed between the edges that come from removing a group of $Y$-subgraphs. Group $n_i$ gives $n_i-1$ edges. Hence, groups $(n_1,\dots,n_l)$ give $|{\bf n}|-l$ edges. The final formula reads 
\begin{gather*}
\beta_d(D_{2d}(W_m))={{m-1}\choose{d}}+\\+\sum_{{\bf n}:|{\bf n}|=d-1} N_{\bf n}\left((m-1-\mu_{\bf n})(d-1-\#{\bf n})+\beta_1^{(2)}\left(S_{\mu_{\bf n}}\right)+(\mu_{\bf n}-1)(m-1-\mu_{\bf n})\right),
\end{gather*}
where $\#{\bf n}$ is the number of groups in ${\bf n}$ (the length of vector $\bf n$). The contribution $\beta_1^{(2)}\left(S_{\mu_{\bf n}}\right)+(\mu_{\bf n}-1)(m-1-\mu_{\bf n})$ comes from the number of independent $Y_h$-cycles in the relevant fan graph. The general formula when $n>2d$ reads as follows.
\begin{gather}\label{wheel-full}
\beta_d(D_{n}(W_m))=\sum_{{\bf n}:|{\bf n}|=d-1} N_{\bf n}(m-1-\mu_{\bf n}){{n-d-\#{\bf n}}\choose{d-\#{\bf n}-1}}+\\ \nonumber+\sum_{{\bf n}:|{\bf n}|=d} N_{\bf n}{{n-d-\#{\bf n}}\choose{d-\#{\bf n}}}+ \\ \nonumber +
\sum_{{\bf n}:|{\bf n}|=d-1} N_{\bf n}\sum_{l=0}^{n-2d}\bigg{(}\beta_1^{(l+2)}\left(S_{\mu_{\bf n}}\right)+\left({{l+\mu_{\bf n}}\choose{l+1}}-1\right)(m-1-\mu_{\bf n})\bigg{)}\times \\ \nonumber \times {{n-d-\#{\bf n}-l-2}\choose{d-\#{\bf n}-2}}.
\end{gather}
The first sum describes the $O\times Y\times Y\times\dots\times Y$-cycles and the distribution of the free $n-2d+1$ particles. Second sum is the number of $Y\times Y\times\dots\times Y$-cycles, where all $Y$-subgraphs lie on the perimeter - there are $n-2d$ free particles. The last sum describes the number of independent $Y_h\times Y_p\times\dots\times Y_p$-cycles. Here we used the formula for the number of $Y_h$-cycles for $n$ particles on a fan graph with $\mu$ leaves and $m-1$ spokes \cite{HKRS}
\[n_Y^{(n)}(\mu,m-1)=\beta_1^{(n)}\left(S_{\mu_{\bf n}}\right)+\left({{n+\mu-2}\choose{n-1}}-1\right)(m-1-\mu).\]
Sometimes, in formula (\ref{wheel-full}), we get to evaluate ${{0}\choose{0}}=1$, ${{0}\choose{-1}}=0$, ${{-1}\choose{-1}}=1$.

The highest non-vanishing Betti number is $\beta_m$ and its value is the number of the possible distributions of $n-2m$ free particles between the central $S_m$ graph and the free $m-2$ edges on the perimeter.
\[\beta_m(D_n(W_m))=\sum_{k=0}^{n-2m}{{n-m-k-2}\choose{m-2}}\beta_1^{(k+2)}(S_{m-1}),\ n\geq 2m.\]

\subsection{Wheel graphs via \'{S}wi\k{a}tkowski discrete model}\label{sec:wheel-swiatkowski}
In this section we show that the homology of configuration spaces of wheel graphs is generated by product cycles. The strategy is to consider two consecutive vertex cuts that bring any wheel graph to the form of a linear tree.
 \begin{figure}[H]
\centering
\includegraphics[width=0.5\textwidth]{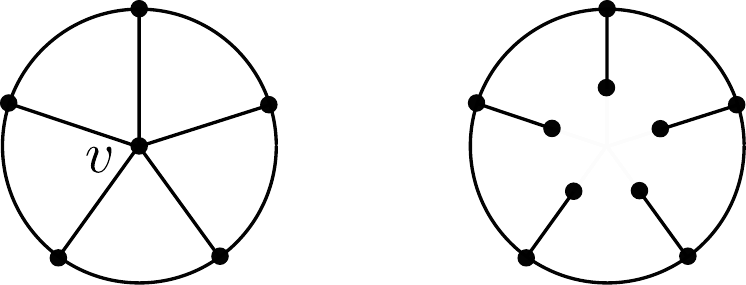}
\caption{Vertex blowup at the hub of wheel $W_{m+1}$ resulting with net graph $N_{m}$.}
\label{blowup-wheel1}
\end{figure}

 \begin{figure}[H]
\centering
\includegraphics[width=0.7\textwidth]{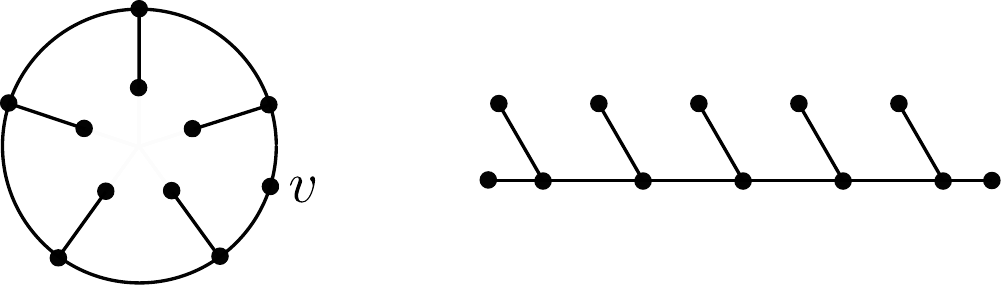}
\caption{Blowup of a vertex in net graph $N_{m}$ resulting with linear tree graph $T_{m}$.}
\label{blowup-wheel2}
\end{figure}

Throughout, we use the knowledge of generators of the homology groups for tree graphs to construct a set of generators for net graphs and wheel graphs. Translating the results of paper \cite{MS17} to the \'{S}wi\k{a}tkowski complex, we have that the generators of $H_d(S(T_{m}))$ are of the form 
\[c_d=c_{Y_1}\dots c_{Y_d}v_1\dots v_ke_1^{n_1}\dots e_l^{n_l},\]
subject to relations
\begin{equation}\label{tree-rel}
c_de\sim c_dv,\ {\rm if\ }e\cap v\neq\emptyset.
\end{equation}
This means that computing the rank od $H_d(S(T_{m}))$ boils down to considering all possible distributions of $n-2d$ free particles among the connected components of $T_m-(v_h(Y_1)\cup\dots\cup v_h(Y_d))$. By $v_h(Y_1)$ we denote the hub vertex of the $Y$-subgraph $Y_i$. Hence, $H_d(S(T_{m}))$ is freely generated by generators of the form
\begin{equation}\label{tree-basis}
[Y_1,\dots,Y_d,n_1,\dots,n_{2d+1}],\ n_1+\dots +n_{2d+1}=n-2d,
\end{equation}
where $n_i$ is the number of particles on $i$th connected component of $T_m-(v_h(Y_1)\cup\dots\cup v_h(Y_d))$. In the first step, we connect two endpoints of $T_{m}$ to obtain net graph $N_{m}$ (Fig. \ref{blowup-wheel2}).
\begin{lemma}\label{lemma:net}
The homology groups of $C_n(N_m)$ are freely generated by the product $Y$-cycles and the distributions of free particles on the connected components $N_m-(v_h(Y_1)\cup\dots\cup v_h(Y_d))$ which we denote by
\begin{equation}\label{net-basis}
[Y_1,\dots,Y_d,n_1,\dots,n_{2d}],\ n_1+\dots +n_{2d}=n-2d.
\end{equation}
The Betti numbers read
\[\beta_d(C_n(N_m))={{m}\choose{d}}{{n-1}\choose{2d-1}}.\]
\end{lemma}
\begin{proof}
Long exact sequence corresponding to vertex blow-up from figure \ref{blowup-wheel2} reads
\begin{gather*}
\dots\xrightarrow{\Psi_{n,d+1}} H_d\left(S_{n-1}(T_m)\right)\xrightarrow{\delta_{n,d}}H_d\left(S_n(T_m)\right)\xrightarrow{\Phi_{n,d}}H_d\left(\tilde S^v_n(N_m)\right)\xrightarrow{\Psi_{n,d}} \\ \nonumber \xrightarrow{\Psi_{n,d}}H_{d-1}\left(S_{n-1}(T_m)\right)\xrightarrow{\delta_{n,d-1}}H_{d-1}\left(S_n(T_m)\right)\xrightarrow{\Phi_{n,d-1}}\dots,
\end{gather*}
Let us next show that the connecting homomorphism $\delta$ is in this case injective. Map $\delta_{n,d}$ acts on generators (\ref{tree-basis}) as
\begin{gather*}
\delta_{n,d}([Y_1,\dots,Y_d,n_1,\dots,n_{2d+1}])=[Y_1,\dots,Y_d,n_1+1,\dots,n_{2d+1}]+ \\ -[Y_1,\dots,Y_d,n_1,\dots,n_{2d+1}+1],
\end{gather*}
where $n_1$ and $n_{d+1}$ are respectively the numbers of particles on the leftmost and on the rightmost connected component of $T_m-(v_h(Y_1)\cup\dots\cup v_h(Y_d))$. One can check that vectors $\{[Y_1,\dots,Y_d,n_1+1,\dots,n_{2d+1}] -[Y_1,\dots,Y_d,n_1,\dots,n_{2d+1}+1]\}$ corresponding to different choices of $Y$-subgraphs of $T_m$ are linearly independent. Hence, any vector from $\im \delta_{n,d}$ can be uniquely decomposed in this basis and its preimage can be unambiguously determined by subtracting the particles from $n_1$ and $n_{d+1}$. By injectivity of $\delta$, 
\[H_d\left(\tilde S^v_n(N_m)\right)\cong\coker(\delta_{n,d}).\]
Hence, the rank of $H_d(S_n(N_m))$ is equal to $\rk({\coker}_{n,d})=\beta_d\left(S_n(T_m)\right)-\beta_d\left(S_{n-1}(T_m)\right)$. The Betti numbers of $S_n(T_m)$ can be computed by counting the distributions of $n-2d$ particles on $2d+1$ connected components multiplied by the number of $d$-subsets of $Y$-subgraphs of $T_m$. The result is 
\[\beta_d\left(S_n(T_m)\right)={{m}\choose{d}}{{n}\choose{2d}}.\]
The claim of the lemma follows directly from the above formula. The result is the same as the number of distributions of $n-2d$ particles on $2d$ connected components of $N_m-(v_h(Y_1)\cup\dots\cup v_h(Y_d))$.
\end{proof}
Let us next consider the homology sequence associated with the vertex blow-up from $W_{m+1}$ to $N_m$ (fig. \ref{blowup-wheel1}).
\begin{gather*}
\dots\xrightarrow{\Psi_{n,d+1}} \bigoplus_{h\in H(v)-\{h_0\}}H_d\left(S_{n-1}(N_m)\right)\xrightarrow{\delta_{n,d}}H_d\left(S_n(N_m)\right)\xrightarrow{\Phi_{n,d}}H_d\left(\tilde S^v_n(W_{m+1})\right)\xrightarrow{\Psi_{n,d}} \\ \nonumber \xrightarrow{\Psi_{n,d}}\bigoplus_{h\in H(v)-\{h_0\}}H_{d-1}\left(S_{n-1}(N_m)\right)\xrightarrow{\delta_{n,d-1}}H_{d-1}\left(S_n(N_m)\right)\xrightarrow{\Phi_{n,d-1}}\dots,
\end{gather*}
We next describe the kernel of map $\delta$. Our aim is to show that it is free abelian which in turn gives us that the short exact sequences for $H_d(S_n(W_{m+1}))$ split and yield $H_d(S_n(W_{m+1}))\cong \coker(\delta_{n,d})\oplus\ker(\delta_{n-1,d})$. Map $\delta_{n,d}$ assigns to generators (\ref{net-basis}) of $H_d(S_n(W_{m+1}))$ the differences of generators derived from a given generator by adding one particle to a connected component of $N_m-(v_h(Y_1)\cup\dots\cup v_h(Y_d))$. In order to write down the action of map $\delta$, let us first establish some notation. The connected components of $N_m-(v_h(Y_1)\cup\dots\cup v_h(Y_d))$ are either isomorphic to edges or to linear tree graphs. The number of connected components that are edges which have one vertex of degree one in $N_m$ is equal to $d$. The number of the remaining connected components is always equal to $d$, but their type depends on the distribution of subgraphs $Y_1,\dots,Y_d$ in $N_m$. The situations that are relevant for the description of $\ker\delta$ are those, where a particle is added by map $\delta$ to two connected components which contain an edge which before the blow-up was adjacent to the hub of $W_{m+1}$. There are at most $2d$ such components, as removing the hub-vertices of two neighbouring $Y$-subgraphs of $N_m$ yields a connected component of the edge type which is not adjacent to the hub of $W_{m+1}$. We label these components by numbers $1,\dots,l$ (we always have $d\leq l\leq 2d$) and the occupation numbers of these components are $n_1,\dots,n_l$. We choose component $1$ to be the component adjacent to edge $e(h_0)$ and increase the labels in the clockwise direction from the component with label $1$. The remaining components are labelled by numbers $l+1,\dots,2d$. Map $\delta$ acts on basis elements of $\bigoplus_{h\in H(v)-\{h_0\}}H_d\left(S_{n-1}(N_m)\right)$ as follows. 
\begin{gather*}
\delta_{n,d}\left([Y_1,\dots,Y_d,n_1,\dots,n_{p_i},\dots,n_{2d}]_i\right)= \\ =[Y_1,\dots,Y_d,n_1+1,\dots,n_{p_i},\dots,n_{2d}]-[Y_1,\dots,Y_d,n_1,\dots,n_{p_i}+1,\dots,n_{2d}],
\end{gather*}
where $[Y_1,\dots,Y_d,n_1,\dots,n_{p_i},\dots,n_{2d}]_i$ is a generator corresponding to the $h_i$-component of $\bigoplus_{h\in H(v)-\{h_0\}}H_d\left(S_{n-1}(N_m)\right)$ and $p_i$ is the label of the connected component which contains edge $e(h_i)$. One easily observes that if $p_i=1$, i.e. edges $e(h_0)$ and $e(h_i)$ belong to the same connected component of  $N_m-(v_h(Y_1)\cup\dots\cup v_h(Y_d))$, generator $[Y_1,\dots,Y_d,n_1,\dots,n_{p_i},\dots,n_{2d}]_i\in \ker\delta_{n,d}$. Such an element is represented in $S_n(W_{m+1})$ as cycle $c_{Y_1}\dots c_{Y_d}c_O$, where $O$ is the cycle in $W_{m+1}$ which contains edges $e(h_0)$, $e(h_i)$ and the hub of $W_{m+1}$. Similarly, element \[[Y_1,\dots,Y_d,n_1,\dots,n_{p_i},\dots,n_{2d}]_i-[Y_1,\dots,Y_d,n_1,\dots,n_{p_j},\dots,n_{2d}]_j\]
is in the kernel of $\delta_{n,d}$ whenever edges $e(h_i)$ and $e(h_j)$ belong to the same connected component of  $N_m-(v_h(Y_1)\cup\dots\cup v_h(Y_d))$. The last type of elements of $\ker\delta_{n,d}$ are combinations of generators that when acted upon by $\delta_{n,d}$, compose to the boundary of a $Y$-cycle centred at the hub of $W_{m+1}$. Such kernel elements correspond to cycles $c_{Y_1}\dots c_{Y_d}c_{Y_h}$ in $S_n(W_{m+1})$, where $Y_h$ is a $Y$-cycle, whose hub-vertex is the hub-vertex of $W_{m+1}$. The precise form of such kernel elements is the following.
\begin{gather*}
[Y_1,\dots,Y_d,n_1,\dots,n_{p_i},\dots,n_{p_j}+1,\dots]_i-[Y_1,\dots,Y_d,n_1+1,\dots,n_{p_i},\dots,n_{p_j},\dots]_i+\\+[Y_1,\dots,Y_d,n_1+1,\dots,n_{p_i},\dots,n_{p_j},\dots]_j-[Y_1,\dots,Y_d,n_1,\dots,n_{p_i}+1,\dots,n_{p_j},\dots]_j,
\end{gather*}
where $i<j$. In order to manage the relations between the above kernel elements, we use the already mentioned fact that they are in a one-to-one correspondence with $1$-cycles ($O$-cycles and $Y$-cycles) in a configuration space of the disconnected graph $W_{m+1}-(v_h(Y_1)\cup\dots\cup v_h(Y_d))$. More specifically, the disconnected graph $W_{m+1}-(v_h(Y_1)\cup\dots\cup v_h(Y_d))$\footnote{$W_{m+1}-(v_h(Y_1)\cup\dots\cup v_h(Y_d))$ is a disconnected topological space. We give this space the structure of a graph by adding a vertex to the open end of each open edge.} is a disjoint sum of a number of edges and of one fan graph. We regard the $1$-cycles ($O$-cycles or $Y$-cycles) at the hub as generators of the first homology group of the configuration space of the fan graph multiplied by different distributions of particles on the disjoint edge-components of $W_{m+1}-(v_h(Y_1)\cup\dots\cup v_h(Y_d))$. Fan graphs are planar, hence by equation (\ref{eq:1st-hom}) there is no torsion in $\ker\delta_{n,d}$. Hence, $H_d(S_n(W_{m+1}))$ is torsion-free and short exact sequence for $H_d(S_n(W_{m+1}))$ gives in this case
\[\beta_d(S_n(W_{m+1}))\cong \beta_d(S_{n}(N_{m}))-m \beta_d(S_{n-1}(N_{m}))+\rk(\ker\delta_{n,d})+\rk(\ker\delta_{n,d-1}).\]
The computation of ranks of kernels of maps $\delta_{n,d}$ is a combinatorial task which has been accomplished using the correspondence with cycles in configuration spaces of fan graphs in subsection \ref{sec:other-wheel}.

\subsection{Graph $K_{3,3}$}\label{sec:k33}
Graph $K_{3,3}$ is shown on Fig. \ref{K33}. We will draw graph $K_{3,3}$ in two ways: 1) immersion in $\RR^2$, Fig. \ref{K33}a), ii) embedding in $\RR^3$, Fig. \ref{K33}b).
 \begin{figure}[H]
\centering
\includegraphics[width=0.7\textwidth]{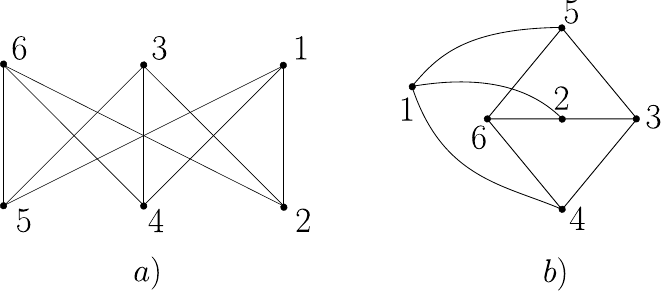}
\caption{Graph $K_{3,3}$.}
\label{K33}
\end{figure} 
Graph $K_{3,3}$ has the property that all its vertices are of degree three. High homology groups of graphs with such a property have been studied in \cite{Knudsen}. In particular, we have the following result.
\begin{theorem}
Let $\Gamma$ be a simple graph, whose all vertices have degree $3$. Denote by $N$ the number of vertices of graph $\Gamma$ and label the vertices by labels $1,\dots, N$. Moreover, denote by $\mc{Y}=\{Y_1,\dots Y_N\}$ the set of $Y$-subgraphs of $\Gamma$ such that the hub of $Y_k$ is vertex $k$. Group $H_N(S_n(\Gamma))$ is freely generated by product cycles
\[e_1^{n_1}\dots e_K^{n_K}\bigotimes_{Y\in\mc{Y}}c_Y,\ n_1+\dots+e_K=n-2N.\]
Group $H_{N-1}(S_n(\Gamma))$ is generated by product cycles of the form
\[e_1^{n_1}\dots e_K^{n_K}v\bigotimes_{Y\in\mc{\tilde Y}}c_Y,\ n_1+\dots+e_K=n-2(N-1),\]
where $\mc{\tilde Y}\subset \mc{Y}$  is such that $|\mc{\tilde Y}|=N-1$, and $v\in V(\Gamma)$ is the unique vertex that satisfies $v\cap(\cup_{Y\in\mc{\tilde Y}}Y)=\emptyset$.
The above generators are subject to relations
\[e_1^{n_1}\dots e_j^{n_j}\dots e_K^{n_K}v\bigotimes_{Y\in\mc{\tilde Y}}c_Y\sim e_1^{n_1}\dots e_j^{n_j+1}\dots e_K^{n_K}\bigotimes_{Y\in\mc{\tilde Y}}c_Y,\]
whenever $e_j\cap v\neq\emptyset$.
\end{theorem}
As we show in section \ref{sec:2nd-product}, the second homology group of configuration spaces of such graphs is also generated by product cycles. Later in this section, by comparing the ranks of homology groups computed via the discrete Morse theory, we argue that $H_4(C_n(K_{3,3}))$ is also generated by product cycles. Interestingly, in $H_3(C_n(K_{3,3}))$ there is a new non-product generator. Using this knowledge, we explain the relations between the product and non-product cycles that give the correct rank of $H_3(C_n(K_{3,3}))$.

\paragraph*{Second homology group}\ There are no pairs of disjoint cycles in $K_{3,3}$, hence the product part for $n=2$ is empty. When $n=3$, there are $12$ $O\times Y$-cycles. This can be seen by choosing the $Y$-graph centered at vertex $1$ on Fig. \ref{K33}b) - there are $2$ cycles disjoint with such a $Y$-subgraph. There are $6$ $Y$-subgraphs in $K_{3,3}$, hence we get the number of $O\times Y$-cycles. One checks by a straightforward calculation that $8$ of them are independent. Hence,
\[\beta_2(D_3(K_{3,3}))=8.\]
When $n=4$, there are new product cycles of the $Y\times Y$-type. There are ${6\choose 2}=15$ cycles of this type, however there are relations between them. Such relations between the $Y\times Y$-cycles arise when one of the cycles is in relation with a different $Y$-cycle. This happens only when we have a situation as on Fig. \ref{rel_cycle_y}. Therefore, cycles of the $Y\times Y$-type, where the hubs of the $Y$-subgraphs, are connected by an edge, are all independent (Fig. \ref{K33_subgraphs}a)). The number of such cycles is $9$. The relations occur between $Y\times Y$-cycles, where the hubs of the subgraphs are not connected by an edge (Fig. \ref{K33_subgraphs}b)). There are $6$ such cycles. The number of relations is $4$. To see this, consider $Y$-subgraph, whose hub is vertex $1$ (Fig. \ref{K33}). Denote this subgraph by $Y_1$. It is straightforward to see that in graph $K_{3,3}-Y_1$ we have $c_{Y_3}\sim c_{Y_6}$. Hence, 
\[(c_{Y_1}\otimes c_{Y_3})\sim (c_{Y_1}\otimes c_{Y_6}).\]
Analogous relations for $Y$-subgrphs that lie on the same side of the $K_{3,3}$ graph as $Y_1$ (see Fig. \ref{K33}a)) read
\[(c_{Y_3}\otimes c_{Y_1})\sim (c_{Y_3}\otimes c_{Y_6}),\ (c_{Y_6}\otimes c_{Y_3})\sim (c_{Y_6}\otimes c_{Y_1}).\]
From the above equations only two are independent. Similar situation happens for relations between pairs of graphs from the other side. The complete set of relations reads 
\[(c_{Y_1}\otimes c_{Y_3})\sim (c_{Y_1}\otimes c_{Y_6})\sim (c_{Y_3}\otimes c_{Y_6}),\ (c_{Y_2}\otimes c_{Y_4})\sim (c_{Y_2}\otimes c_{Y_5})\sim (c_{Y_4}\otimes c_{Y_5}).\]
Therefore,
\[\beta_2(D_4(K_{3,3}))=8+9+2=19.\]
For $n>4$, we have to take into account the distribution of free particles. Whenever two non-neighbouring $Y$-subgraphs are considered, all distributions of free particles are equivalent (Fig. \ref{K33_subgraphs}b)). When the subgraphs are adjacent, there are two different parts of $K_{3,3}$, where the particles can be distributed, see Fig. \ref{K33_subgraphs}a). This gives the formula
\[\beta_2(D_n(K_{3,3}))=8+2+9(n-3)=9n-17,\ n\geq 4.\]
 \begin{figure}[H]
\centering
\includegraphics[width=0.8\textwidth]{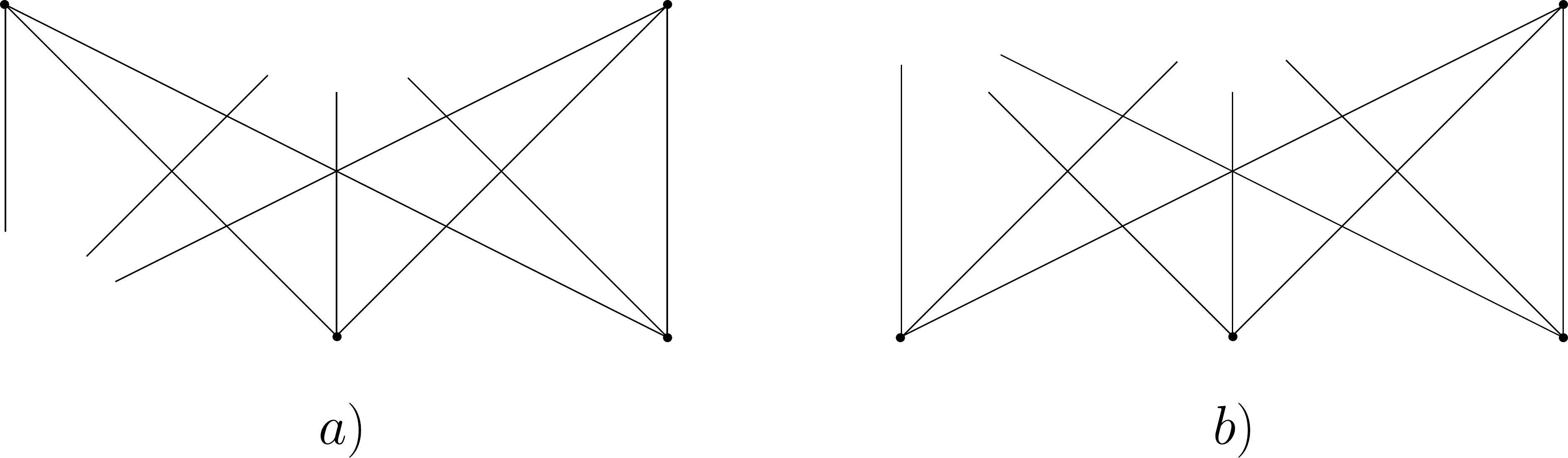}
\caption{Graph $K_{3,3}$ after removing two $Y$-subgraphs.}
\label{K33_subgraphs}
\end{figure}

\paragraph*{Higher homology groups}\ Let us first look at the third homology group. The are no product cycles for $n=4$ however, from the Morse theory for the subdivided graph from Fig. \ref{K33morse4p} we have  
\[\beta_3(D_4(K_{3,3}))=1.\]
 \begin{figure}[H]
\centering
\includegraphics[width=0.5\textwidth]{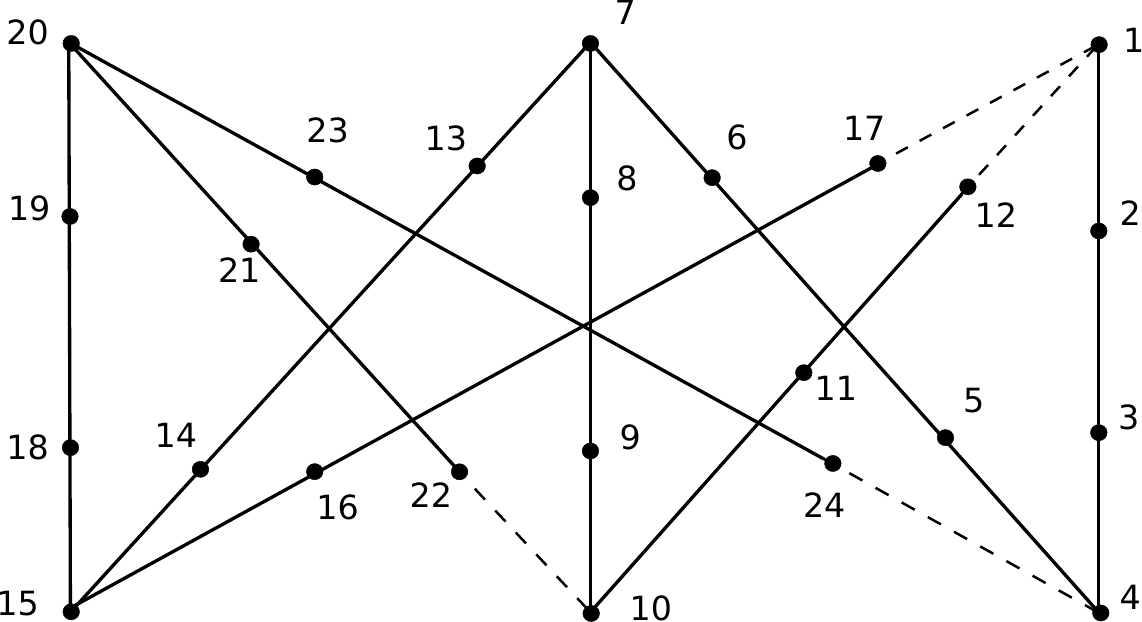}
\caption{Graph $K_{3,3}$ sufficiently subdivided for $n=4$. The deleted edges are marked with dashed lines.}
\label{K33morse4p}
\end{figure}
The Morse complex has dimension $3$. The generator of $H_3(D_4(K_{3,3}))$ is isomorphic to a closed $3$-manifold with Euler characteristic $\chi=11$. The cycle on the level of the Morse complex has the form
\begin{gather*}
c=\left\{e_1^{12},e_4^{24},e_{10}^{22},2\right\}-\left\{e_1^{17},e_4^{24},e_{10}^{22},2\right\}-\left\{e_1^{12},e_{10}^{22},e_{15}^{18},16\right\}+\left\{e_1^{17},e_4^{24},e_{10}^{22},11\right\}+\\\left\{e_1^{17},e_7^{13},e_{10}^{22},8\right\}+\left\{e_1^{12},e_4^{24},e_{15}^{18},16\right\}-\left\{e_4^{24},e_7^{13},e_{10}^{22},8\right\}-\left\{e_1^{12},e_4^{24},e_{10}^{22},5\right\}.
\end{gather*}
\noindent For $n=5$, we have the $Y\times Y\times O$-cycles. These are the cycles, where the $Y$-subgraphs are adjacent. For every pair of adjacent $Y$ subgraphs there is an unique $O$-cycle. An example of such a cycle is 
\[c_{Y_1}\times c_{Y_2}\times \left(\{e_3^4\}+\{e_4^6\}-\{e_5^6\}-\{e_3^5\}\right).\]
The number of all such cycles is equal to the number of pairs of adjacent $Y$-subgraphs which is $9$. Adding the properly embedded generator of $H_3(D_4(K_{3,3}))$, we get
\[\beta_3(D_5(K_{3,3}))=10.\]
For $n\geq 6$, all $Y\times Y\times Y$-cycles are independent. Consider two ways of choosing three $Y$-subgraphs. The first way is to remove two $Y$-graphs from the same side and one from the opposite side. This results with the partition of $K_{3,3}$ into three components (Fig. \ref{K33_subgraphs3}a)). Removing three $Y$-graphs from the same side splits $K_{3,3}$ into three parts (Fig. \ref{K33_subgraphs3}b)). Therefore, 
\[\beta_3(D_n(K_{3,3}))=1+9(n-4)+{6\choose 3}{{n-4}\choose{2}},\ n\geq 6.\]

 \begin{figure}[H]
\centering
\includegraphics[width=0.8\textwidth]{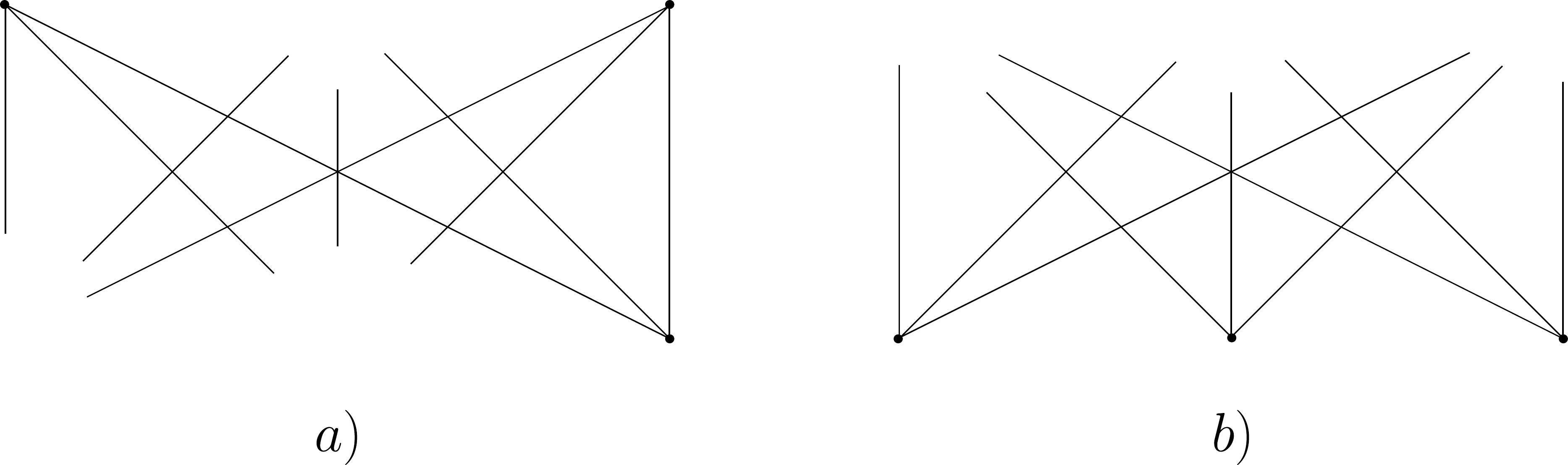}
\caption{Graph $K_{3,3}$ after removing three $Y$-subgraphs.}
\label{K33_subgraphs3}
\end{figure}

\noindent The product contribution to higher homology groups requires considering different choices of $Y$-subraphs. There are no $Y\times Y\times\dots\times Y\times O$-cycles in $H_p(D_n(K_{3,3}))$ for $p\geq 4$. As direct computations using discrete Morse theory show, there are also no non-product generators (see table \ref{morse_table}). Therefore, only $Y\times Y\times\dots\times Y$-cycles contribute to $H_p(D_n(K_{3,3}))$ for $p\geq 4$. Removing four $Y$-graphs from $K_{3,3}$ always results with the splitting into $5$ parts, removing five $Y$-graphs gives $7$ parts and removing all six $Y$-graphs gives $9$ parts. Summing up,
\begin{gather*}
\beta_4(D_n(K_{3,3}))={6\choose4}{{n-4}\choose{4}},\quad \beta_5(C_n(K_{3,3}))={6\choose5}{{n-4}\choose{6}}, \\
\beta_6(C_n(K_{3,3}))={{n-4}\choose{8}}. \\
\end{gather*}
All homology groups higher than $H_6$ are zero for any number of particles.

\subsection{Triple tori in $C_n(K_{2,p})$}\label{sec:K_2p}
In this section we study a family of graphs, where some cycles generating the homology groups of the $n$-particle configuration space are not product. This is the family of complete bipartite graphs $K_{2,p}$ (see figure \ref{fig:k2p}a). The first interesting graph from this family is  $K_{2,4}$. As we show below, its $3$-particle configuration space gives rise to a $2$-cycle which is a triple torus. It turns out that such triple tori together with products of $Y$ cycles generate the homology groups of $C_n(K_{2,p})$. The most convenient discrete model for studying $C_n(K_{2,p})$ is the \'{S}wi\k{a}tkowski model.
 \begin{figure}[ht]
\centering
\includegraphics[width=0.9\textwidth]{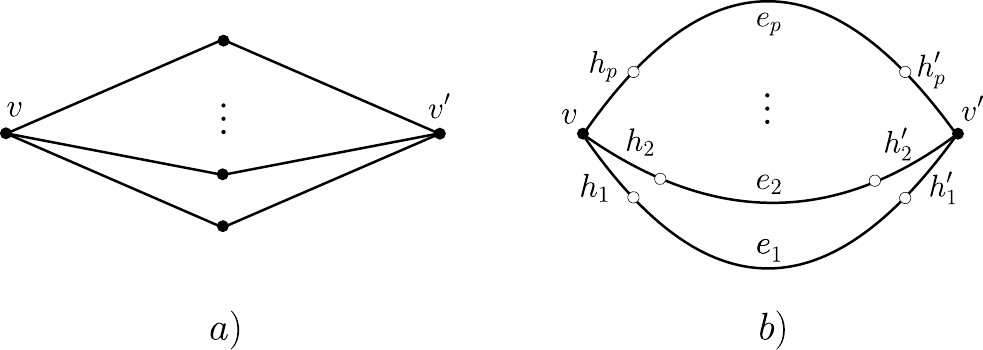}
\caption{a) Graph $K_{2,p}$. b) Graph $\Theta_p$.}
\label{fig:k2p}
\end{figure}
In fact, we study the \'{S}wi\k{a}tkowski configuration space of graph $\Theta_p$ (see \ref{fig:k2p}b) which is topologically equivalent to $K_{2,p}$, but it has the advantage that its discrete configuration space is of the optimal dimension. Because there are no $3$-cells in $S_n(\Theta_p)$, hence automatically we get that
\[H_i(C_n(K_{2,p}))=0\ {\rm for\ }i\geq 3.\]
This in turn means that $H_2(C_n(K_{2,p}))$ as the top homology group is a free group. The first homology group can be computed using the methods of papers \cite{HKRS,KoPark}.
\begin{lemma}
The first homology group of $C_n(K_{2,p})$ is equal to $\ZZ^{p(p-1)}$ for $n\geq 2$ and $p-1$ for $n=1$.
\end{lemma}
By counting the number of $0$-, $1$- and $2$-cells in $S_n(K_{2,p})$, we compute the Euler characteristic (see also \cite{Gal01}).
\begin{lemma}
The Euler characteristic of $S_n(K_{2,p})$ for $n\geq 3$ and $p\geq 3$ is 
\[\chi=(p-1)^2{{n-3+p}\choose{p-1}}-2(p-1){{n-2+p}\choose{p-1}}+{{n-1+p}\choose{p-1}}.\]
\end{lemma}
\noindent On the other hand, $\chi(S_n(K_{2,p}))=1-\beta_1(S_n(K_{2,p}))+\beta_2(S_n(K_{2,p}))$. Therefore, we compute the second Betti number of $C_n(K_{2,p})$ as
\begin{gather}\label{eq:h2k2p}
\beta_2(C_n(K_{2,p}))=(p-1)^2{{n-3+p}\choose{p-1}}-2(p-1){{n-2+p}\choose{p-1}}+\\ \nonumber+{{n-1+p}\choose{p-1}}+\frac{p(p-1)}{2}-1\ {\rm for\ }n\geq 3{\rm\ and\ } p\geq 3.
\end{gather}
In the remaining part of this section we describe the generators of $H_2(C_n(K_{2,p}))$ and the relations that lead to the above formula. We represent them in terms of $2$-cycles in $S_n(\Theta_p)$.
\begin{example}{\bf Generators of $H_2(S_n(\Theta_3))$.} 
Group $H_2(S_n(\Theta_3))$ is generated by products of $Y$-cycles at vertices $v$ and $v'$. More precisely, consider the following two $Y$-cycles
\begin{gather*}
c_{123}=e_1(h_2-h_3)+e_2(h_3-h_1)+e_3(h_1-h_2),\\
c'_{123}=e_1(h_2'-h_3')+e_2(h_3'-h_1')+e_3(h_1'-h_2').
\end{gather*}
Group $H_2(S_n(\Theta_3))$ is freely generated by cycles 
\[c_{123}c'_{123}e_1^{n_1}e_2^{n_2}e_3^{n_3}.\]
This can be seen by comparing the number of cycles of the above form with $\beta_2(S_n(\Theta_3))$ from formula \ref{eq:h2k2p}. In both cases the answer is the number of distributions of $n-2$ particles among edges $e_1,\ e_2,\ e_3$ (the problem of distributing $n-2$ indistinguishable balls into $3$ distinguishable bins) which is ${{n-2}\choose{2}}=\frac{1}{2}(n-2)(n-3)$.
\end{example}
\noindent From now on, we denote the $Y$-cycles as
\begin{gather}\label{eq:y-cycle}
c_{ijk}=e_i(h_j-h_k)+e_j(h_k-h_i)+e_k(h_i-h_j),\ i<j<k, \\ \nonumber
c'_{ijk}=e_i(h'_j-h'_k)+e_j(h'_k-h'_i)+e_k(h'_i-h'_j),\ i<j<k.
\end{gather}
Cycle $c_{ijk}$ is the $Y$-cycle of the $Y$-subgraph, whose hub vertex is $v$ and which is spanned on edges $e_i,e_j,e_k$. Cycle $c'_{ijk}$ corresponds to an analogous $Y$-subgraph, whose hub is $v'$.
\begin{example}{\bf The generator of $H_2(S_3(\Theta_4))$.}
Formula (\ref{eq:h2k2p}) tells us that $\beta_2(C_3(K_{2,4}))=1$. The corresponding generator in $S_3(\Theta_4)$ has the following form.
\begin{equation*}
c_{\Theta}=-(h_1-h_2)c'_{134}+(h_1-h_3)c'_{124}-(h_1-h_4)c'_{123}.
\end{equation*}
By expanding the $Y$-cycles, one can see that the above chain is a combination of all $2$-cells of $S_3(\Theta_4)$, hence, $C_n(K_{2,4})$ has the homotopy type of a closed $2$-dimensional surface. Its Euler characteristic is equal to $-4$, hence this is a surface of genus $3$. By the classification theorem of surfaces \cite{GX13}, we identify $C_n(K_{2,4})$ to have the homotopy type of a triple torus (fig. \ref{fig:triple-torus}).
 \begin{figure}[ht]
\centering
\includegraphics[width=0.4\textwidth]{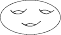}
\caption{A triple torus.}
\label{fig:triple-torus}
\end{figure}
\end{example}
\noindent From now on, we denote the $\Theta$-cycles as
\begin{equation}\label{eq:theta-cycle}
c_{ijkl}=-(h_i-h_j)c'_{ikl}+(h_i-h_k)c'_{ijl}-(h_i-h_l)c'_{ijk},\ i<j<k<l.
\end{equation}
Cycle $c_{ijkl}$ involves cells from $S_3(\Theta_4)$ for $\Theta_4$ being the subgraph of $\Theta_p$ spanned on edges $e_i,e_j,e_k,e_l$. Using the notation set in equations (\ref{eq:theta-cycle}) and (\ref{eq:y-cycle}), we propose the following generators of $H_2(S_n(\Theta_p))$.
\begin{gather*}
c_{ijk}c'_{rst}e_1^{n_1}\dots e_p^{n_p},\ i<j<k,\ r<s<t,\ n_1+\dots+n_p=n-4,\\
c_{ijkl}e_1^{n_1}\dots e_p^{n_p},\ i<j<k<l,\ n_1+\dots+n_p=n-3.
\end{gather*}
Let us start with $n=3$. The key to describe the relations between the $\Theta$-cycles spanned on different $\Theta_4$ subgraphs is to consider graph $\Theta_5$.
\begin{prop}
The $\Theta$-cycles in graph $\Theta_5$ satisfy the following relation
\begin{equation}\label{theta-rel}
c_{1234}-c_{1235}+c_{1245}-c_{1345}+c_{2345}=0.
\end{equation}
\end{prop} 
\noindent In graph $\Theta_p$, many relations of the form (\ref{theta-rel}) can be written down by choosing different $\Theta_5$ subgraphs. The linearly independent ones are picked by choosing the corresponding $\Theta_5$-subgraphs that are spanned on edge $e_1$ and some other four edges of $\Theta_p$. Such a choice can be made in ${p-1}\choose{4}$ ways. Subtracting the number of linearly independent relations from the number of all $\Theta_4$ subgraphs, we get
\[\beta_2(C_3(K_{2,p}))={{p}\choose{4}}-{{p-1}\choose{4}}={{p-1}\choose{3}}.\]
Increasing the number of particles to $n=4$ introduces products of $Y$-cycles and new relations. First of all, by proposition \ref{theta-dist} different distributions of additional particles in the $\Theta$-cycle can be realised are combinations of different products of $Y$-cycles.
\begin{prop}\label{theta-dist}
In graph $\Theta_4$, we have the following relations
\begin{gather*}
(e_1-e_2)c_{1234}=c_{124}c'_{123'}-c_{123}c'_{124},\\
(e_1-e_3)c_{1234}=c_{123}c'_{134'}-c_{134}c'_{123},\\
(e_1-e_4)c_{1234}=c_{124}c'_{134'}-c_{134}c'_{124}.
\end{gather*}
\end{prop}
\noindent Hence, all $\Theta$-cycles generate a subgroup of $H_2(S_n(\Theta_p))$ which is isomorphic to $\ZZ^{{p-1}\choose{3}}$. The last type of relations we have to account for\footnote{We do not mention here the typical relations between different $Y$-cycles on $Y$-subgraphs of the $S_p$ graphs which are met while computing the first homology group of the configuration spaces of star graphs (see \cite{HKRS}). Such relations are also inherited by the products of $Y$-cycles.} are the new relations between products of $Y$-cycles.
\begin{prop}\label{prod-rel}
In graph $\Theta_5$, products of $Y$-cycles satisfy
\[c_{123}c'_{145}+c_{145}c'_{123}+c_{125}c'_{134}+c_{134}c'_{125}-(c_{124}c'_{135}+c_{135}c'_{124})=0.\]
\end{prop}
\noindent Again, many relations of type \ref{prod-rel} can be written by picking different $\Theta_5$ subgraphs. Similarly as in the case of relations \ref{theta-rel}, the linearly independent ones are chosen by fixing $e_1$ to be the common edge of the $\Theta_5$ subgraphs. Hence, the number of linearly independent relations is ${p-1}\choose{4}$. In particular, we have
\[\beta_2(C_4(K_{2,p}))={{p-1}\choose{3}}+\left(\beta_1(C_2(S_p))\right)^2-{{p-1}\choose{4}},\]
where $\left(\beta_1(C_2(S_p))\right)^2$ is the number of independent product cycles after taking into the account the relations within the two opposite star subgraphs. All the above relations are inherited by the cycles in $S_n(\Theta_p)$ after multiplying them by a suitable polynomial in the edges of $\Theta_p$. In this way, they yield equation (\ref{eq:h2k2p}).

\subsection{When is $H_2(C_n(\Gamma))$ generated only by product cycles?}\label{sec:2nd-product}
In this section we prove the following theorem.
\begin{theorem}\label{thm:2nd-hom}
Let $\Gamma$ be a simple graph, for which $|\{v\in V(\Gamma):\ d(v)>3\}|=1$. Then group $H_2(C_n(\Gamma))$ is generated by product cycles.
\end{theorem}
\noindent In the proof we use the \'Swi\k{a}tkowski discrete model. The strategy of the proof is to first consider the blowup of the vertex of degree greater than $3$ and prove theorem \ref{thm:2nd-hom} for graphs, whose all vertices have degree at most $3$. For such a graph, we choose a spanning tree $T\subset \Gamma$. Next, we subdivide once each edge from $E(\Gamma)-E(T)$. We prove the theorem inductively by showing in lemma \ref{lemma:deg3induction} that the blowup at an extra vertex of degree $2$ does not create any non-product generators. The base case of induction is obtained by doing the blowup at every vertex of degree $2$ in $\Gamma-T$.  This way, we obtain graph which is isomorphic to tree $T$ and we use the fact that for tree graphs the homology groups of $S_n(T)$ are generated by products of $Y$-cycles.
\begin{lemma}\label{lemma:deg3induction}
Let $\Gamma$ be a simple graph, whose all vertices have degree at most $3$. Let $T$ be a spanning tree of $\Gamma$. Let $v\in V(\Gamma)$ be a vertex of degree $2$ and $\Gamma_v$ the graph obtained from $\Gamma$ by the vertex blowup at $v$. If $H_2(S_n(\Gamma_v))$ is generated by product cycles, then $H_2(S_n(\Gamma))$ is also generated by product cycles.
\end{lemma}
\begin{proof}
Long exact sequence corresponding to the vertex blow-up reads
\begin{gather*}
\dots\xrightarrow{\Psi_{n,3}} H_2\left(S_{n-1}(\Gamma_v)\right)\xrightarrow{\delta_{n,2}}H_2\left(S_n(\Gamma_v)\right)\xrightarrow{\Phi_{n,2}}H_2\left(\tilde S^v_n(\Gamma)\right)\xrightarrow{\Psi_{n,2}} \\ \nonumber \xrightarrow{\Psi_{n,2}}H_{1}\left(S_{n-1}(\Gamma_v)\right)\xrightarrow{\delta_{n,1}}H_{1}\left(S_n(\Gamma_v)\right)\xrightarrow{\Phi_{n,1}}\dots.
\end{gather*}
We aim to show that the corresponding long exact sequence
\[0\rightarrow \coker\left(\delta_{n,2}\right)\xrightarrow{}H_2\left(\tilde S^v_n(\Gamma)\right)\xrightarrow{}\ker\left(\delta_{n,1}\right)\rightarrow 0\]
splits. To this end, we construct a homomorphism $f:\ \ker\left(\delta_{n,1}\right)\rightarrow H_2\left(\tilde S^v_n(\Gamma)\right)$ such that $\Psi_{n,2}\circ f=id_{\ker(\delta_{n,1})}$. In the construction we use the explicit knowledge of elements of $\ker\left(\delta_{n,1}\right)$. Such a knowledge is accessible, as we know the generating set of $H_{1}\left(S_{n-1}(\Gamma_v)\right)$ - because all vertices of $\Gamma_v$ have degree at most $3$, it consists of $Y$-cycles and $O$-cycles, subject to the $\Theta$-relations (equations (\ref{theta3-rel}) and (\ref{ab-y-rel})) and the distribution of free particles which say that $[ce]=[cv]$ whenever $v$ is a vertex of $e$. Recall that cycle $c$ represents element of $\ker\left(\delta_{n,1}\right)$ whenever $[ce]=[ce']$, where $e$ and $e'$ are the edges incident to vertex $v$. This happens if and only if cycles $ce$ and $ce'$ are related by a $\Theta$-relation or a particle-distribution relation. However, because all vertices of $\Gamma$ have degree at most $3$, it is not possible to write the $\Theta$ relations in the form $ce-ce'=\partial(b)$ for any $c$. Hence, cycles $ce$ and $ce'$ must be related by the particle distribution, i.e. there exists a path in $\Gamma_v$ which is disjoint with ${\rm Supp}(c)$ and which joins edges $e$ and $e'$. The desired homomorphism $f$ is constructed as follows. For a generator $c$ of $H_{1}\left(S_{n-1}(\Gamma_v)\right)$, find path $p(c)$ which joins $e$ and $e'$ and is disjoint with ${\rm Supp}(c)$. Having found such a path, we complete it to a cycle $O_{p(c)}$ in a unique way by adding to $p$ vertex $v$ and edges $e,e'$. From cycle $O_{p(c)}$ we form the $O$-cycle $c_{O_{p(c)}}$ (see definition \ref{definition:o-cycle}). Homomorphism $f$ is established after choosing the set of independent generating cycles and paths that are disjoint with them. It acts as $f:\ [c]\mapsto [c\otimes c_{O_{p(c)}}]$. Clearly, we have $\Psi_{n,2}([c\otimes c_{O_{p(c)}}])=[c]$ by extracting from $c_{O_{p(c)}}$ the part which contains half-edges incident to $v$. 

This way, we obtained that $H_2\left(\tilde S^v_n(\Gamma)\right)\cong \ker\left(\delta_{n,1}\right)\oplus\coker\left(\delta_{n,2}\right)$ and that elements of $\ker\left(\delta_{n,1}\right)$ are represented by product $c_O\otimes c_Y$ cycles. By the inductive hypopaper, elements of $\coker\left(\delta_{n,2}\right)$ are the product cycles that generate $H_2\left(S_{n-1}(\Gamma_v)\right)$ subject to relations $ce\sim ce'$.
\end{proof}

The last step needed for the proof of theorem \ref{thm:2nd-hom} is showing that the blowup of $\Gamma$ at the unique vertex of degree greater than $3$ does not create any non-product cycles. Here we only sketch the proof of this fact which is analogous to the proof of lemma \ref{lemma:deg3induction}. Namely, using the knowledge of relations between the generators of $H_{1}\left(S_{n-1}(\Gamma_v)\right)$, one can show that the elements of $\ker\left(\delta_{n,1}\right)$ are of two types: i) the ones that are of the form $\partial(c\otimes b_{p(c)})$, where $[c]\in H_{1}\left(S_{n-1}(\Gamma_v)\right)$ and $b_{p(c)}$ is the $1$-cycle corresponding to path $p(c)\subset \Gamma_v$ which is disjoint with ${\rm Supp}(c)$ and whose boundary are edges incident to $v$, ii) pairs of cycles of the form $(c(e_j-e_0),c(e_0-e_j))$, where $e_0, e_i, e_j$ are edges incident to $v$ and $[c]\in H_{1}\left(S_{n-2}(\Gamma_v)\right)$. Such pairs are mapped by $\delta_{n,1}$ to $c\otimes\left((e_j-e_0)(e_0-e_i)+(e_0-e_i)(e_0-e_j)\right)$ which is equal to $\partial(c\otimes c_{0ij})$, where $c_{0ij}$ is the $Y$-cycle corresponding to the $Y$-graph in $\Gamma$ centred at $v$ and spanned by edges $e_0,e_i,e_j$. Next, in order to show splitting of the homological short exact sequences, we consider a homomorphism $f:\ \ker\left(\delta_{n,1}\right)\rightarrow H_2\left(\tilde S^v_n(\Gamma)\right)$, for which $\Psi_{n,2}\circ f=id_{\ker(\delta_{n,1})}$. Such a homomorphism maps $[c]$ to $[c\otimes c_{O_{p(c)}}]$, where $O_{p(c)}$ is the cycle which contains path $p(c)$ and vertex $v$. Pairs $([c(e_j-e_0)],[c(e_0-e_j)])$ are mapped by $f$ to cycles $c\otimes c_{0ij}$. We obtain that $H_2\left(\tilde S^v_n(\Gamma)\right)\cong \ker\left(\delta_{n,1}\right)\oplus\coker\left(\delta_{n,2}\right)$, where the generators of $\ker\left(\delta_{n,1}\right)$ are in a one-to-one correspondence with the product homology classes of $H_2\left(\tilde S^v_n(\Gamma)\right)$ described above. Elements of $\coker\left(\delta_{n,2}\right)$ are also represented by product cycles. These cycles are the generators of $H_2\left(S_{n}(\Gamma_v)\right)$ subject to relations $ce_0\sim ce_i$, $i=1,\dots,d(v)$, where $e_0,e_1,\dots,e_{d(v)}$ are edges incident to $v$.

The task of characterising all graphs, for which $H_2(S_n(\Gamma))$ is generated by product cycles requires taking into account the existence of non-product generators from section \ref{sec:K_2p}. As we show in section \ref{sec:K_2p} the existence of pairs of vertices of degree greater than $3$ in the graph implies that there may appear some multiple tori in the generating set of $H_2(C_n(\Gamma))$ stemming from subgraphs isomorphic to graph $K_{2,4}$. Furthermore, the class of graphs, for which higher homologies of $C_n(\Gamma)$ are generated by product cycles is even smaller. Recall graph $K_{3,3}$ whose all vertices have degree $3$, but $H_3(C_n(K_{3,3}))$ has one generator which is not a product of $1$-cycles (see section \ref{sec:k33}). 

\section{Summary}
In the first part of this paper, we explained that quantum statistics on a topological space $X$ are classified by conjugacy classes of unitary representations of the fundamental group of the configuration space $C_n(X)$. Conversely, every unitary representation of the graph braid group gives rise to a flat complex vector bundle over space $C_n(X)$. We interpret different isomorphism classes of flat complex vector bundles over $C_n(X)$ as fundamentally different families of particles. Among these families we find for example bosons, corresponding to the trivial flat bundle, and fermions that may correspond to a non-trivial flat bundle. Interestingly, there also exist intermediate possibilities called anyons who can live on a trivial as well as on a non-trivial bundle. The existence of more than two isomorphism classes is {\it a priori} possible. However interesting and desirable, an explicit construction of non trivial flat bundles for configuration spaces of $X=\RR^2$ or $X=\RR^3$ is difficult, hence some simplified mathematical models are needed. This motivates the study of configuration spaces of particles on graphs which are computationally more tractable. Topological invariants that give a coarse grained picture of the structure of the set of isomorphism classes of flat complex vector bundles over $C_n(X)$ are the homology groups of configuration spaces. In particular, we point out the important role of Chern characteristic classes that map the flat vector bundles to torsion components of the homology groups of $C_n(X)$ with coefficients in $\ZZ$. In the second part of this paper, we compute homology groups of configuration spaces of certain families of graphs. We summarise the computational results as follows.
\begin{itemize}
\item Configuration spaces of tree graphs, wheel graphs and complete bipartite graphs $K_{2,p}$ have no torsion in their homology. This means that the set of flat bundles over configuration spaces of such graphs has a simplified structure, namely every flat vector bundle is stably equivalent to a trivial vector bundle. Hence, these families of graphs are good first candidates for a class of simplified models for studying the properties of non-abelian statistics.
\item Computation of the homology groups of configuration spaces of some small canonical graphs via the discrete Morse theory shows that in some cases there is a $\ZZ_2$-torsion in the homology. However, we were not able to provide an example of a graph which has a torsion component different than $\ZZ_2$ in the homology of its configuration space.
\item It is a difficult task to accomplish a full description of the homology groups of graph configuration spaces using methods presented in this work. One fundamental obstacle is that such a task requires the knowledge of possible embeddings of $d$-dimensional surfaces in $C_n(\Gamma)$ which generate the homology. However, cycles generating the homology in dimension $2$ of graph configuration spaces have the homotopy type of tori or multiple tori. This fact allowed us to find all generators of the second homology group of configuration spaces of a large family of graphs in section \ref{sec:2nd-product}.
\end{itemize}

Let us next summarise the relevance of distinguishing between trivial and non-trivial bundles. 
\begin{itemize}
\item {\bf Trivial bundles}. They are relevant in the context of quantum computing where one is interested mainly in universality of unitary representations of braid groups and the dimension of the representations grow exponentially with the number of particles. In that context, the fact that one can have different isomorphism classes of vector bundles does not seem to play a significant role. In fact, it is even better not to have many isomorphism classes. If we know that there is just one isomorphism class (all bundles are isomorphic to the trivial bundle) then all representations of the braid group are related to each other via the isomorphism of the corresponding bundles and the problem of classifying them should become more tractable. As we show in this paper, this happens when one considers high-dimensional representations (stable range) of graph braid groups where there is no torsion in the homology of $C_n(\Gamma)$.
\item {\bf Non-trivial bundles}. They become relevant in situations where the rank of the bundle is not too high, i.e. if the considered bundles are not in the stable range. The corresponding representations of braid groups appear in the effective models of non-abelian Chern-Simons particles which are point-like sources mutually interacting via a topological non-Abelian Aharonov-Bohm effect. A model of such particles constrained to move on graphs would be constructed by defining a separate Chern-Simons hamiltonian for each cell of the closure of $C_n(\Gamma)$ viewed as a subset of $\Gamma^{\times n}$. The non-abelian braiding would show up as proper gluing conditions for the wave-functions on the boundaries of cells from $C_n(\Gamma)$ while studying self-adjoint extensions of such a hamiltonian. The moduli space of flat $U(n)$ bundles over $C_n(\Gamma)$ is the space of possible parameters that appear as the gluing conditions (see e.g. \cite{BE92}). This area is still quite unexplored and some more progress has to be made to see how this theory works explicitly for concrete graphs.
\end{itemize}

Finally, the {\it a priori} large number of possible unitary representations of braid groups can be cut down by taking into account the anyonic fusion rules \cite{honeycomb}, i.e. by applying the framework of modular tensor categories. Unitary representations of the braid group that are described by modular tensor categories, are defined by specifying the fusion rules, the so-called $F$-matrix that ensures associativity of fusion and the $R$-matrix that describes braiding of pairs of particles \cite{honeycomb}. Ocneanu rigidity \cite{rigidity} asserts that for anyons on the plane, there is only a finite number of representations of the braid group that arise from the above construction. However, graphs in principle provide more freedom, as pairs of particles may braid differently in different parts of the graph.There are two main differences between braiding particles on the plane and braiding particles on a graph: i) the generating braids are more complicated than just the ones that correspond to braiding of pairs of neighbouring particles, ii) the variety of relations is much bigger than in the case of particles on the plane. In other words, one can have more than just one $R$-matrix for particles on a graph. This property of graph braid groups can be seen already on the level of its abelian representations \cite{HKRS}. Namely, consider a $2$-connected graph which consists of a number of $3$-connected components that are connected to each other (Fig. \ref{fig:abelian}). The theory of abelian representations of graph braid groups \cite{HKRS} tells us that a pair of particles can exchange as bosons or fermions, depending in which component of the graph the exchange takes place. An analogous situation will take place in the case of non-abelian representations -- one can assign different $R$-matrices to different components of the graph.
 \begin{figure}[H]
\centering
\includegraphics[width=0.7\textwidth]{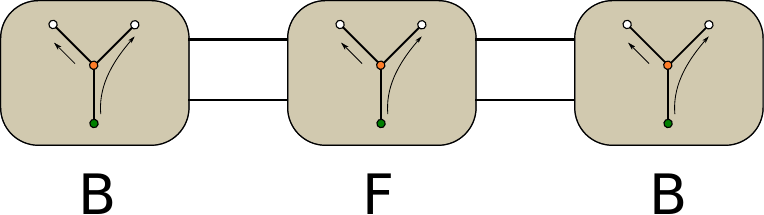}
\caption{A graph that consists of three $3$-connected components (depicted as boxes). The schematically pictured abelian representation of $Br_2(\Gamma)$ is such that exchange of particles in each of the components results with a fermionic (F) or bosonic (B) phase factor. The situation is more complicated for non-abelian representations, but the general characteristic survives - one can choose different $R$-matrices for each of the components.}
\label{fig:abelian}
\end{figure}
One expects that the fusion rules will nevertheless significantly restrict the number of admissible representations of graph braid groups. 

\begin{acknowledgements} 
TM acknowledges the financial support of the National Science Centre of Poland -- grants {\it Etiuda} no. $2017/24/T/ST1/00489$ and {\it Preludium} no. $2016/23/N/ST1/03209$. AS was supported by the National Science Centre of Poland grant {\it Sonata Bis} no. $2015/18/E/ST1/00200$.
\end{acknowledgements}

\end{document}